\documentclass[conference]{IEEEtran}

\usepackage{booktabs} 
\usepackage{etoolbox}
\usepackage{balance}  
\usepackage{amsthm}

\usepackage{comment}

\makeatletter
\def\@copyrightspace{}
\makeatother

\makeatletter
\def\@mkbibcitation{}
\makeatother

\makeatletter
\def\BState{\State\hskip-\ALG@thistlm}
\makeatother

\usepackage[linesnumbered,ruled]{algorithm}
\usepackage[noend]{algpseudocode}
\usepackage[utf8]{inputenc}

\usepackage[inline]{enumitem}

\usepackage{multirow}
\usepackage{rotating}
\usepackage{mathtools}

\usepackage{pgfplots}

\makeatletter

\def\addlegendimage{\pgfplots@addlegendimage}
\makeatother

\usepackage{color,soul}
\usepackage{todonotes}

\graphicspath{{Images/}}
\graphicspath{{Results/}}
\DeclareGraphicsExtensions{.pdf,.jpg,.png}
\usepackage{wrapfig}
\usepackage{url}
\usepackage{subfig}

\usepackage{amssymb}
\usepackage{amsmath}
\usepackage{array,tabularx}
\usepackage{tabstackengine}
\usepackage{cancel}

\usepackage{fancyhdr}
\usepackage{listings}
\usepackage{enumitem}
\usepackage{dsfont}

\definecolor{ao(english)}{rgb}{0.0, 0.5, 0.0}
\definecolor{bananayellow}{rgb}{1.0, 0.88, 0.21}
\definecolor{amber}{rgb}{1.0, 0.75, 0.0}

\newcommand{\yat}[1]{[[\emph{\color{red}Yasser: #1}]]}
\newcommand{\dg}[1]{[[\emph{\color{ao(english)}Dhrub: #1}]]}

\newtheorem{theorem}{Theorem}
\newtheorem{lemma}{Lemma}

\newcounter{definitionctr}
\newenvironment{definition}{
   \medskip\noindent         \refstepcounter{definitionctr}
   \textbf{Definition \thedefinitionctr .}
   }{}  
\numberwithin{definitionctr}{section}

\newcounter{examplectr}
\newenvironment{example}{
   \smallskip \noindent         \refstepcounter{examplectr}
   \textbf{Example \theexamplectr .}
   }{}  
\numberwithin{examplectr}{section}

\pgfplotsset{compat=1.14}

\begin{document}


\title{PIQUE: Progressive Integrated QUery Operator with Pay-As-You-Go Enrichment}

\author{
			Dhrubajyoti Ghosh,
			Roberto Yus,
			Yasser Altowim,
			Sharad Mehrotra \\
			University of California, Irvine, King Abdulaziz City for Science and Technology \\
			\texttt{\{dhrubajg,ryuspeir\}@uci.edu}, \texttt{yaltowim@kacst.edu.sa},
			\texttt{sharad@ics.uci.edu}  \\
		}

\maketitle

\begin{abstract}

Big data today in the form of text, images, video, and sensor data needs to be enriched (i.e., annotated with tags) prior to be effectively queried or analyzed. Data enrichment (that, depending upon the application could be compiled code, declarative queries, or expensive machine learning and/or signal processing techniques) often cannot be performed in its entirety as a pre-processing step at the time of data ingestion. Enriching data as a separate offline step after ingestion makes it unavailable for analysis during the period between the ingestion and enrichment. To bridge such a gap, this paper explores a novel approach that supports progressive data enrichment  during query processing in order to support interactive exploratory analysis.
Our approach is based on integrating an operator, entitled PIQUE, to support a prioritized execution of the enrichment functions during query processing. Query processing with the PIQUE operator significantly outperforms the baselines in terms of rate at which answer quality improves during query processing.

\end{abstract}

\section{Introduction}
\label{sect:introduction}

Today, organizations have access to potentially limitless information in the form of web data repositories, continuously generated sensory data, social media posts, captured audio/video data, click stream data from web portals, and so on~\cite{OnlineLiveStat}. Before such data can be used, it often needs to be enriched (i.e., associated with tags) using appropriate machine learning, signal processing techniques, complex queries, and/or compiled code. Examples of such enrichment include sentiment analysis extraction over social media~\cite{TweetSentimentAnalysis3}, named entity extraction in text~\cite{NamedEntity_Extraction1}, face detection in images~\cite{face_detection_2}, and sensor interpretation and fusion over sensory inputs~\cite{sensor_fusion_1}. 




Traditionally data enrichment is treated as a separate offline process performed in the back-end. Recently, there is an inherent drive towards real time analysis, making data enrichment an integral part of online data processing. Several industrial systems (e.g., Storm~\cite{storm}, Spark Streaming~\cite{SparkStreaming}, Flink~\cite{Flink}, Kafka Stream~\cite{KafkaStream}) were introduced that support data enrichment at the ingestion time. Recent work~\cite{IDEA_ingestion_framework} has explored an ingestion framework that aims at optimizing enrichment by batching such operations. While limited data enrichment (e.g., by running cheap functions) at the ingestion time is typically feasible, running a suite of computationally expensive functions to generate the best possible tags is often impossible due to the speed and volume of the incoming data which will make data ingestion a bottleneck.

One such application where ingestion can become a bottleneck is that of a smart building wherein a diverse set of sensors are deployed to track the movement of people inside it\footnote{This sample application is based on our experience in creating a smart IoT test bed at Donald Bren Hall of UCI~\cite{tippers}.}. Analysis of WiFi signal might provide a cheap (but relatively coarse) localization of people whereas face detection/recognition over images captured by hundreds of cameras within the building may offer an accurate, albeit more expensive, alternative. Exhaustive analysis of continuously captured video is likely computationally infeasible and/or wasteful.


Another application is a tweet monitoring system that tracks events in real-time over social media~\cite{tweet_topic_follow_up_demo}. Different enrichment functions may be used to determine the tweet topics, references to locations in the text (if any), the sentiment/tone, etc. Although cheap variants of enrichment functions for such tasks may exist, enrichment functions achieving a high-quality labelling may incur significant computational cost. A user may be interested in analyzing social media posts (e.g., determining the geographical regions where the sentiment of tweets for a specific topic/event is mostly positive). Such analysis will require executing several enrichment functions such as topic determination, location determination, and sentiment analysis.



Motivated by such scenarios, we develop a mechanism to enrich data {\em progressively} at the time of query processing. Our mechanism is based on implementing a PIQUE ``{\it Progressive Integrated Query Operator  with Pay-as-you-go Enrichment}'' operator that integrated with databases, supports progressive query answering in a way that maximizes the rate at which the quality of the query answer improves $w.r.t$ time. Our implementation of PIQUE operator assumes a setup wherein some enrichment may have been applied on the data at the ingestion time but data still needs to be further enriched at the query time to meet the quality requirement of end-applications. A query containing PIQUE operator executes in epochs wherein a PIQUE operator chooses which objects to enrich and to what degree (i.e., using which enrichment functions). To determine such object function pairs, PIQUE prioritizes/ranks objects and enrichment functions to improve the quality of the answer as quickly as possible. When a PIQUE operator is called, it first analyzes the evaluation progress to generate an execution plan for the current epoch that has highest potential of improving the answer quality in that epoch. While PIQUE can be integrated into any database engine to support progressive computation, it offers the most benefit if the underlying query processing engine is made PIQUE-aware. In such a system, PIQUE can collaboratively filter objects based on other deterministic/precise attributes and restrict its selection of objects to only those objects whose enrichment can influence the quality of the query answer.

\vspace{0.4cm}
\noindent In summary, our contributions in this paper are as follows:

\begin{itemize}[leftmargin=*]


\item We propose a progressive approach to data enrichment and query evaluation entitled PIQUE that quantizes the execution time into epochs that are adaptively enriched. 
   
\item We present an algorithm for the problem of generating an execution plan that has the highest potential of improving the quality of the answer set in the corresponding epoch.
       
\item We develop an efficient probabilistic strategy for estimating the benefit of applying various enrichment functions on different objects.

    
\item We experimentally evaluate PIQUE in different domains (i.e., images and tweets) using real datasets and enrichment functions and demonstrate the efficiency of our solution.


\end{itemize}

\section{Preliminaries}
\label{sect:problemDefinition}

In this section, we describe the high level view of the  PIQUE operator and its semantics. We also specify 
the requirements of the query processing engine to support PIQUE. PIQUE operator is defined for datasets where objects can be associated with one (or more) tags each of which is computed using enrichment functions associated with the tag.

\vspace{0.1cm}
\noindent
\textbf{Datasets.} Let $O = \{ o_1, o_2, \dots, o_{|O|}\}$ be a dataset of objects that have objects of different types such as images, tweets, etc. Each object has a certain number of \emph{precise} attributes and can be associated with a number of \emph{tags} of different \emph{tag types}. Each tag type is denoted by $T_i$. Table~\ref{table:runningExampleImage} shows an example with a dataset of six images. In this dataset, there are two precise attributes, {\em Timestamp} and {\em Location}, and two types of tags are associated with each objects: $T_1 = $ {\em Person} and $T_2 = $ {\em Expression}. (Note that the values in brackets in the table for the tags correspond to the ground truth values.)

\setlength\tabcolsep{1.5pt} 
\begin{table}[t]
\small
\centering
\begin{center}
 \begin{tabular}{|
 >{\centering\arraybackslash}p{10mm}
|>{\centering\arraybackslash}p{13mm}|>{\centering\arraybackslash}p{14mm}|>{\centering\arraybackslash}p{16mm}|>{\centering\arraybackslash}p{27mm}|} 
 \hline
 \textbf{Object} & \textbf{Location} & \textbf{Person} & \textbf{Expression} & \textbf{Timestamp} \\ 
 \hline
 $o_1$ & 2065 & (John) & (Smile)  & 2019-05-15 15:48:00 \\ 
 \hline
 $o_{2}$ & 2082 & (David) & (Smile)  & 2019-05-15 15:52:00 \\
 \hline
 $o_3$ & 2088 & (Jim) & (Neutral) & 2019-05-15 15:54:00 \\
 \hline
 $o_{4}$ & 2090 & (Harry) & (Neutral)  & 2019-05-16 11:08:00 \\
 \hline
 $o_{5}$ & 2065 & (David)  & (Smile)  & 2019-05-16 11:10:30 \\  
 \hline
 $o_{6}$ & 2086 & (John) & (Frown)  & 2019-05-16 11:12:00 \\
 \hline
\end{tabular}
\caption{Running Example.}\label{table:runningExampleImage}
\vspace{-1em}
\label{table:1}
\end{center}
\end{table}
\setlength\tabcolsep{6pt} 

\setlength\tabcolsep{1.5pt} 
\begin{table}[t]
\small
\centering
\begin{center}
  \begin{tabular}{|c|c|}
 \hline
  \textbf{Notation} & \textbf{Definition} \\ 
 \hline
  $O$  & Dataset of objects \\ 
 \hline
 $T_{i}$ & Tag type \\ 
  \hline
 $t_{l} \in T_{i}$ & Tag $t_l$ of tag type $T_i$ \\
 \hline
 $F_i$ & Set of all enrichment functions of tag type $T_i$\\
 \hline
 $f^i_j$ & The $j$-th enrichment function of tag type $T_i$  \\
 \hline
 $q^i_j$ & Quality of enrichment function $f^i_j$  \\   
 \hline
  $c^i_j$ & Cost of enrichment function $f^i_j$  \\  
  \hline
  $Q$ & Query  \\ 
  \hline 
  $exp$ & An expression in PIQUE \\
 \hline
  $R^i_l$ & A predicate with tag $t_l \in T_i$ \\ 
  \hline
  $\mathds{R}$ & Set of predicates in a PIQUE Expression \\
 \hline
  $p_k$ & Predicate probability of $o_k$\\
 \hline
 
 $M_i$ & Combine function of
 			tag type $T_i$ \\  
 \hline
 $ s_k$ & State of $o_k$\\  
 \hline
 $\mathcal P_k$ & Expression satisfiability probability
 					of $o_k$\\  
 \hline
 $h_k$ & Uncertainty value of $o_k$ \\  
 \hline
  $\mathsf{A}_w$ & Answer set of epoch $w$ \\ 
  \hline
  $F_\alpha(\mathsf{A}_w)$ & $F_\alpha$ measure of answer set in epoch $w$ \\
  \hline
  $E(F_\alpha(\mathsf{A}_w))$ & Expected $F_\alpha$ measure of answer set in epoch $w$ \\
  \hline
  $\mathcal P^{\tau}_{w}$ & The threshold probability value of epoch $w$\\ 
  \hline
  $\mathsf{CS}_w$ & Candidate set chosen in epoch $w$ \\ 
 \hline
 $\mathsf{TS}_w$ & Set of triples generated in epoch $w$ \\ 
 \hline
 $\mathsf{EP}_w$ & Execution plan generated in epoch $w$ \\ 
 \hline
 \end{tabular}
\caption{Frequently used notations.}
\label{table:notation_table}
\end{center}
\end{table}

\vspace{0.1cm}
\noindent
\textbf{Enrichment Functions.} Let $F_i = \{f^i_1, f^i_2, \dots, f^i_k\}$ be a set of enrichment functions that takes as input an object (i.e., image or tweet) and a tag of a particular tag type and outputs the probability value of the object having that tag. E.g., a function $f^i_j$ takes as input an image object $o_k$ and a tag {\tt Person = John}, and returns the probability of the object containing the tag. Such an output is denoted by the notation of $f^i_j(o_k,t_l)$. Given the dataset in Table~\ref{table:runningExampleImage}, the tag {\tt Person = John} can be evaluated by performing face recognition using multiple functions such as functions based on a Decision Tree (DT) classifier, a Random Forest (RF) classifier, a Neural Network (NN) classifier, and an ensemble of the above respectively. Each of these functions takes as input an image and a tag (e.g., {\tt Person = John}) and outputs the probability of the object containing that tag (e.g., $0.8$).


Each function $f^i_j$ is associated with a quality and cost value, which we denote as $q^i_j$ and $c^i_j$, respectively. The quality of a function measures the accuracy of the function in detecting the tags. For example, the quality of a function that is internally implemented using a machine learning classifier is typically measured using the area under the curve (AUC) score of the classifier~\cite{BRADLEY19971145}. The cost of a function represents the average cost of running the function on an object. Our approach is agnostic to the way the quality and cost values are set. In Section~\ref{sect:experimental setup}, we show how these values can be determined either as a pre-processing step or during the query evaluation.

For each of the tag types $T_i$, one enrichment function $f^i_j $ is identified as the initial {\em seed} function. 
Such a function is expected to be of low computation cost and can thus be executed on all the objects at the ingestion time or as a pre-processing step prior to the query evaluation. In the above example, a decision tree classifier, that is 10x faster than the RF classifier and 40x faster than the NN classifier, may serve as the seed function for the tag type {\em Person}\footnote{Depending on the face recognition algorithm (e.g., detecting patterns over a small number of pixels within a window or detecting them over a large window), a classifier cost could be as high as four seconds to forty seconds per image~\cite{face_recognition_deva}. Such expensive enrichment functions cannot be applied on the entire data either upon the ingestion time or as a pre-processing step. PIQUE, in contrast, can execute such functions on a need basis on a small number of objects during the query execution time and is, thus, able to achieve high quality results without having to execute such functions on all objects.}.



For clarity, we will present the paper as if only the seed enrichment functions were run on the objects prior to the arrival of the query. In Section~\ref{sect:experiments}, we discuss how PIQUE deals with the case where the outputs of other enrichment functions, run in previous queries, are cached and study the impact of different levels of caching on PIQUE's performance.

\vspace{0.2cm}
\noindent
\textbf{PIQUE Operator.} The PIQUE operator in 
queries is akin to  table-level functions specified in  database systems such as SQL server \cite{SQLServer}.
The syntax for a PIQUE operator is as follows: 
\begin{equation*}
    PIQUE(\langle object\ \ set \rangle, epoch, \langle expression \rangle )
\end{equation*}

An example of a PIQUE operator in a query is shown below. This query retrieves all the images of John smiling inside of rooms 2080, 2081, and 2082 between 2 p.m. and 6 p.m. on the $14^{\text{th}}$ of October, 2019. 

\begin{footnotesize}
 \begin{lstlisting}[mathescape = true,label={lst:QueryExampleImage4}]
  SELECT * 
  FROM PIQUE(ImageDataset, 4,(Person("John") AND Expression("Smile"))) as O1 
  WHERE O1.location IN ["2080","2081,"2082"] 
  AND O1.Timestamp >= "2019-10-14 14:00:00" AND O1.Timestamp <= "2019-10-14 18:00:00"
 \end{lstlisting}
\end{footnotesize}

The PIQUE operator, defined over a set of objects 
(in this example, the ImageDataset) 
chooses the appropriate subset of objects from the
object set and applies appropriate enrichment functions (associated with the corresponding tags) 
to generate the set of objects that satisfy the associated expression (in this example, {\tt Person = John AND Expression = Smile}). 
The selected object-enrichment function pairs
are chosen in a way such that they can be computed in $epoch$ duration
(i.e., $4$ units of time in the example). In general, several PIQUE operators can be associated with a given query (e.g., one for each object collection with tags). Although, we will assume that the epoch associated with each of the PIQUE operators in a query is the same.

\vspace{0.1cm}
\noindent
\textbf{Query Processing with PIQUE Operator.}\label{query definition} 
PIQUE operators can be incorporated into existing database systems simply as table level functions. Used this way, they will first execute on the table containing objects, select the set of object-enrichment function pairs, execute the corresponding enrichment functions, modify the tags associated with the objects appropriately, and then execute the query\footnote{We assume that the functions already executed for each object (formalized as {\em state} in Section ~\ref{sect:overview}), are encoded as a hidden attribute with the objects. PIQUE utilizes the state to choose appropriate enrichment functions. This ensures that subsequent invocation to PIQUE will select different enrichment functions in different epochs.}. However, an optimized PIQUE-aware query processor could optimize the execution further. For instance, by pushing the selection conditions in the WHERE clause of the query, prior to selecting objects to enrich, can bring substantial improvements, specifically if the predicates in the WHERE clause are very selective. 
Furthermore, PIQUE returns a set of answers at the end of an epoch. One could continuously call the corresponding query repeatedly to refine the result sets returned by PIQUE or, if the underlying query engine is PIQUE-aware (as assumed in the rest of the paper), the system automatically executes the query and refines answers in epochs\footnote{Answers returned by a query processor using PIQUE at the end of each epoch may invalidate some of the answers returned in previous epochs. For instance, if a classifier marks a person incorrectly as "John" in one iteration, it may choose to rectify the answer in a later epoch. Dealing with such changes across epochs is assumed to be dealt by the database engine. }.


\noindent
\textbf{Aside.} In the main body of the paper, for 
simplicity of exposition, we consider that the objects that needs to be enriched by PIQUE can fit in memory. In Section~\ref{sect:disk_based_approach}, we show how PIQUE supports the scenario where the objects do not fit in memory (i.e., disk-resident objects) and present the experimental results related to this scenario in Section~\ref{sect:experimental_result}. 

Also for  notational simplicity,  we will represent expression associated with PIQUE operator as a set of predicates connected by boolean connective operators: And ($\Lambda$) or Or ($\vee$) and refer to this set as $\mathds{R}$. (A predicate is used henceforth to refer to a predicate on tags.) A predicate, denoted as $R^i_l \in \mathds{R}$, consists of a tag type $T_i$, a tag $t_l$, and an operator, which can be either ``=" (equal) or ``!=" (not equal). An example of a predicate can be {\tt Person=  John}.

\section{Overview of PIQUE}
\label{sect:overview}

Given a dataset $O$, an $epoch$, and an expression $exp$ containing predicates that require enrichment by using a set of enrichment functions for each tag type, PIQUE's goal is to maximize the rate at which the quality of the answer set $A$ improves. We first define the quality metric for $A$, then we explain the progressive answer semantics based on the quality metric, and finally, explain the different high-level steps to evaluate the PIQUE operator.

\vspace{0.1cm}
\noindent
\textbf{Quality Metric.} 
Each query answer has a notion of quality associated with it. A quality metric represents how close the answer set $A$ is to the ground truth set $G$. In a set-based answer, the $F_\alpha$-measure is the most popular and widely used quality metric. It is the harmonic mean of precision and recall. Precision ($Pre$) is defined as the fraction of correct objects in $A$ to the total number of objects in $A$ whereas recall ($Rec$) is defined as the fraction of correct objects in $A$ to the total number of objects in $G$. More formally, the $F_\alpha$-measure is computed as follows:
\begin{equation}
F_\alpha (A)= \frac{(1 + \alpha) \cdot Pre (A) \cdot Rec(A)}{(\alpha \cdot Pre(A) + Rec(A))}
\end{equation}

\noindent
where $Pre(A) = |A \cap G|/|A| $, $ Rec(A) =|A \cap G|/|G|$, and $ \alpha \in [0,1]$ is the weight factor assigned to the precision value. 

\vspace{0.1cm}
\noindent
\textbf{Progressive Answer Semantics.} 
The quality of a progressive approach can be measured by the following discrete sampling function~\cite{progressive-duplicate-detection}: 
\begin{equation}\label{progressiveness-metric}
\vspace{-0.2cm}
\begin{split}
\vspace{-0.4cm}
& Qty(A) = \sum\limits_{i=1}^{|V|} W(v_i) \cdot Imp(v_i)
\end{split}
\end{equation}
where $V = \{ v_1, v_2, \dots, v_{|V|}\}$ is a set of sampled cost values (s.t. $v_i > v_{i-1}$),  $W$ is a weighting function that assigns a weight value $W(v_i) \in [0,1]$ to every cost value $v_i$ (s. t. $W(v_i) > W(v_{i-1})$), and $Imp(v_i)$ is the improvement in the $F_\alpha$ measure of $A$ that occurred in the interval $(v_{i-1}, v_i]$. (For convenience, we assume the existence of a cost value $v_0 = 0$ henceforth.) In other words, $Imp(v_i)$ is the $F_\alpha$ measure of $A$ at $v_i$ minus the $F_\alpha$ measure of $A$ at $v_{i-1}$\footnote{{Our approach can be used to optimize for the case where the goal is to generate the highest possible quality result given an evaluation cost budget $BG$. This is achieved by having $V = \{ v_1 = BG\}$, setting $W(v_1) = 1$, and configuring the approach to terminate when the budget $BG$ is consumed.}}.

PIQUE divides the execution time into multiple epochs. Each epoch consists of three phases: plan generation, plan execution, and answer set selection. A plan consists of a list of triples where each triple consists of an object $o_k$, a predicate $R^i_l$, and an enrichment function $f^i_m$. Identifying the list of triples in the plan depends on a trade-off between the cost of evaluating those triples and the expected improvement in the quality of the answer set that will result from evaluating them. Figure~\ref{fig:progressive algorithm} presents a high-level overview of the PIQUE. Given a PIQUE expression, in the following we explain how the predicate probability values of the objects are calculated $w.r.t.$ a predicate in the expression and then we explain each of the previously mentioned steps in details.

\vspace{0.1cm}
\noindent
\textbf{Predicate Probability.} 
Each object $o_k \in O$ is associated with a \emph{predicate probability} {\em w.r.t.} each predicate $R^i_l$, denoted as $p(o_k,R^i_l)$. This predicate probability is defined as the probability of the object $o_k$ satisfying the predicate $R^i_l$ and is computed based on the probability outputs of the enrichment functions in $F_i$ as follows:
\begin{equation}
p(o_k,R^i_l) = M_i(g^i_1(o_k, t_l), g^i_2(o_k,t_l) , \; \dots, \; g^i_{|F_i|}(o_k,t_l))
\end{equation}
\noindent
where $g^{i}_j(o_k, t_l) = f^i_j(o_k, t_l)$ if the operator in $R^i_l$ is ``=''. Otherwise, $g^i_j(o_k, t_l) = 1 - f^i_j(o_k, t_l)$. Function $M_i$ is a combine function that takes the outputs of the enrichment functions in $F_i$ and returns the predicate probability of $o_k$ satisfying predicate $R^i_l$. Our approach is agnostic to the way the function $M_i$ is implemented. In Section~\ref{sect:experimental setup}, we show one example implementation of such a function.

\vspace{0.1cm}

\subsection{Plan Generation Phase}
\vspace{-0.1cm}
The process of generating a plan takes three inputs: the state, expression satisfiability probability, and uncertainty values of each object in \emph{O}. We define these terms as follows:

\begin{definition}\label{DefState}
$State$ of an object $o_k$ {\em w.r.t.} a predicate $R^i_l$ is denoted as $s(o_{k},R^{i}_{l})$ and defined as the set of enrichment functions in $F_i$ that have already been executed on $o_k$ to determine if it contains tag $t_l$. 
\end{definition}

PIQUE maintains for each object a state value for each predicate mentioned in $exp$. The set of all the state values of an object $o_k$ {\em w.r.t.} a query $Q$ is referred to as a \emph{state vector} and denoted as $S_k$.

\begin{definition}\label{Defjoint}
\emph{The Expression Satisfiability Probability} (ESP) value of an object $o_k$ is denoted as $\mathcal P_k$ and defined as the probability of $o_k$ satisfying all the predicates present in the PIQUE expression. The computation of ESP depends on the type of the query. 
\end{definition}


For a PIQUE expression with a single predicate, we use the predicate probability of an object as the ESP value of the object. For a PIQUE expression with multiple predicates, we use the individual predicate probability values to determine the ESP value of an object. We assume that any two predicates containing two different tags of different tag types (e.g., $\textit{Person} $ = $ {\tt John}$ $\Lambda$ $\textit{ Expression} $ = $ {\tt Neutral}$) to be independent and any two predicates containing tags of the same type (e.g., $\textit{ Person} $ = $ {\tt John}$ $ \vee$ $\textit{ Person} $ = $ {\tt David}$) to be mutually exclusive.

\begin{definition}
The \emph{Uncertainty} of each $o_k \in O$ is measured using the {\em entropy} of the object. Given a discrete random variable $X$, the entropy is defined as follows :  
\begin{equation}
H(X) = - \displaystyle\sum_{x} Pr(X = x) \cdot log(Pr(X = x))
\vspace{-0.2cm}
\end{equation}
where $x$ is a possible value of $X$ and $Pr(X = x)$ is the probability of $X$ taking the value $x$. In our setup, we consider a tag type $T_i$ as a random variable and the tags of this type are considered as the possible values of the random variable.  
\end{definition}

The uncertainty value of an object $o_k$ {\em w.r.t.} predicate $R^i_l$ is denoted as $h(o_k, R^{i}_{l})$ and calculated using the predicate probability value as follows:

\vspace{-0.2cm}
\begin{equation}\label{eqn:uncertainty}
\vspace{-0.2cm}
\begin{split}
h(o_k, R^{i}_{l}) &= -p(o_k,R^i_l)\cdot log(p(o_k,R^i_l)) - \\
&(1-p(o_k,R^i_l)) \cdot  log(1-p(o_k,R^i_l))
\end{split}
\end{equation}

\begin{algorithm}[t]
\begin{algorithmic}[1]
\footnotesize
\Procedure{PIQUE}{$O$, $epoch$, $exp$, $\{F_1, F_2, \dots \}$}
\State $\mathsf{A} \gets \varnothing$
\State $S$, $P$, $H$ $\gets$ \Call{Initialize-Data-Structure}
{} // Initialize the state, predicate probability and uncertainty data structures 
 \While{ \Call{Is-Not-Fully-Tagged}{$O$}} \label{alg:stopping criteria}
 		\State $curEpochTime \gets 0$
 		\State $<\mathsf{EP}_w, t_0>$ $\gets$ \Call{Generate-Plan}{$S$, $P$, $H$} // $t_0$: plan generation time
        \While{$curEpochTime < (epoch - t_0)$} \label{epoch}
          \State $ t_1 \gets time() $
          \State $(o_k,R^i_l,f^i_m) \gets  \mathsf{EP}_w.pop()$
          \State \Call{Execute-Triple}{$o_k,R^i_l,f^i_m$} 
          \State $S$, $P$, $H$ $\gets$\Call{Update-Data-structure}{$S$, $P$, $H$} 
          \State $t_2 \gets time() $
          \State $curEpochTime \gets curEpochTime + (t_2 - t_1) $ 
         \EndWhile\label{epoch}
         \State $ \mathsf{A}$ $\gets$ \Call{Select-Answer-Set}{$S$, $P$, $H$} // $\mathsf{A}$ can be consumed by the application
         
  \EndWhile\label{alg:stopping criteria}
\State \Return $ \mathsf{A}$    
\EndProcedure
\captionof{figure}{Progressive Approach}\label{fig:progressive algorithm}
\end{algorithmic}
\end{algorithm}

Given the three inputs above, the process of generating a plan in PIQUE can be viewed as consisting of the following four steps 
(described in details in Section~\ref{sect:plan generation}): 

\begin{enumerate}[leftmargin=*]
\item \textbf{Candidate Set Selection.} Chooses a subset of objects from $O$ which will be considered for the selection of a plan. We denote the candidate set of epoch $w$ as $\mathsf{CS}_w$. This step aims at reducing the complexity of the plan generation phase by considering a lower number of objects for selection.




\item \textbf{Generation of Triples.} Generates triples for the objects present in $\mathsf{CS}_w$. For each object $o_k \in \mathsf{CS}_w$, it determines an enrichment function for $o_k$ for each predicate in the query $Q$. Since $Q$ consists of $|\mathds{R}|$ predicates, this step will generate $|\mathds{R}|$ triples for each object $o_k$ $\in$ $\mathsf{CS}_w$. The set of those $ |\mathds{R}| \cdot |\mathsf{CS}_w| $ triples is denoted as $\mathsf{TS}_{w}$. 

\item \textbf{Benefit Estimation of Triples.} Estimates the benefit of each triple in $\mathsf{TS}_{w}$. As our quality metric is the $F_{\alpha}$-measure, we define a metric called the {\em expected} $F_\alpha$-measure to estimate the quality of the answer set at each epoch. Based on this, we define a triple {\em Benefit} as the improvement in the expected $F_\alpha$ value of the answer set (that will be result from executing that triple) per unit cost.  

  
  
 

\item \textbf{Selection of Triples.} Compares between the benefit values of the triples and generates a plan that consists of the triples with the highest benefit values. The triples in the plan will be sorted in a decreasing order according to that value. We denote the plan selected in epoch $w$ as $\mathsf{EP}_w$.


 
\end{enumerate}

 
\subsection{Plan Execution Phase} This phase iterates over the list of triples in $\mathsf{EP}_w$ and, for each triple ($o_k$, $R^i_l$, $f^i_m$), executes it by calling the function $f^i_m$ with these two parameters: the object $o_k$ and the tag $t_l$ present in the predicate $R^i_l$. Triples are executed until the allotted time for the epoch is consumed. After the plan is executed, PIQUE updates the state, predicate probability, and uncertainty values for the objects involved in the executed triples of $\mathsf{EP}_w$.

\subsection{Answer Set Selection Phase}
\label{Answercomputealgorithm}
At the end of each epoch $w$, PIQUE determines an answer set $\mathsf{A}_w$ from $O$ based on the output of the plan execution phase. Our strategy of constructing $\mathsf{A}_w$ chooses a subset from $O$ which optimizes the quality of the answer set. Since the ground truth of the objects in $O$ is not known, we use the expression satisfiability probability values of the objects to measure the quality metrics (i.e., the $F_\alpha$-measure) in an expected sense. Given an answer set $\mathsf{A}_w$, the expected quality of the answer set $\mathsf{A}_w$, denoted as $E(F_\alpha(\mathsf{A}_w))$, is defined as follows: 
\begin{equation}\label{eqn:expected_f1_measure}
E(F_\alpha(\mathsf{A}_w)) =\frac{(1+\alpha)\cdot \sum\limits_{o_i \in \mathsf{A}_w}\mathcal P_i}{\alpha \cdot \sum\limits_{o_j \in O} \mathcal P_j+ |\mathsf{A}_w|}
\vspace{-0.6cm}
\end{equation}
where
\begin{equation}
\begin{split}
\vspace{-1cm}
E(Pre(\mathsf{A}_w)) =  \frac{\sum\limits_{o_i \in \mathsf{A}_w} \mathcal P_i}{|\mathsf{A}_w|} \ \text{and} \    E(Rec(\mathsf{A}_w))= \frac{\sum\limits_{o_i \in \mathsf{A}_w} \mathcal P_i}{\sum\limits_{o_j \in O} \mathcal P_j}
\end{split}
\end{equation}

Our strategy sets $\mathsf{A}_w$ to the subset of $O$ that has the highest expected $F_\alpha$ value among all possible subsets of $O$. Our strategy is based on the following observation.


\begin{theorem}\label{theorem:threshold}
Let $L$ be the list of all objects of $O$ sorted in a decreasing order of their ESP values and Let $L^k$ be the list that consists of the first $k$ objects in $L$. The expected $F_\alpha$ values of the possible subsets of $L$ follow this pattern: 
$E(F_\alpha(L^{1})) \, < \, E(F_\alpha(L^{2})) \, < \, \dots \, < \, E(F_\alpha(L^{\tau - 1})) \, > \,  E(F_\alpha(L^{\tau})) \,  > \, E(F_\alpha(L^{\tau + 1})) \, > \, \dots \, > \, E(F_\alpha(L))$.
\end{theorem}

\begin{proof}
In this proof, we show that if $E(F_\alpha)$ measure of the answer set decreases for the first time, due to the inclusion of a particular object in the answer set, then it will keep decreasing monotonically with the inclusion of any further objects. Based on the notations used in the above theorem, we can rewrite the expressions of $E(F_\alpha(L^{\tau}))$, $E(F_\alpha(L^{\tau+1}))$ and $E(F_\alpha(L^{\tau+2}))$ as follows:

\begin{equation}\label{eqn:f1_expression_3}
\vspace{-0.2cm}
\begin{split}
\vspace{-0.4cm}
&E(F_\alpha(L^{\tau})) =\frac{(1+\alpha).\frac{k_1}{\tau}.\frac{k_1}{k_2}}{\alpha \cdot\frac{k_1}{\tau}+\frac{k_1}{k_2}}
\\
&= \frac{(1+\alpha)\cdot k_1}{\alpha \cdot k_2+\tau},
 k_1 = \sum\limits_{i=1}^\tau \mathcal P_i,  
 k_2 = \sum\limits_{i=1}^{|O|} \mathcal P_i 
\end{split}
\end{equation}

Similarly, the value of $E(F_\alpha(L^{\tau+1})) = \frac{(1+\alpha)(k_1+\mathcal P_{\tau+1})}{(\alpha k_2 + \tau + 1)}$ and the value of $E(F_\alpha(L^{\tau+2}))= \frac{(1+\alpha)(k_1+  \mathcal P_{\tau+1} +\mathcal P_{\tau+2})}{(\alpha k_2+\tau+2)}$.


\begin{equation}
\begin{split}
&E(F_\alpha(L^{\tau+1})) < E(F_\alpha(L^{\tau}))\\ 
&\Rightarrow \frac{(1+\alpha)\cdot(k_1+\mathcal P_{\tau+1})}{(\alpha  k_2 + \tau + 1)} < \frac{(1+\alpha)\cdot k_1}{\alpha  k_2+\tau}  \\
&\Rightarrow (k_1+\mathcal P_{\tau+1})(\alpha k_2+\tau) <  k_1(\alpha k_2+\tau + 1) \\
&\Rightarrow \alpha k_1 k_2+k_1 \tau+\alpha k_2\mathcal P_{\tau+1}+\tau\mathcal P_{\tau+1} < \alpha k_1k_2+k_1\tau+k_1
\end{split}
\end{equation}

Simplifying some more steps, we get the following condition: $\frac{(k_1+\mathcal P_{\tau+1}+\mathcal P_{\tau+2})}{(\alpha k_2+ \tau+2)} < \frac{(k_1+\mathcal P_{\tau+1})}{\alpha k_2+\tau +1}$. From this condition we can conclude that $E(F_\alpha(L^{\tau+2})) < E(F_\alpha(L^{\tau+1})) $.
\end{proof}

Based on the above observation, we set $\mathsf{A}_w$ to $L^{\tau - 1}$ and we refer to the ESP of the ${\tau}^{th}$ object as the {\em threshold} probability of epoch $w$ and denote it as $\mathcal P^\tau_w$. In the following we provide an example of determining threshold.

\begin{example}
Let $L = \{o_1, o_4, o_5, o_2, o_3\}$ such that $\mathcal P_1 = 0.9 > \mathcal P_4 = 0.8 > \mathcal P_5 = 0.75 > \mathcal P_2 = 0.3 > \mathcal P_3 = 0.2 $. In this case, $E(F_\alpha(L^1)) =0.46 < E(F_\alpha(L^2)) =0.68 < E(F_\alpha(L^3)) =0.82 > E(F_\alpha(L^4)) =0.79 > E(F_\alpha(L^5)) = 0.74$. Hence, $\mathsf{A}_w = L^3 = \{o_1, o_4, o_5\}$ and $\mathcal P^\tau_w = \mathcal P_5 = 0.75$. 
\end{example}

\section{Plan Generation}
\label{sect:plan generation}

In this section, we explain in details the four steps of the plan generation phase.

\subsection{Candidate Set Selection}\label{section:candidate_set_selection}
The objective of this step is to choose a subset of objects $\mathsf{CS}_w$ from $O$ which will be considered for the generation of the plan $\mathsf{EP}_w$. Our strategy of choosing $\mathsf{CS}_w$ is based on the observation that evaluating a triple corresponding to an object $o_k \not \in \mathsf{A}_{w-1}$ in epoch $w$ ensures that $E(F_{\alpha}(\mathsf{A}_{w}))$ increases monotonically. That is, $E(F_\alpha (\mathsf{A}_{w})) >= E(F_\alpha (\mathsf{A}_{w-1}))$ irrespective of the outcome of executing that triple.
In contrast, evaluating a triple corresponding to an object $o_k \in A_w$ could either increase or decrease  $E(F_{\alpha}(\mathsf{A}_{w}))$  depending on the outcome of the triple execution. Let us illustrate this observation using the following example.

\begin{example}
Let $O = \{o_1, o_4, o_5, o_2, o_3\}$ be a dataset such that $\mathcal P_1 = 0.9, \mathcal P_4 = 0.8, \mathcal P_5 = 0.75, \mathcal P_2 = 0.3, \mathcal P_3 = 0.2 $. In this case, the answer set $\mathsf{A}_{w-1} = \{o_1, o_4, o_5\}$,  $\mathcal P^{\tau}_{w-1}$ = 0.75  and  $E(F_\alpha(A_{w-1})) =0.82$. 
Evaluating a triple for either objects $o_2$ or $o_3$  will cause $E(F_{\alpha}(\mathsf{A}_{w}))$ to increase or remain the same as $E(F_{\alpha}(\mathsf{A}_{w-1}))$. 
For instance, suppose the execution of a triple, corresponding to $o_2$, in epoch $w$ causes the probability value $\mathcal P_2$ to increase from $0.3$ to $0.78$ (which is $ > \mathcal P^{\tau}_{w-1} = 0.75$). Then, $E(F_\alpha(\mathsf{A}_{w}))$ will increase from $0.82$ to $0.87$. Even if $\mathcal P_2$ increases from $0.3$ to $0.5$ (which is $<\mathcal P^{\tau}_{w-1}$), $E(F_\alpha(\mathsf{A}_{w}))$ will increase from $0.82$ to $0.83$. If $\mathcal P_2$ decreases from $0.3$ to $0.1$,   $E(F_\alpha(\mathsf{A}_{w}))$ will remain the same as $E(F_\alpha(\mathsf{A}_{w-1}))$.


 Now, consider object $o_4 \in \mathsf{A}_{w-1}$. If $\mathcal P_4$ decreases from $0.8$ to $0.76$ (which is $ > \mathcal P^{\tau}_{w-1}$), then $E(F_\alpha(\mathsf{A}_{w}))$ will decrease to $0.81$. Likewise, if $\mathcal P_4$ decreases to $0.7$ (which is $< \mathcal P^{\tau}_{w-1}$), $E(F_\alpha(\mathsf{A}_{w}))$ will decrease to $0.8$. $E(F_\alpha(\mathsf{A}_{w}))$ will only increase if $\mathcal P_4$ increases as a result of executing the triple; e.g.,  $E(F_\alpha(\mathsf{A}_{w}))$ will increase to $0.83$ if $\mathcal P_4$ increases to $0.85$.


\end{example}


The intuition of the above example is captured in the following theorem that guides our choice of candidate set selection.

\begin{theorem}\label{theorem:thresholdFunctionInsideObjectAll}
 Let $(o_k, R^i_l, f^i_m)$ be a triple evaluated in epoch $w$ and let $\mathcal P_k$ and $\mathcal P_k'$ be the ESP values of $o_k$ at the end of  epoch $w-1$ and $w$ respectively. Then, the following conditions hold for $E(F_\alpha(\mathsf{A}_{w}))$:
\begin{enumerate}[leftmargin=*]
     \item If $\mathcal P_k >= \mathcal P^{\tau}_{w-1}$, then $E(F_\alpha(\mathsf{A}_{w})) >= E(F_\alpha(\mathsf{A}_{w-1}))$  iff $\mathcal P_k' > \mathcal P_k$.
     \item If $\mathcal P_k < \mathcal P^{\tau}_{w-1}$, then $E(F_\alpha(\mathsf{A}_{w})) >= E(F_\alpha(\mathsf{A}_{w-1}))$  irrespective of the values of $\mathcal P_k$ and $\mathcal P_k'$.
\end{enumerate}
\end{theorem}

\begin{proof}
From Equation \ref{eqn:expected_f1_measure}, it can be seen that if the ESP value of any object inside the answer set increases, then both the numerator and denominator increases by the same amount (an amount which is less than 1). This implies that the value of the fraction $E(F_{\alpha}(\mathsf{A}_w))$ will increase. Similarly, if the ESP value decreases by any amount less than 1 from both the numerator and denominator, then the value of $E(F_{\alpha}(\mathsf{A}_w))$ will decrease. Hence, it proves the first part of the theorem. 

If the ESP value of an object outside of the answer set increases and becomes higher than $\mathcal P^{\tau}_{w-1}$, then the numerator of $E(F_{\alpha}(\mathsf{A}_w))$ increases by the amount of $\mathcal P_k'$. The denominator of $E(F_{\alpha}(\mathsf{A}_w))$ also increases (the term $\alpha \cdot \sum\limits_{o_j \in O} \mathcal P_j$) but it increases by a smaller amount (i.e., $\mathcal P_k'-\mathcal P_k$) as compared to the numerator. This implies that $E(F_{\alpha}(\mathsf{A}_w))$ will be higher than $E(F_{\alpha}(\mathsf{A}_{w-1}))$ proving the second part of the theorem. The detailed proof of this theorem is shown in Appendix \ref{appendix:proofs}.
\vspace{-0.3cm}
\end{proof}

Based on the above theorem, our approach considers $\mathsf{CS}_w$ to be the set \{$o_k$ $\vert$ $o_k$ $\in$ $O$\ -\ $ \mathsf{A}_{w-1}$\} to guarantee progressive increase in $E(F_\alpha(\mathsf{A}_{w}))$. Limiting the candidate set also reduces the complexity since we no longer need to consider triples corresponding to objects from $\mathsf{A}_{w-1}$.

\begin{table*}[t]\label{DecisionTable}
\footnotesize
\centering
\begin{center}
 \begin{tabular}{|c|c|c|c|c|c|} 
 \hline
 \textbf{Tag Type} & \textbf{Tag} & \textbf{State} & \textbf{Uncertainty Ranges} & \textbf{Next Function} & \textbf{$\Delta$ Uncertainty} \\
 \hline
 $T_i$ & $t_l$ &[0, 0, 0, 1] & [0 - 0.1), [0.1-0.2), $\cdots$, [0.9-1]  & $f^{i}_{2}$, $f^{i}_{1}$, $\cdots$, $f^{i}_{3}$ & -0.04, -0.12, $\cdots$, -0.22 \\ 
 \hline
 $T_i$ & $t_l$ & [0, 0, 1, 1] & [0 - 0.1), [0.1-0.2), $\cdots$, [0.9-1] & $f^{i}_{1}$, $f^{i}_{1}$, $\cdots$, $f^{i}_{2}$ & -0.02, -0.11, $\cdots$, -0.28  \\
 \hline
  $\cdots$ & $\cdots$ & $\cdots$ & $\cdots$ & $\cdots$ & $\cdots$ \\
 \hline
 $T_i$ & $t_l$ & [1, 1, 1, 0] & [0 - 0.1), [0.1-0.2), $\cdots$, [0.9-1] & $f^{i}_{4}$, $f^{i}_{4}$, $\cdots$, $f^{i}_{4}$ & -0.03, -0.15, $\cdots$, -0.18 \\
 
 \hline
\end{tabular}
\vspace{-.2cm}

\caption{An example decision table. State [0,0,0,1] is a state in which only the  function $f^i_4$ has been applied on the object.}\label{DecisionTable}
\vspace{-0.8cm}
\label{table:1}
\end{center}
\end{table*}

\subsection{Generation of Triples}
\label{triple generation}

The objective of this step is to generate a set of (object, predicate, function) triples, denoted as $\mathsf{TS}_w$, from the objects present in $\mathsf{CS}_w$. For each object $o_k$ $\in$ $\mathsf{CS}_w$, this step identifies which function should be executed next on the object {\em w.r.t.} each predicate in the query $Q$. Since $Q$ consists of $|\mathds{R}|$ predicates, this step will generate $|\mathds{R}|$ triples for $o_k$.

The enrichment functions included in the triples are determined using a decision table the structure of which  is shown in Table \ref{DecisionTable}. (In Section \ref{sect:experimental setup}, we discuss how this table is constructed.) For each pair of a tag type $T_i$ and a tag $t_j$ $\in$ $T_i$, this table stores, for each possible state value, a list of $m$ uncertainty ranges, a list of $m$ enrichment functions and a list of $m$ $\Delta$ uncertainty values.

Given an object $o_k$ $\in$ $\mathsf{CS}_w$ and a predicate $R^i_l$, the function that should be executed on $o_k$ {\em w.r.t.} $R^i_l$ is determined using the state value $s(o_k, R^{i}_{l})$ and the uncertainty value $h(o_{k},R^{i}_{l})$ as follows. If $h(o_{k},R^{i}_{l})$ lies in the $x^{\text{\em th}}$ uncertainty range (where $x<m$) of state $s(o_k, R^{i}_{l})$, then the next function that should be executed on $o_k$ is the $x^{\text{\em th}}$ function in the list of enrichment functions. For example, if $s(o_k, R^{i}_{l}) = [0, 0, 0, 1]$ (first row in Table \ref{DecisionTable}) and $h(o_{k},R^{i}_{l}) = 0.92$, then our technique returns $f^{i}_{3}$ (last function in the list). Such a function is expected to provide the highest reduction in $h(o_{k},R^{i}_{l})$ (last column of the table) among the remaining functions that have not been executed on $o_k$. In our example, applying $f^{i}_{3}$ on $o_k$ is expected to reduce the uncertainty of $o_k$ by $0.22$ (the last value of the last column). (The $\Delta$ uncertainty value will be used in estimating the benefit of triples as we will see in the next section.)

\subsection{Benefit Estimation of Triples}\label{BenefitEstimation}

In this step, we derive the benefit of triples present in $\mathsf{TS}_w$ using a benefit metric.

\begin{definition}
The \emph{Benefit} of a triple ($o_k$, $R^i_l$, $f^i_m$), denoted as $Benefit$($o_k$, $R^i_l$, $f^i_m$), is defined as the increase in the expected $F_\alpha$ value of the answer set (that is caused by the evaluation of the triple) per unit cost. Formally, it is calculated as follows: 
\begin{equation}\label{eqn:benefit_estimation}
Benefit(o_k,R^i_l,f^i_m)  =\frac{\splitfrac{E(\hat{F_\alpha}(\mathsf{A}_{w}, (o_k, R^i_l, f^i_m)))}{ - E(F_\alpha(\mathsf{A}_{w-1}))}}{c^{i}_{m}}
\end{equation}
\noindent where $E(\hat{F_\alpha}(\mathsf{A}_{w}, (o_k, R^i_l, f^i_m))))$ is the expected quality of $\mathsf{A}_{w}$ if only triple \emph{$(o_k,R^i_l,f^i_m)$} is executed in epoch $w$ and $c^{i}_{m}$ is the cost of evaluating function $f^i_m$. For ease of notation, we will refer to $E(\hat{F_\alpha}(\mathsf{A}_{w}, (o_k, R^i_l, f^i_m)))$ as $E(\hat{F_\alpha}(\mathsf{A}_{w}))$.
\end{definition}



A naive strategy for estimating the benefit of a triple ($o_k$, $R^i_l$, $f^i_m$) would first estimate the new expression satisfiability probability of object $o_k$ that will be obtained if function $f^i_m$ is executed on $o_k$, denoted as $\hat{\mathcal P}_k$. Then, it would use the algorithm in Section~\ref{Answercomputealgorithm} to estimate  $\hat{F_\alpha}(w)$ and hence the denominator in Equation \ref{eqn:benefit_estimation}. That is, it would generate a new answer set (Section~\ref{Answercomputealgorithm}) assuming that the expression satisfiability probability of $o_k$ is set to $\hat{\mathcal P_k}$ (instead of $\mathcal P_k$) and the expression satisfiability probability of all the other objects (i.e., \{$o_r$ $\in$ $O$, $r$$\neq$$k$\}) are fixed to their values in epoch $w-1$. Such a strategy exhibits a time complexity of $\mathcal O(n)$ for estimating the benefit of each triple and hence $\mathcal O(n^2)$ for estimating the benefits of the triples in $\mathsf{TS}_w$, where $n$ is the number of objects.

Instead, we propose an efficient strategy that estimates a {\em relative} benefit value for each triple ($o_k$, $R^i_l$, $f^i_m$) using the expression satisfiability probability $\hat{\mathcal P_k}$. We discuss below how $\hat{\mathcal P_k}$ is estimated and then present our efficient strategy. This strategy exhibits a complexity of $\mathcal O(1)$ for estimating the benefit of each triple and $\mathcal O(n)$ for estimating the benefits of the triples in $\mathsf{TS}_w$, where $n$ is the number of objects.

\subsubsection{Estimation of expression satisfiability probability value}\label{sect:estimation_qsp} Given a triple ($o_k$, $R^i_l$, $f^i_m$), in order to estimate $\hat{\mathcal P}_k$, we first need to estimate the uncertainty value of $o_k$ and the predicate probability value of $o_k$ {\em w.r.t.} predicate $R^i_l$ that will be obtained if function $f^i_m$ is executed on $o_k$. We denote these two values as $\hat{h}(o_{k},R^{i}_{l})$ and $\hat{p}(o_k,R^{i}_{l})$, respectively.

The uncertainty value $\hat{h}(o_{k},R^{i}_{l})$ can be estimated using the decision table based on the state $s(o_{k},R^{i}_{l})$ and the current uncertainty value $h(o_{k},R^{i}_{l})$. Suppose that the retrieved $\Delta$ uncertainty is $u$. Then, $\hat{h}(o_{k},R^{i}_{l})$ = $h(o_{k},R^{i}_{l})$ + $u$. For example, consider the decision table in Table \ref{DecisionTable}. Assuming  that the state $s(o_{k},R^{i}_{l}) = [0, 0, 1, 1]$ and the uncertainty value $h(o_{k},R^{i}_{l}) = 0.93$, then $\hat{h}(o_{k},R^{i}_{l}) = 0.93 - 0.28 = 0.65$. Given the value of $\hat{h}(o_{k},R^{i}_{l})$, the predicate probability $\hat{p}(o_{k},R^{i}_{l})$ is then calculated by the inverse equation of entropy (i.e., Equation \ref{eqn:uncertainty}) as follows:
\begin{equation}\label{eqn:InverseUncertainty}
\begin{split}
\hat{h}(o_{k},R^{i}_{l}) &= -\hat{p}(o_{k},R^{i}_{l}) \cdot log(\hat{p}(o_{k},R^{i}_{l}))- \\ &(1-\hat{p}(o_{k},R^{i}_{l})) \cdot log(1-\hat{p}(o_{k},R^{i}_{l}))
\end{split}
\end{equation}

Note that there are two possible solutions for this equation. The first possible value of $\hat{p}(o_{k},R^{i}_{l})$ is higher than the current value ${p}(o_{k},R^{i}_{l})$ whereas the second one is lower than ${p}(o_{k}, R^{i}_{l})$. Based on our approach in Section \ref{section:candidate_set_selection}, we only consider the value that is higher than $\mathcal P_k$ because it always increases the expected $F_\alpha$ value of the answer set whereas the second value may increase it or keep it the same.

\begin{example}
Suppose that the current predicate probability $p(o_{k},R^{i}_{l})$ is $0.7$ and the current ESP value $\mathcal P_k$ is $0.66$. This implies that the uncertainty value of $o_k$ {\em w.r.t.} predicate $R^i_l$ is $0.92$ (derived using Equation \ref{eqn:uncertainty}). Assuming that the decision table is the one shown in Table \ref{DecisionTable} and the state value $s(o_{k},R^{i}_{l}) = [0, 0, 1, 1]$, then, $\hat{h}(o_{k},R^{i}_{l}) = 0.92 - 0.28 = 0.64$. Using Equation \ref{eqn:InverseUncertainty}, the estimated predicate probability $\hat{p}(o_{k},R^{i}_{l})$ will have two possible values $0.84$ and $0.16$. As explained earlier in Definition \ref{Defjoint}, depending on the tags involved in the predicates, the two possible values of $\hat{p}(o_{k},R^{i}_{l})$ can then be used to derive two possible values for the ESP value $\hat{\mathcal P}_k$. Suppose that these two values are $0.76$ and $0.12$.  This implies that if triple ($o_k$, $R^i_l$, $f^i_m$) is executed in epoch $w$, then the expression satisfiability probability of $o_k$ will either increase from $0.66$ to $0.76$ or decrease from $0.66$ to $0.12$. Based on our observation in Section \ref{section:candidate_set_selection}, we only set $\mathcal P_k$  to $0.76$ because it always increases the expected $F_\alpha$ value.

\end{example}

\subsubsection{Efficient Benefit Estimation}
\label{sect:efficient benefit estimation} Our strategy for estimating the benefit values is based on the relationship between the benefit metric of a triple, as defined in Equation \ref{eqn:benefit_estimation}, and the local properties of such a triple (i.e., the expression satisfiability probability of the object, the cost of the function, etc.). The objective is to eliminate the step of generating a new answer set (using Equation \ref{eqn:benefit_estimation}) to estimate the increase in the quality of the answer set (i.e., $\hat{F_\alpha}(w)$).

Let ($o_k$, $R^i_l$, $f^i_m$) and ($o_q$, $R^s_t$, $f^s_v$) be two triples present in $\mathsf{TS}_w$. Let the estimated ESP value of $o_k$ (i.e., $\hat{\mathcal P}_k$) if triple ($o_k$, $R^i_l$, $f^i_m$) is executed in epoch $w$ be $\hat{\mathcal P}_k = (\mathcal P_k+ \Delta{\mathcal P_k})$ and the estimated ESP value of $o_q$ if triple ($o_q$, $R^s_t$, $f^s_v$) is executed in epoch $w$ be $\hat{\mathcal P}_q = (\mathcal P_q +\Delta{\mathcal P_q})$.

Let $m_1$ be the number of objects that are part of $\mathsf{A}_{w-1}$ and will move out of $\mathsf{A}_{w}$ as a result of changing the expression satisfiability probability of object $o_k$ from $\mathcal P_k$ to $\hat{\mathcal P}_k$. We have observed that since $o_k \not\in \mathsf{A}_{w-1}$ (because $\mathsf{CS}_w$ consists only of objects from outside the answer set $\mathsf{A}_{w-1}$) and $\hat{\mathcal P}_k > \mathcal P_k$, the number of objects in $\mathsf{A}_{w}$ will be less than or equal to that of $\mathsf{A}_{w-1}$. Hence, $m_1 \geq 0$. Furthermore, let $m_2$ be the number of objects that are part of $\mathsf{A}_{w-1}$ and will move out of $\mathsf{A}_{w}$ as a result of changing the expression satisfiability probability of object $o_q$ from $\mathcal P_q$ to $\hat{\mathcal P}_q$. Similarly, $m_2 \geq 0$.

Then, the possible values of $m_1$ and $m_2$ can be as follows. Both $m_1$ and $m_2$ can be greater than zero and $m_1$ is greater than $m_2$. Similarly, both of them can be greater than zero and $m_1$ is less than or equal to $m_2$. Both $m_1$ and $m_2$ are equal to zero which implies that the number of objects in both $\mathsf{A_{w-1}}$ and $\mathsf{A_{w}}$ are the same. Given these three cases, we make the following observations about the benefit values of the triples in $\mathsf{TS}_w$. 

\begin{theorem}\label{theorem:SelectionCriteriaWithThresholdChange}
The triple ($o_k$, $R^i_l$, $f^i_m$) will have a higher benefit value than the triple ($o_q$, $R^s_t$, $f^s_v$) in epoch $w$ irrespective of the values of $m_1$ and $m_2$ if the following condition holds:

\begin{equation}\label{eqn:DeriveBenefitCondition3}
\vspace{-0.1cm}
\frac{\mathcal P_k(\mathcal P_k+\Delta \mathcal P_k)}{c^{i}_{m}} > \frac{\mathcal P_q(\mathcal P_q+\Delta \mathcal P_q)}{c^{s}_{v}}
\end{equation}
\end{theorem}

\textbf{Proof.} We prove this theorem as follows: given three possible cases of $m_1$ and $m_2$, we consider all possible combinations (16 possible combinations) of the values of $\mathcal P_k$, $\mathcal P_q$, $\Delta \mathcal P_k$, and $\Delta \mathcal P_q$ and show if Equation \ref{eqn:DeriveBenefitCondition3} holds, then the benefit of $(o_k, R^i_l, f^i_m)$ will be higher than the benefit of $(o_q, R^s_t, f^s_v)$. In the following we provide the proof of three such scenarios:

 For ease of notation, we denote the threshold $\mathcal P^\tau_{w-1}$ of epoch $w-1$ as $\mathcal P_\tau$ in this proof. The notation of $E(F_\alpha(\mathsf{A}_{w-1}))$, i.e., $\frac{(1+\alpha)(\mathcal P_1+ \cdots +\mathcal P_\tau)}{\alpha(\mathcal P_1+\mathcal P_2+ \cdots +\mathcal P_{|O|})+ \tau}$ can be simplified as $\frac{X}{Y+ \tau}$. We denote the value of $\frac{\mathcal P_k}{c^l_n}$ by $\nu_k$ and the value of $\frac{\mathcal P_q}{c^s_v}$ by $\nu_q$. Rewriting the expression of benefit for both the triples, the benefit value of triple ($o_k$, $R^l_m$, $f^l_n$) will be higher than the triple ($o_q$, $R^s_t$, $f^s_v$), when the following condition holds:

Simplifying the expression of benefit the triples, benefit value of triple ($o_k$, $R^l_m$, $f^l_n$) will be higher than the triple ($o_q$, $R^s_t$, $f^s_v$), when the following condition holds: 
\begin{equation}\label{pm_benefit_compare}
\begin{split}
&\nu_k(\frac{\splitfrac{X-(1+\alpha)\cdot(\mathcal P_\tau+\mathcal P_{\tau-1}+...\mathcal P_{\tau-(m_1-1)})+}{(1+\alpha)\cdot(\mathcal P_k+\Delta \mathcal P_k)}}{Y+(\tau-m_1) + \alpha\cdot\Delta{\mathcal P_k}}) > \\
&\nu_q(\frac{\splitfrac{X-(1+\alpha)\cdot(\mathcal P_\tau+\mathcal P_{\tau-1}+...\mathcal P_{\tau-(m_2-1)})+}{(1+\alpha)\cdot(\mathcal P_q+\Delta \mathcal P_q)}}{Y+(\tau-m_2) + \alpha\cdot\Delta{\mathcal P_q}}) 
\end{split}
\end{equation}


\vspace{0.1cm}
\noindent
\textbf{Case 1: $\Delta \mathcal P_k$ $<$ $\Delta \mathcal P_q$, $\mathcal P_k + \Delta \mathcal P_k$ $>$ $\mathcal P_q + \Delta \mathcal P_q$,  and $m_1 > m_2$}. 

\noindent
Comparing both the denominators of Equation \ref{pm_benefit_compare}, we can see that $(\tau-m_1)<(\tau-m_2)$ and $\Delta \mathcal P_k$ $<$ $\Delta \mathcal P_q$. This implies that the denominator on the L.H.S. is smaller than the denominator on the R.H.S. In the numerator of L.H.S, the value of $(\mathcal P_\tau+\mathcal P_{\tau-1}+...\mathcal P_{\tau-(m_1-1)})$ is higher than $(\mathcal P_\tau+\mathcal P_{\tau-1}+...\mathcal P_{\tau-(m_2-1)})$ as $m_1$ is higher than $m_2$. Furthermore, if $\nu_k(\mathcal P_k +\Delta \mathcal P_k) > \nu_q(\mathcal P_q+ \Delta \mathcal P_q)$ then the numerator of the L.H.S. will always be higher than the numerator of the R.H.S. Thus we can conclude that the Equation \ref{pm_benefit_compare} will be satisfied when the condition $\nu_k(\mathcal P_k+\Delta \mathcal P_k) > \nu_q(\mathcal P_q+ \Delta \mathcal P_q)$ is satisfied.


\vspace{0.1cm}
\noindent
\textbf{Case 2: $\Delta \mathcal P_k$ $>$ $\Delta \mathcal P_q$, $\mathcal P_k + \Delta \mathcal P_k$ $>$ $\mathcal P_q + \Delta \mathcal P_q$, and $m_1 > m_2$}. In Equation \ref{pm_benefit_compare}, the value of $(\mathcal P_\tau+\mathcal P_{\tau-1}+...\mathcal P_{\tau-(m_1-1)})$ on the L.H.S is higher than $(\mathcal P_\tau+\mathcal P_{\tau-1}+...\mathcal P_{\tau-(m_2-1)})$ as $m_1$ is higher than $m_2$. In the denominator, although the value of $\Delta \mathcal P_k$ is higher than $\Delta \mathcal P_q$, the total value of $(\tau-m_1 + \alpha \cdot \Delta \mathcal P_k) $ will be lower than $ (\tau-m_2 + \alpha \cdot \Delta \mathcal P_q)$ as both $\Delta \mathcal P_k$ and $\Delta \mathcal P_q$ are less than one.

\vspace{0.1cm}
\noindent
\textbf{Case 3: $\Delta \mathcal P_k$ $>$ $\Delta \mathcal P_q$, $\mathcal P_k + \Delta \mathcal P_k$ $<$ $\mathcal P_q + \Delta \mathcal P_q$, and $m_1,m_2=0$}. Let us compare the L.H.S and R.H.S. of Equation \ref{pm_benefit_compare}. In the numerator, if the term of $\nu_k(\mathcal P_k+\Delta \mathcal P_k)$ is higher than the value of $\nu_q(\mathcal P_q+\Delta \mathcal P_q)$, then the numerator of L.H.S will be higher than the numerator of R.H.S. Hence, the value of the expression in the left hand side will be higher.

\vspace{0.1cm}
\noindent
\textbf{Case 4: $\Delta \mathcal P_k$ $<$ $\Delta \mathcal P_q$, $\mathcal P_k + \Delta \mathcal P_k$ $<$ $\mathcal P_q + \Delta \mathcal P_q$, and $m_1,m_2=0$}. In Equation \ref{pm_benefit_compare}, after simplifying some steps further, we derive that the condition in which the L.H.S. will be higher than the R.H.S. is as follows: $\nu_k(\mathcal P_k+\Delta \mathcal P_k)\Delta \mathcal P_q$ $>$ $\nu_q(\mathcal P_q+\Delta \mathcal P_q)\Delta \mathcal P_k$. According to the assumption $\Delta \mathcal P_q$ value is higher than the value of $\Delta \mathcal P_k$. This implies that, if the condition $\nu_k(\mathcal P_k+\Delta \mathcal P_k)$ $>$ $\nu_q(\mathcal P_q+\Delta \mathcal P_q)$ is satisfied, then the L.H.S will be higher than the right hand side.

The above proofs will also hold for the scenarios where $m_1=m_2$ and $m_1>0$. Only difference in the proofs of that scenario from the proofs of Cases 1-4 will be as follows: an additional constant term (i.e., the term of $(\mathcal P_\tau+\mathcal P_{\tau-1}+\cdots+\mathcal P_{\tau-(m_1-1)})$) will be added to the numerators of both the sides of Equation \ref{pm_benefit_compare}. The remaining steps will remain the same as the proofs of Cases 1-4.

\vspace{0.1cm}
\noindent
\textbf{Case 5: $\Delta \mathcal P_k$ $<$ $\Delta \mathcal P_q$, $\mathcal P_k + \Delta \mathcal P_k$ $>$ $\mathcal P_q + \Delta \mathcal P_q$,  and $m_1 < m_2$}. Comparing both the denominators of Equation \ref{pm_benefit_compare}, we can see that $(\tau-m_1)<(\tau-m_2)$ and $\Delta \mathcal P_k$ $<$ $\Delta \mathcal P_q$. This implies that the denominator of L.H.S. is lower than the denominator of R.H.S. In the numerators, the value of $(\mathcal P_\tau+\mathcal P_{\tau-1}+...\mathcal P_{\tau-(m_1-1)})$ is smaller than $(\mathcal P_\tau+\mathcal P_{\tau-1}+...\mathcal P_{\tau-(m_2-1)})$ as $m_1$ is smaller than $m_2$. This condition makes the numerator of L.H.S higher than R.H.S. Furthermore, if $\nu_k(\mathcal P_k +\Delta \mathcal P_k) > \nu_q(\mathcal P_q+ \Delta \mathcal P_q)$ is satisfied then the numerator of L.H.S. will be higher than R.H.S.

\vspace{0.1cm}
\noindent
\textbf{Case 6: $\Delta \mathcal P_k$ $>$ $\Delta \mathcal P_q$, $\mathcal P_k + \Delta \mathcal P_k$ $>$ $\mathcal P_q + \Delta \mathcal P_q$, and $m_1 < m_2$}. From Equation \ref{pm_benefit_compare}, we can derive the following equation: \begin{equation}\label{pm_upperbound_n}
\begin{split}
&\nu_k(X-(1+\alpha)\cdot(\mathcal P_\tau+\mathcal P_{\tau-1}+...\mathcal P_{\tau-(m_1-1)})+(1+\alpha)\cdot\\
&(\mathcal P_k+\Delta \mathcal P_k))\cdot(Y+(\tau-m_2) + \alpha\cdot\Delta{\mathcal P_q})>\nu_j(X- \\
&(1+\alpha)\cdot(\mathcal P_\tau+\mathcal P_{\tau-1}+...\mathcal P_{\tau-(m_2-1)})+(1+\alpha)\cdot
\\&(\mathcal P_q+\Delta \mathcal P_q))\cdot(Y+(\tau-m_1) + \alpha\cdot\Delta{\mathcal P_k})
\end{split}
\end{equation}

In the above equation, the value of $(\mathcal P_\tau+\mathcal P_{\tau-1}+...\mathcal P_{\tau-(m_1-1)})$ is lower than $(\mathcal P_\tau+\mathcal P_{\tau-1}+...\mathcal P_{\tau-(m_2-1)})$ as $m_1$ is smaller than $m_2$. This favors the value in the left hand side of the equation. Furthermore, if the value of $\nu_k(\mathcal P_k+\Delta \mathcal P_k)$ is higher than the value of $\nu_q(\mathcal P_q+\Delta \mathcal P_q)$, then Equation \ref{pm_benefit_compare} will be satisfied.

\vspace{0.1cm}
\noindent
\textbf{Case 7: $\Delta \mathcal P_k$ $>$ $\Delta \mathcal P_q$, $\mathcal P_k + \Delta \mathcal P_k$ $<$ $\mathcal P_q + \Delta \mathcal P_q$, and $m_1 < m_2$}. Comparing the L.H.S. of the Equation \ref{pm_upperbound_n} with the R.H.S. of the equation, we can see that if the value of $\nu_k(\mathcal P_k+\Delta \mathcal P_k)$ is higher than the value of $(\mathcal P_q+\Delta \mathcal P_q)$, then the whole expression of L.H.S will be higher and hence Equation \ref{pm_benefit_compare} will be satisfied.


\vspace{0.1cm}
\noindent
\textbf{Case 8: $\Delta \mathcal P_k$ $<$ $\Delta \mathcal P_q$, $\mathcal P_k + \Delta \mathcal P_k$ $<$ $\mathcal P_q + \Delta \mathcal P_q$, and $m_1 < m_2$}. Let us consider Equation \ref{pm_benefit_compare} and compare the left hand side of the equation with the right hand side. After simplifying some further steps, we can derive that the condition in which the left hand side will be higher than the right hand side is as follows: $\nu_k(\mathcal P_k+\Delta \mathcal P_k)\Delta \mathcal P_q > \nu_q(\mathcal P_q+\Delta \mathcal P_q)\Delta \mathcal P_k$. According to the assumption of this case, $\Delta \mathcal P_q$ value is higher than the value of $\Delta \mathcal P_k$. This implies that, if the condition $\nu_k(\mathcal P_k + \Delta \mathcal P_k)$ $>$ $\nu_q(\mathcal P_q+\Delta \mathcal P_q)$ holds, then Equation \ref{pm_benefit_compare} will be satisfied.

The above proofs (i.e., the proofs of Cases 5-8) will also hold for the scenarios where $m_1>m_2$, due to the symmetric nature of the assumptions. Based on the proofs of Cases 1-8, we can conclude that given two triples ($o_k$, $R^l_m$, $f^l_n$) and ($o_q$, $R^s_t$, $f^s_v$), if the condition of $\nu_k(\mathcal P_k + \Delta \mathcal P_k)$ $>$ $\nu_q(\mathcal P_q+\Delta \mathcal P_q)$ is satisfied then the first triple will have higher benefit value compared to the second triple.

Based on the above theorem, PIQUE calculates a {\em relative} benefit value for each triple ($o_k$, $R^i_l$, $f^i_m$) in epoch $w$ as follows:
\begin{equation}\label{eqn:benefit}
\vspace{-0.2cm}
Benefit(o_k, R^i_l, f^i_m)  = \frac{\mathcal P_k(\mathcal P_k+\Delta \mathcal P_k)}{c^{i}_{m}} 
\end{equation}

After calculating the relative benefit values for the triples in $\mathsf{TS}_w$, we pass them to the triple selection step.

\subsection{Selection of Triples}
\label{subsection:triple selection}

This step chooses the list of triples that should be included in $\mathsf{EP}_w$. We compare between the benefit values of the triples in $\mathsf{TS}_w$ and generate a plan that consists of the triples with the highest benefit values. To this end, PIQUE maintains a priority queue (denoted as $PQ$) of the triples in $\mathsf{TS}_w$ that orders them based on their benefit values. Note that the same priority queue is maintained over the different epochs; it is created from scratch in the first epoch and then updated efficiently in the subsequent epochs.

In the plan execution phase, PIQUE retrieves one triple at a time from $PQ$ and then executes it. This process continues until the allotted time for the epoch is consumed.

\section{Disk-Based version of PIQUE}\label{sect:disk_based_approach}
So far, we have implicitly assumed that objects fit in memory. In this section, we present a high-level overview of how PIQUE deals with the case where the objects that require tagging can not fit in memory. A detailed description of this extension of PIQUE is provided in Appendix \ref{sect:disk_based_approach}. In Section \ref{exp:disk_based_approach}, we present the experimental results related to this case.

The main idea is to divide the objects in $O$ equally into a set of blocks. Initially all blocks are stored on disk. At any instance of time, PIQUE maintains two priority queues: $PQ^m$ for the in-memory blocks and $PQ^d$ for the disk-resident blocks. The blocks in $PQ^m$ (resp. $PQ^d$) are sorted in an increasing (resp. decreasing) order based on their {\em block benefit} values, which will be explained later. 

The execution plan in this case consists of two parameters: a number of blocks $b$ and a set of triples $S$. During the plan execution phase, PIQUE writes the first $x$ blocks in $PQ^m$ to disk, loads the first $x$ blocks in $PQ^d$ in memory, and then executes the triples in $S$ in the remaining time of the epoch\footnote{In the first few epochs, there can be room in memory for additional blocks. In this case, the number of blocks to be written to disk does not need to be equal to the number of blocks to be loaded in memory.}.



The plan generation phase in this case consists of these five steps: candidate set selection, generation of triples, benefit estimation of triples, benefit estimation of blocks, and plan generation. The first three steps are identical to those in the in-memory case (Sections \ref{section:candidate_set_selection}, \ref{triple generation}, and \ref{BenefitEstimation}). Note that, in these three steps, PIQUE considers all objects in $O$; i.e., not only the in-memory objects. That is, the set $\mathsf{CS}_w$ can consist of disk-resident objects, the set $\mathsf{TS}_w$ can consist of triples that correspond to disk-resident objects, and so on. 

In the fourth step, we compute for each block its benefit value which is defined as the total benefit of the triples in $\mathsf{TS}_w$ that correspond to the objects in that block. Note that we only need to update the benefit value of blocks from which objects were executed in the previous epoch. Finally, the fifth step enumerates a small number of alternative plans (each with a different value of $b$) and chooses the plan with the highest benefit value, defined as the total benefit of the triples in $S$.

\section{Experimental Evaluation}
\label{sect:experiments}

In this section, we empirically evaluate PIQUE using three real datasets. We compare PIQUE with three baseline algorithms using two different quality metrics. 



\subsection{Experimental Setup}\label{sect:experimental setup}

\noindent
\textbf{Datasets.} We consider two image and one Twitter datasets. 

\begin{itemize}[leftmargin=*]
\itemsep-0.1em
\item \textbf{MUCT Dataset} \cite{MUCTDataset} consists of $3,755$ face images of people from various age groups and ethnicities. It consists of types of tags {\em Gender} and {\em Expression} and five precise attributes {\em Timestamp}, {\em PersonId}, {\em SessionId}, {\em CameraId}, and {\em CameraAngle}. This small dataset has been selected so that all approaches can be executed to completion. We divided the dataset into three parts: training, validation, and testing. The sizes of these parts are $850$, $850$, and $2,055$ images, respectively. The first two parts were used to learn some parameters (explained next) whereas the testing part was used in our experiments. 




\item \textbf{CMU Multi-PIE Dataset} \cite{multi_pie_dataset} contains more than 750k face images of $337$ people recorded over the span of five months. This dataset consists of three types of tags: {\em Gender}, {\em Expression}, and {\em Age} and five precise attributes {\em Timestamp}, {\em PersonID}, {\em SessionId}, {\em CameraId}, and {\em CameraAngle}. From this dataset, we obtained a subset of $500$k images. We used $5$k images for each of the training and validation parts and used the remaining $490$k images in our experiments. 


\item \textbf{Stanford Twitter Sentiment (STS) Corpus}  \cite{Twitter_Sentiment_Dataset} consists of 1.6M tweets collected using Twitter API during a period of about 3 months in 2009. This dataset consists of one tag type {\em Sentiment} and three precise attributes {\em Timestamp}, {\em Location} and {\em User}. From this dataset, we used $50$k tweets for each of the training and validation parts and used the remaining tweets in our experiments. 
\end{itemize}

\noindent
\textbf{Queries.} The general structure of our experiments queries follows the format described in Section~\ref{sect:problemDefinition}. Table~\ref{table:Query} presents the five query templates used in our experiments. For example, $Q_1$ has one predicate on tag and five precise predicates\footnote{Note that in Table~\ref{table:Query}, we have mentioned both the precise predicates and the predicates on tag together but the PIQUE operator only contains the predicates on the tags.}. Since precise predicates are evaluated before the predicates on tags, we vary the attribute values in the precise predicates (i.e., the $x_i$ values in Table~\ref{table:Query}) to determine the {\em selectivity} of the queries, which is the percentage of the objects in the dataset that require tagging. The reported results and execution time of each query are averaged over 40 runs. 


\begin{table}[t]\label{QueryTable}
\footnotesize
\centering
\begin{center}
 \begin{tabular}{|>{\centering\arraybackslash}p{5mm}|>{\centering\arraybackslash}p{68mm}|} 
 \hline
 \textbf{Id} & \textbf{Query} \\
 \hline
 $Q1$ & $\textit{Gender} = {\tt Male} \: \Lambda \:  \textit{Timestamp} = x_1 \:$
 $\: \Lambda \: \textit{PersonID} = x_2$ 
 $\: \Lambda \: \textit{ SessionID} = x_3$ 
  $\: \Lambda \: \textit{CameraID} = x_4$  
  $\: \Lambda \: \textit{CameraAngle} = x_5$ \\
 \hline
 $Q2$ & $\textit{Sentiment} = {\tt Positive} \: \Lambda \: \textit{Timestamp} = x_1$
 $\: \Lambda \: \textit{Location} = x_2 \: \Lambda \: \textit{User} = x_3 $\\
 \hline
 $Q3$ & $\textit{Gender} = {\tt Male} \: \Lambda \: \textit{Expression} = {\tt Smile} \: \Lambda  \: \textit{Timestamp} = x_1 \: \Lambda \: \textit{PersonID} = x_2 \: \Lambda \: \textit{SessionID} = x_3$ 
  $\: \Lambda \: \textit{CameraID} = x_4$  
  $\: \Lambda \: \textit{CameraAngle} = x_5$ \\
 \hline
 $Q4$ & $\textit{Gender} = {\tt Male} \: \Lambda  \:  \textit{Age} = {\tt 30} \: \Lambda \: \textit{Timestamp} = x_1$
 $\: \Lambda \: \textit{PersonID} = x_2$ 
 $\: \Lambda \: \textit{ SessionID} = x_3$ 
  $\: \Lambda \: \textit{CameraID} = x_4$  
  $\: \Lambda \: \textit{CameraAngle} = x_5$ \\
 \hline
 $Q5$ & $\textit{Gender} = {\tt Male}  \: \Lambda \:  \textit{Expression} = {\tt Smile} \: \Lambda \: \textit{Age} = {\tt 30} \: \Lambda \:  \textit{Timestamp} = x_1$
 $\: \Lambda \: \textit{PersonID} = x_2$ 
 $\: \Lambda \: \textit{SessionID} = x_3$ 
  $\: \Lambda \: \textit{CameraID} = x_4$  
  $\: \Lambda \: \textit{CameraAngle} = x_5$
   \\
 \hline
\end{tabular}
\vspace{-0.18cm}
\caption{Queries used.} \label{QueryTable}
\label{table:Query}
\end{center}
\end{table}

\vspace{0.1cm}
\noindent
\textbf{Enrichment Functions.} In the image datasets, we considered these classifiers as the enrichment functions for determining the tags: Decision Tree, Gaussian Naive-Bayes, Random Forest, and Multi-layer Perceptron. In the Twitter dataset, we used the Gaussian Naive-Bayes, k-Nearest-Neighbors, Support Vector Machine, and Decision Tree classifiers. 
We chose these classifiers because they are widely used for  gender classification~\cite{gender_classification}, facial expression classification~\cite{facial_expression_recognition_survey}, age classification~\cite{age_detection_survey} on images and sentiment analysis of tweets~\cite{sentiment_analysis_classifiers}. 

We used the python implementation of those classifiers available in the scikit-learn library \cite{scikit-learn} and train them using the training parts of the datasets. We calibrate the probability outputs of the enrichment functions using appropriate calibration mechanisms. We used the isotonic regression model \cite{Zadrozny:2002:TCS:775047.775151} to calibrate the Gaussian Naive-Bayes classifier based enrichment functions and the Platt's sigmoid model \cite{Platt99probabilisticoutputs} to calibrate the remaining enrichment functions.


\vspace{0.1cm}
\noindent
\textbf{Preprocessing Step.} We execute a feature extraction code (Histogram of Oriented Gradients features from images and related keywords extraction from tweets) on all the objects in the dataset prior to the arrival of the queries. Note that this preprocessing step mimics the work that would be done on streaming data upon its ingestion.

\vspace{0.1cm}
\noindent
\textbf{Parameter Learning.} We discuss below how the parameters of PIQUE are learned.

\vspace{0.1cm}
\noindent

\begin{itemize}[leftmargin=*]
\item \textbf{Enrichment Function Quality and Cost.}
 The quality and cost of enrichment functions are learned using the validation parts of the datasets as follows. For each function $f^i_j$, we execute it on all objects in the validation part and store the output probability and execution cost of each function call. We then use those output probabilities to plot the ROC curve \cite{BRADLEY19971145} of $f^i_j$  and set $q^i_j$ to the area under this curve. The $c^i_j$ value is set to be the average execution time of $f^i_j$.

\item \textbf{Decision Table.}
Note that the enrichment functions were run on the objects of the validation part in different orders to ensure that each state (the third column of Table~\ref{DecisionTable}) is associated with a number of objects. 
Given a state and a tag value $t_j \in T_i$, we group the objects of that state into buckets depending on their uncertainty values such that each bucket corresponds to a single range (the fourth column of Table~\ref{DecisionTable}). Next, for each range, we first determine which enrichment function from the remaining functions reduces the uncertainty value of the bucket objects the most. This function is stored in the fifth column of Table~\ref{DecisionTable}. Finally, we set the $\Delta$ uncertainty of that range to the average uncertainty reduction obtained from running that function on all objects in that bucket. This value is stored in the sixth column of Table~\ref{DecisionTable}. This process is then repeated for each pair of a state and a tag value.

\item \textbf{Combine Function.} The combine function $M_i$ of each tag type $T_i$ is implemented using a weighted summation model where the weight of a enrichment function $f^i_j$ is set to the quality of that function (i.e., $q^i_j$) divided by the total quality of all the enrichment functions in $F_i$.
\end{itemize}

Note that among all these parameters, learning only the quality of the enrichment functions requires labeled validation datasets. The values of the other parameters can be learned online during the first few epochs of the execution and can be continuously adjusted as PIQUE proceeds forward.

\setlength\tabcolsep{1.5pt} 
 \begin{table}[t]
 \setlength\extrarowheight{2pt}
\footnotesize
\centering
\begin{center}
 \begin{tabular}{|>{\centering\arraybackslash}p{13mm}|>{\centering\arraybackslash}p{16mm}|>{\centering\arraybackslash}p{12mm}|>{\centering\arraybackslash}p{12mm}|>{\centering\arraybackslash}p{12mm}|>{\centering\arraybackslash}p{13mm}|}
 \hline
  \textbf{Dataset}  & 
  \textbf{Selectivity} &
  \textbf{(a)}   & \textbf{(b)}  & \textbf{(c)}  & \textbf{Total} \\ 

 \hline
 \multirow{4}{*}{\begin{turn}{90} MUCT \end{turn}} & 2.5\% & 0.00005 & 0.165 &0.00006 & 0.16512 \\ 
     & 5\% & 0.00009 & 0.481 & 0.00011 & 0.48105 \\
     & 7.5\% & 0.00013 & 0.630 & 0.00027 & 0.63108 \\
     & 10\% & 0.00019 & 0.860 & 0.00032 & 0.86176 \\
 \hline
  \multirow{4}{*}{\begin{turn}{90} Twitter \end{turn}} & 0.1\% & 0.05230 & 0.433 & 0.00205 & 0.48698 \\ 
     & 0.2\% & 0.05337 & 0.923 & 0.00388 & 0.97721 \\
     & 0.4\% & 0.05403 & 1.729 & 0.00856 & 1.79193 \\
     & 1\% & 0.05816 & 3.965 & 0.06007 & 4.08398 \\
 \hline
 \multirow{4}{*}{\begin{turn}{90} Multi-PIE \end{turn} } & 0.05\% & 0.00022 & 1.850 & 0.00033 & 1.85008 \\ 
      & 0.1\% & 0.00039 & 2.283 & 0.00082 &  2.28424 \\
      & 0.15\% & 0.00055 & 3.776 & 0.00098 & 3.77755\\
      & 0.2\% & 0.00095 & 5.547 & 0.00110 & 5.54870 \\
 \hline
 \end{tabular}
 \vspace{-0.6em}
\caption{Time taken in the three sub-steps of the initialization step (in seconds). Column \textbf{(a)} shows the object load time, column \textbf{(b)} shows the initial function evaluation time and column \textbf{(c)} shows the data structure creation time.}\label{table:Initialization}
\end{center}
\end{table}
\setlength\tabcolsep{6pt}

\vspace{0.1cm}
\noindent
\textbf{Approaches.}
We compare our approach with three baseline approaches. We used the in-memory versions of these approaches in all experiments except for Experiment 8 wherein the disk-based versions are used.



\begin{itemize}[leftmargin=*]
\item \textbf{Function-Based Ordered-Object Approach -- Baseline 1}: An enrichment function is chosen first and then executed on all the objects. The tagging functions are chosen in the decreasing order of their ($q^i_j/c^i_j$) values. Objects were chosen based on their expression satisfiability probability values at the beginning of execution, starting with the object with highest expression satisfiability probability value.

\item \textbf{Object-Based Ordered-Function Approach -- Baseline 2}: We choose an object \emph{$o_i$} first, iterate over all predicates in $Q$, and for each predicate, execute all the corresponding enrichment functions on that object. The objects were chosen based on their expression satisfiability probability values, starting with the object with the highest expression satisfiability probability value.
 
\item \textbf{Random Epoch-Based Approach -- Baseline 3}: The same as our approach with the exception that the execution plan is generated randomly. That is, at the beginning of each epoch, this approach randomly chooses a set of objects and for each of those objects, constructs one triple that consists of the object, a randomly selected predicate, and an enrichment function chosen randomly from the set of enrichment functions that have been applied on the object.
 \vspace{-0.2cm}
\end{itemize}
\noindent



\vspace{0.2cm}
\noindent
\textbf{Initialization Step.} At the beginning of the query execution time, each of the four approaches performs an initialization step that consists of three sub-steps: 1)~Load all the required objects in memory; 2)~Execute a seed enrichment function $f^i_j$ per each tag type $T_i$ on the objects, where $f^i_j$ is chosen to be the function that has the highest $q^i_j/c^i_j$ value among all functions in $F_i$: and 3)~The following data structures of the objects are created: a hash map for the state values, a hash map for the predicate probability values, and a list for the uncertainty values. Note that, in the plots presented in the following sections, we have omitted the time taken by this initialization step as it is  the same for all approaches. We report this time value, however, for the in-memory versions of the approaches in Table \ref{table:Initialization}. 

\begin{figure}[t]
\centering
\begin{tikzpicture}
   \node(img1){\includegraphics[trim={1.1cm 0.6cm 1cm 1.32cm},clip,width=0.20\textwidth]{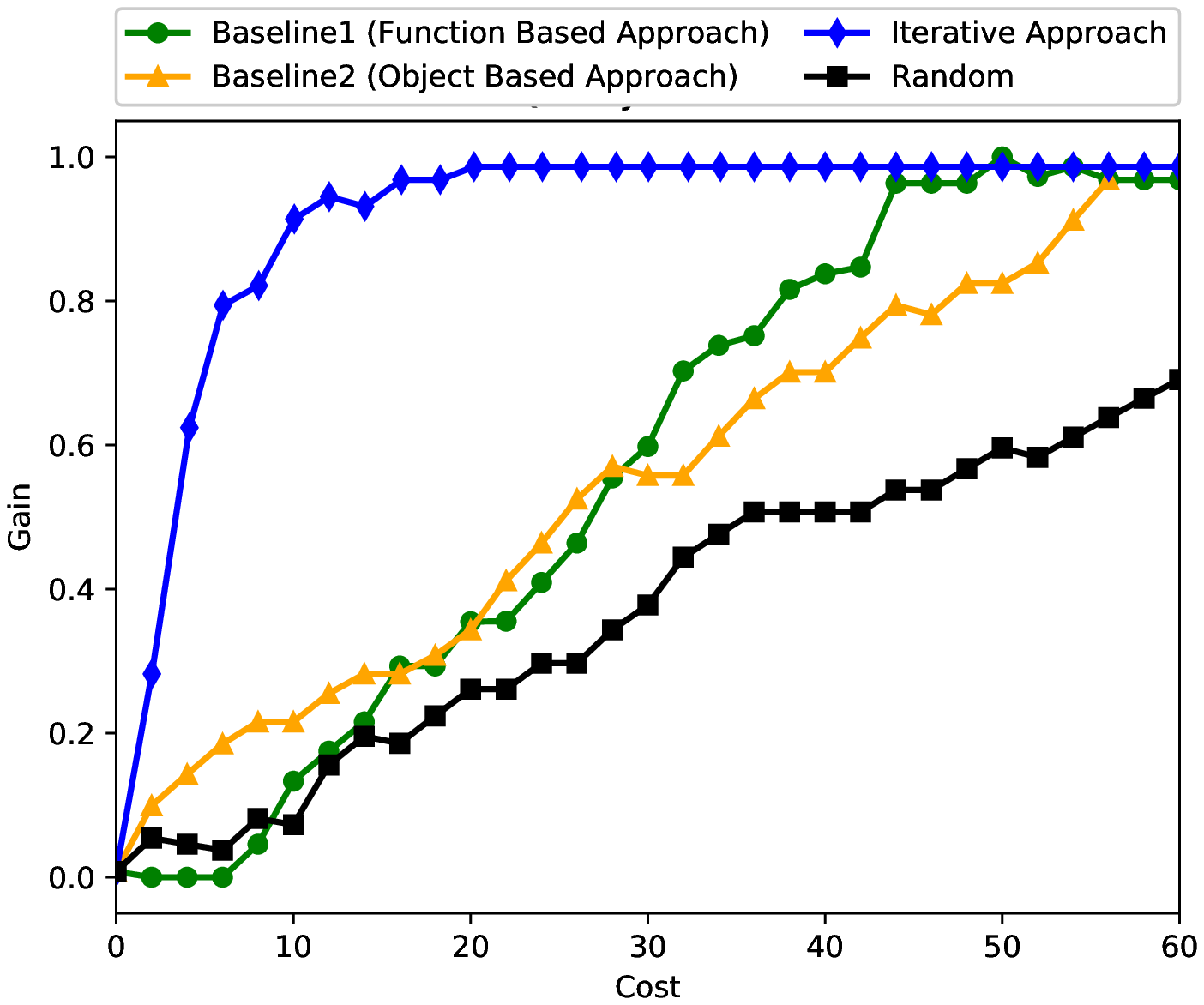}};;
  \node[below=of img1,node distance=0cm, yshift=1.2cm] (x-label3){\small Time (Seconds)};
  \node[left=of img1, node distance=0cm, rotate=90, anchor=center,yshift=-1cm] {\small Quality (Gain)};
   \node[right=of img1,xshift=-1cm] (img2)  {\includegraphics[trim={1.1cm 0.6cm 1cm 1.32cm},clip,width=0.20\textwidth]{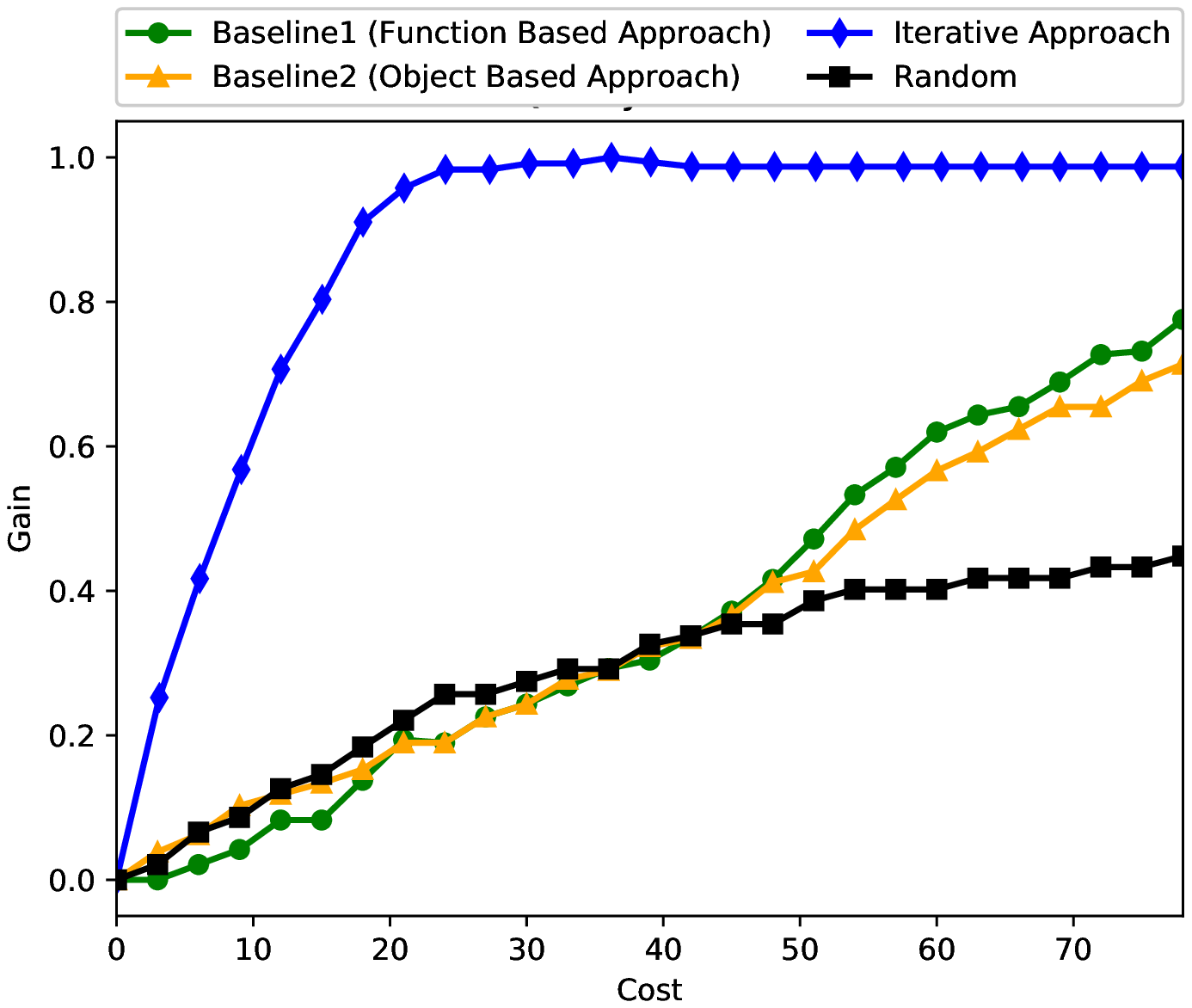}};
  \node[below=of img2, node distance=0cm, yshift=1.2cm](x-label4) {\small Time (Seconds)};
  \node[left=of img2, node distance=0cm, rotate=90, anchor=center,yshift=-1cm] {\small Quality (Gain)};
 \node[below=of x-label3,  node distance=0cm, yshift=1.1cm]{\small (a) 7.5\% selectivity};
  \node[below=of x-label4,  node distance=0cm, yshift=1.1cm]{\small (b) 10\% selectivity};
  \node[anchor=center, node distance=0cm,yshift= 0cm, xshift =-1cm, above = of img1](legend1){\color{blue} \small {$\blacklozenge$}  PIQUE};
    \node[anchor=center, node distance=0cm, right = of legend1](legend2){ \color{ao(english)} \small  { $\bullet$} Baseline 1
    };
   \node[node distance=0cm, anchor=center, right = of legend2] (legend3)  {\color{amber} \small  { $\blacktriangle$} Baseline 2 };
    \node[node distance=0cm, anchor=center, right = of legend3] (legend4)  {\small { $\blacksquare$} Baseline 3};

\end{tikzpicture}
\vspace{-0.4cm}
\caption{Comparison of gain (MUCT dataset) for Q1.}\label{fig:VariationOfGainExperimentMuct}
\end{figure}

\vspace{0.1cm}
\noindent
\textbf{Quality Metrics.}
We consider two metrics for comparing the various approaches: the \emph{$F_1$} measure and the \emph{gain} of the answer set. The gain at a time instant $v_i$ is defined as $x/y$ where $x$ is the $F_1$ value of the answer set at time $v_i$ minus the $F_1$ value of the answer set at the beginning of the query execution and $y$ is the maximum $F_1$ value that can be achieved by the approach minus the $F_1$ value of the answer set at the beginning of the query execution. We compute the $F_1$ measure of the answer sets using the available ground truth of the datasets.




\vspace{0.1cm}
\noindent
\textbf{Progressiveness Metric.}
The progressiveness of each approach is measured by the discrete sampling function of Equation~\ref{progressiveness-metric}. Following~\cite{progressive-duplicate-detection}, in each experiment, we set the time instant value $v_{|V|}$ to the minimum execution time of the involved approaches and then divide the range $[0, v_{|V|}]$ into {\em ten} equal-sized intervals. For example, in Figure~\ref{fig:VariationOfGainExperimentMuct} (a), the total execution time of our approach and the three baseline approaches are $60$, $65$, $66$, and $68$ seconds, respectively. In this case, the time vector is set to $V=\{v_0 = 0$ seconds$, v_1 = 6$ seconds$, v_2 = 12$ seconds$, \dots, v_{10} = 60$ seconds$\}$. We also used this following most widely used weighting function: $W(v_i) =max(1- ((v_i-1)/v_{|V|}),0)$. 

\setlength\tabcolsep{1.5pt} 
\begin{table}[t]\label{QueryTable}
\footnotesize
\centering
\begin{center}
 \begin{tabular}{|>{\centering\arraybackslash}p{8mm}|>{\centering\arraybackslash}p{10mm}|>{\centering\arraybackslash}p{7mm}|>{\centering\arraybackslash}p{7mm}|>{\centering\arraybackslash}p{7mm}|>{\centering\arraybackslash}p{8mm}|>{\centering\arraybackslash}p{10mm}|>{\centering\arraybackslash}p{7mm}|>{\centering\arraybackslash}p{7mm}|>{\centering\arraybackslash}p{7mm}|} 
  \hline
 \textbf{Fig.}  & \textbf{PIQUE} & \textbf{B1} & \textbf{B2}  & \textbf{B3}  & \textbf{Fig.}  & \textbf{PIQUE} & \textbf{B1} & \textbf{B2}  & \textbf{B3}\\
 \hline
 \ref{fig:VariationOfGainExperimentMuct} (a) & 0.86 & 0.48 & 0.47 & 0.42 & \ref{fig:VariationOfGainExperimentMuct} (b) & 0.90 & 0.48 &0.47 & 0.41 \\

  \hline
 \ref{fig:f1measureVariationExperiment} (a) & 0.54 & 0.33 & 0.30 & 0.23 &\ref{fig:f1measureVariationExperiment} (b) & 0.56 & 0.23 & 0.22 & 0.20 \\
 \hline
 \ref{fig:VariationOfGainExperimentMultiPie} (a) & 0.95 & 0.48 & 0.47 & 0.36 & \ref{fig:VariationOfGainExperimentMultiPie} (b)  & 0.93  & 0.47 & 0.25 & 0.09  \\
 \hline
  \ref{fig:VariationOfGainExperimentTwitter} (a) & 0.63 & 0.45 & 0.23 & 0.28 &  \ref{fig:VariationOfGainExperimentTwitter} (b)  & 0.62 & 0.41 & 0.22 &  0.29 \\
 \hline
 \ref{fig:VariationOfQuality(MultiFeature)} (a) & 0.91 & 0.41 &  0.25 & 0.09  & \ref{fig:VariationOfQuality(MultiFeature)} (b) & 0.88 & 0.39 &  0.26 & 0.09 \\
 \hline
 \ref{fig:VariationOfQuality2(MultiFeature)} (a)  & 0.93 & 0.23 & 0.09 &  0.02 & \ref{fig:VariationOfQuality2(MultiFeature)} (b)  & 0.86 & 0.07 & 0.07 &  0.02\\
 \hline
 \ref{fig:CachingVariationOfF1Experiment} (a) & 0.41 & 0.08 & 0.06 & 0.08 & \ref{fig:CachingVariationOfF1Experiment} (b) & 0.23 & 0.09 & 0.07 & 0.05 \\
 \hline
 \ref{fig:VariationOfQualityDiskBased} (a) & 0.91 & 0.17 & 0.10 &  0.08 & \ref{fig:VariationOfQualityDiskBased} (b) & 0.92 & 0.17 & 0.13 & 0.08 \\
 \hline
 \end{tabular}
\vspace{-0.15cm}
\caption{Progressiveness Scores.}\label{ProgressiveScoreTable}
\label{table:Query}
\end{center}
\vspace{-0.3cm}
\end{table}
\setlength\tabcolsep{6pt}

\subsection{Experimental Results}\label{sect:experimental_result}
\noindent
\textbf{Experiment 1 (Trade-off between Quality and Cost).}
We compare PIQUE with the baseline approaches. The epoch time for PIQUE and Baseline 3 is set according to the experiments that will be presented in Experiment 2. The results for the MUCT dataset {\em w.r.t.} the gain and $F_1$-measure of the answer set are shown in Figure~\ref{fig:VariationOfGainExperimentMuct} and Figure~\ref{fig:f1measureVariationExperiment}, respectively. The results for the Multi-Pie and Twitter datasets {\em w.r.t.} the gain metric are shown in Figure~\ref{fig:VariationOfGainExperimentMultiPie} and Figure~\ref{fig:VariationOfGainExperimentTwitter}, respectively\footnote{{Due to space limitations, we only include the plots for two selectivity values of the queries in Figures~\ref{fig:f1measureVariationExperiment}, \ref{fig:VariationOfGainExperimentMultiPie}, and~\ref{fig:VariationOfGainExperimentTwitter}.}}. 
The progressiveness scores (Equation~\ref{progressiveness-metric}) of all approaches involved in these figures and all other figures are presented in Table \ref{ProgressiveScoreTable}.

PIQUE outperforms the baseline approaches significantly for the different selectivity values in all three datasets. In PIQUE, the answer set achieves a very high quality within the first few epochs of the execution. (Note that each epoch is noted by a point in each plot.) Figure~\ref{fig:VariationOfGainExperimentMuct} shows that PIQUE achieves a very high quality gain (e.g., 0.98) within the first 28 and 36 seconds of the execution for the selectivity values $7.5\%$ and $10\%$, respectively. PIQUE achieves a high rate of quality improvement in each epoch. 
For example, in the second epoch of the query execution in Figure~\ref{fig:VariationOfGainExperimentMuct} (a), PIQUE improves the quality of the answer set from $0.25$ to $0.63$ whereas Baseline 1 improves it from $0.10$ to $0.16$, Baseline 2 improves it from $0.01$ to $0.02$, and Baseline 3 improves it from $0.04$ to $0.05$.
 
The reason behind this performance gap is that unlike the baseline approaches, PIQUE follows an adaptive strategy that monitors and reassesses the progress at the beginning of each epoch to identify and execute the triples that are expected to yield the highest improvement in quality.


\begin{figure}[t]
\centering
\begin{tikzpicture}
  \node (img1)  {\includegraphics[trim={1.1cm 0.6cm 1cm 1.32cm},clip,width=0.20\textwidth]{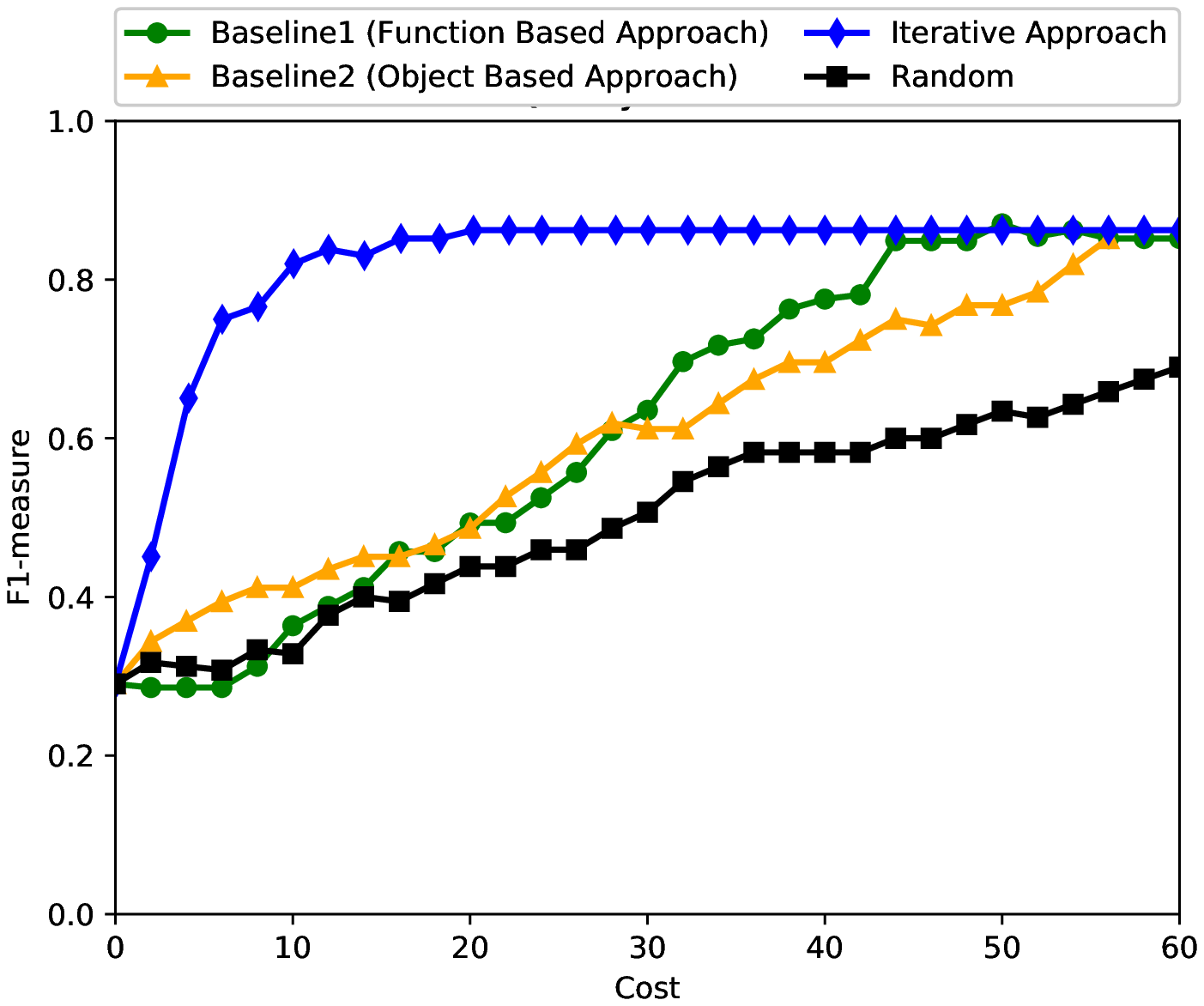}};
  
  \node[below=of img1, node distance=0cm, yshift=1.2cm,font=\color{black}](x-label1) {\small Time (Seconds)};
  \node[left=of img1, node distance=0cm, rotate=90, anchor=center,yshift=-1cm,font=\color{black}] {\small $F_1$  measure};
  \hfill;
  
  \node[right=of img1,xshift=-1cm] (img2)  {\includegraphics[trim={1.1cm 0.6cm 1cm 1.32cm},clip,width=0.20\textwidth]{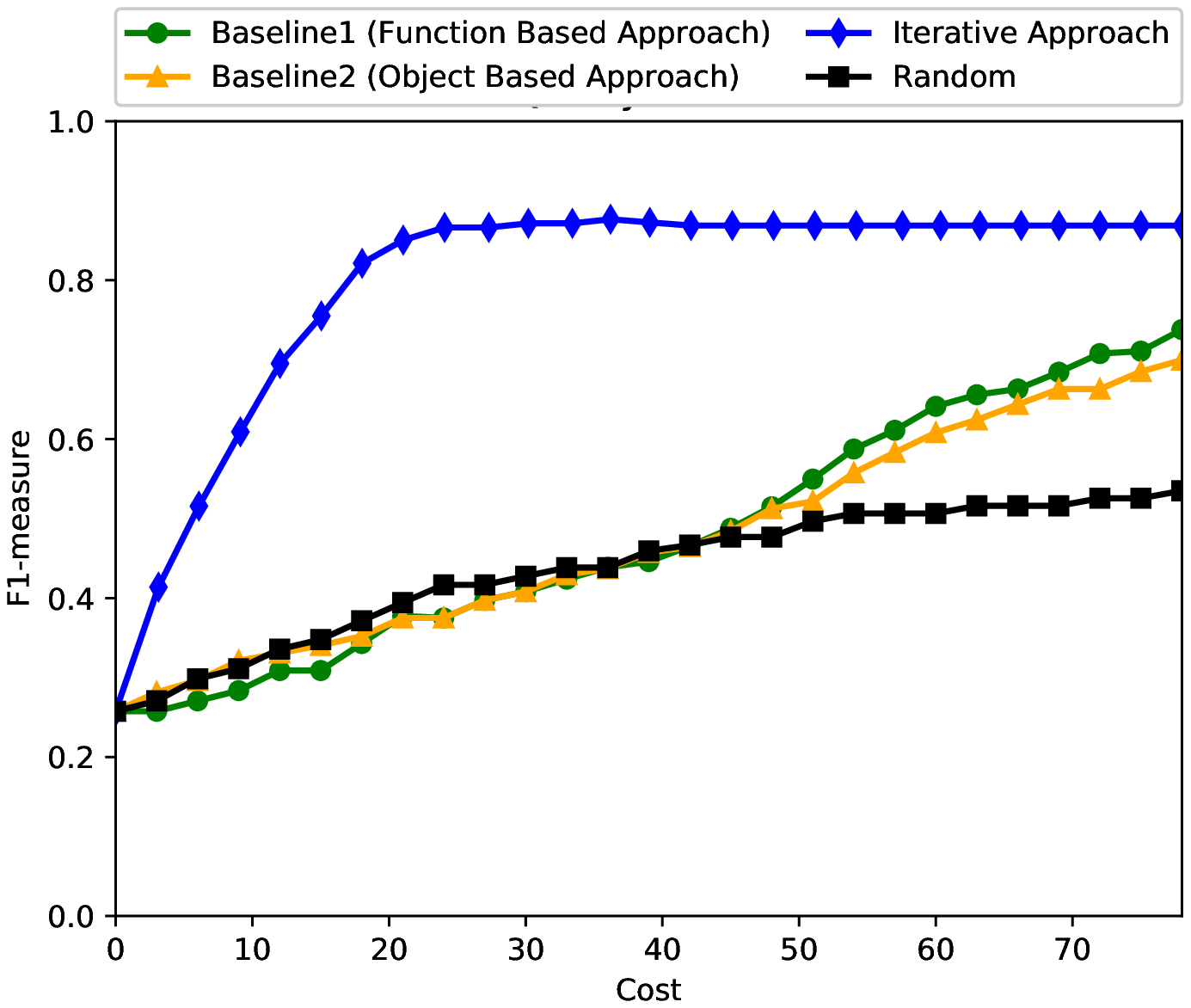}};
  
  \node[below=of img2, node distance=0cm, yshift=1.2cm] (x-label2){ \small Time (Seconds)};
  
  \node[left=of img2, node distance=0cm, rotate=90, anchor=center,yshift=-1cm] (legend1){\small $F_1$  measure};  
  \node[anchor=center, node distance=0cm,yshift= -0.1cm, xshift = -1cm , above = of img1](legend1){\color{blue} \small  {$\blacklozenge$}  PIQUE };
    \node[anchor=center, node distance=0cm, right = of legend1](legend2){\color{ao(english)} \small { $\bullet$} Baseline1 };
    
   \node[node distance=0cm, anchor=center, right = of legend2] (legend3)  {\color{amber} \small { $\blacktriangle$} Baseline2};
   \node[node distance=0cm, anchor=center, right = of legend3] (legend4)  {\small { $\blacksquare$} Baseline3};
   
   \node[below=of x-label1,  node distance=0cm, yshift=1.2cm]{\small (a) 7.5\% selectivity};
  \node[below=of x-label2,  node distance=0cm, yshift=1.2cm,xshift=-0.2cm]{\small (b) 10\% selectivity};

\end{tikzpicture}
\vspace{-0.4cm}
\caption{Comparison of $F_1$ measure (MUCT) for Q1.}
\label{fig:f1measureVariationExperiment}
\end{figure}


\begin{figure}[t]
\centering
\vspace{-0.4cm}
\begin{tikzpicture}
  \node (img1)  {\includegraphics[trim={1.1cm 0.6cm 1cm 1.32cm},clip,width=0.20\textwidth]{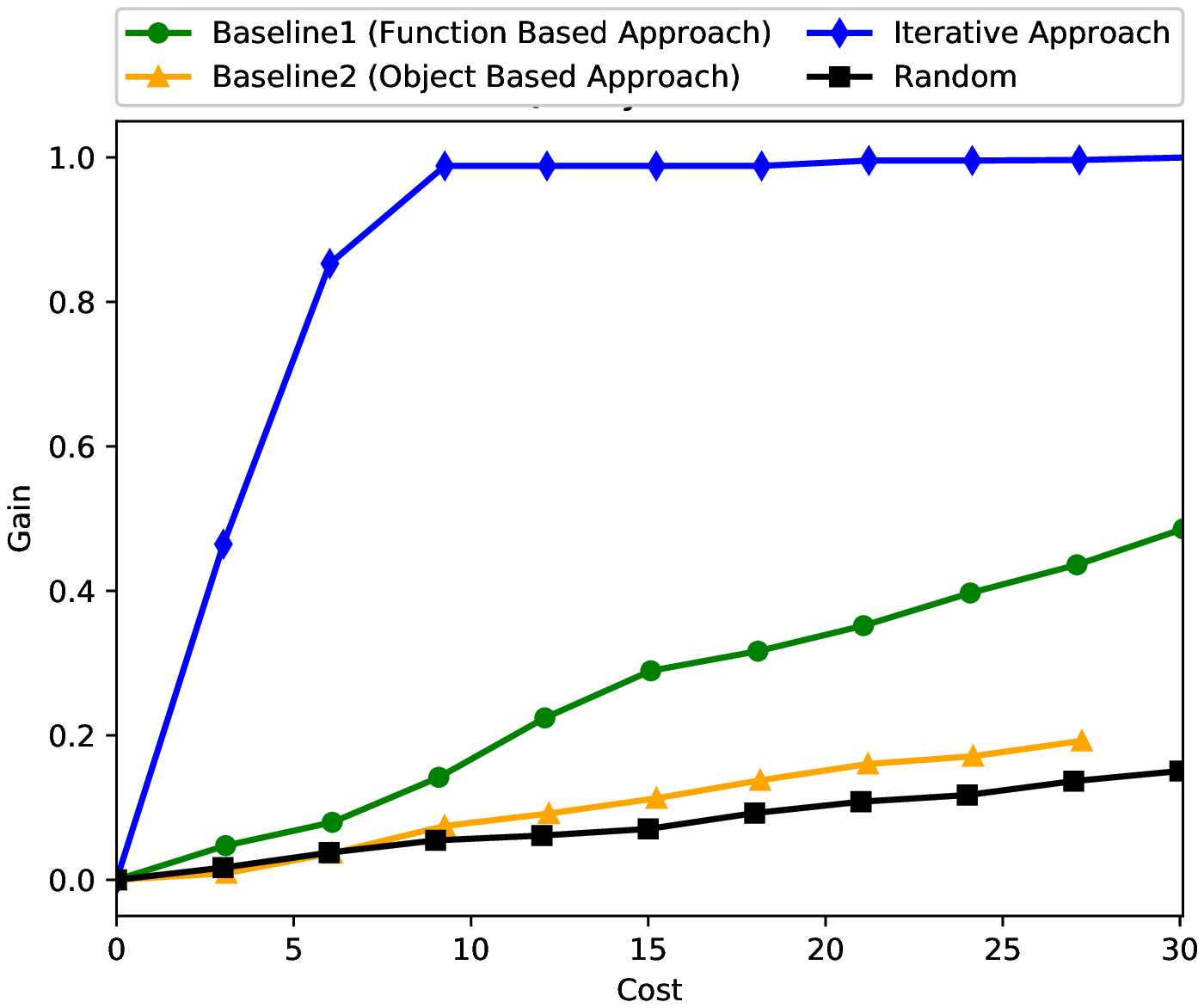}};
  
  \node[below=of img1, node distance=0cm, yshift=1.2cm,font=\color{black}] {\small Time (Seconds)};
  \node[left=of img1, node distance=0cm, rotate=90, anchor=center,yshift=-0.9cm,font=\color{black}] {\small Quality (Gain)};
  \hfill;
  \node[right=of img1,xshift=-1cm] (img2)  {\includegraphics[trim={1.1cm 0.6cm 1cm 1.32cm},clip,width=0.20\textwidth]{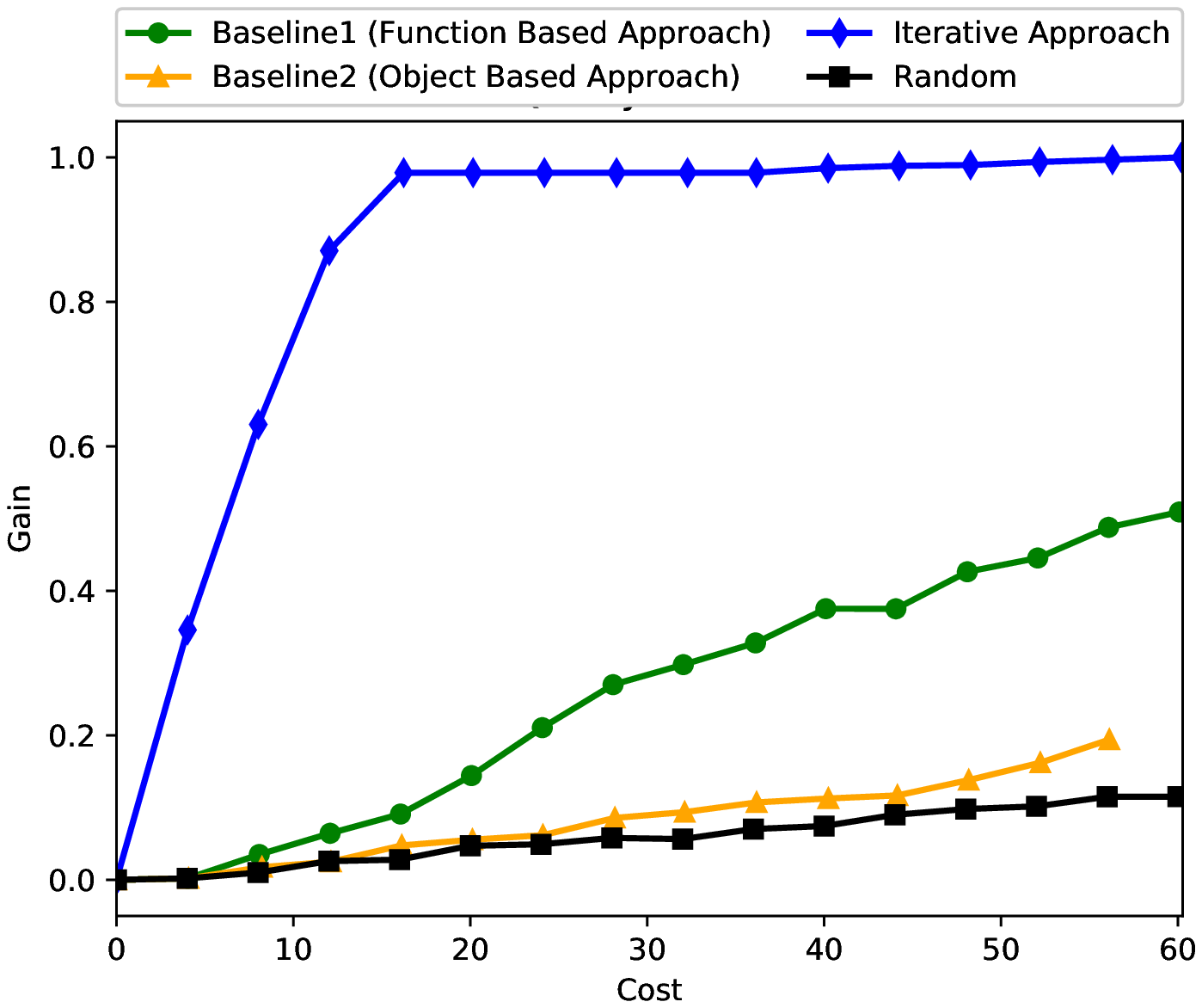}};
  
  \node[below=of img2, node distance=0cm, yshift=1.2cm] {\small Time (Seconds)};
  \node[left=of img2, node distance=0cm, rotate=90, anchor=center,yshift=-1cm] {\small Quality (Gain)};  
  \node[anchor=center, node distance=0cm, xshift = -1cm, yshift= -0.1cm, above = of img1](legend1){\color{blue} \small {$\blacklozenge$} PIQUE };
    \node[anchor=center, node distance=0cm, right = of legend1](legend2){\color{ao(english)} \small  { $\bullet$} Baseline 1 };
    
   \node[node distance=0cm, anchor=center, right = of legend2] (legend3)  {\color{amber} \small  { $\blacktriangle$}  Baseline 2};
  \node[below=of x-label1,  node distance=0cm, yshift=1.2cm]{\small (a) 0.1\% selectivity};
  \node[below=of x-label2,  node distance=0cm, yshift=1.2cm,xshift = 0.4cm]{\small (b) 0.2\% selectivity};
  \node[node distance=0cm, anchor=center, right = of legend3] (legend4)  {\small { $\blacksquare$} Baseline 3};
        
\end{tikzpicture}
\caption{Comparison of gain (Multi-PIE dataset) for Q1.}
\label{fig:VariationOfGainExperimentMultiPie}
\end{figure}


\begin{figure}[t]
\centering
\begin{tikzpicture}
  \node (img1)  {\includegraphics[trim={1.1cm 0.6cm 1cm 1.32cm},clip,width=0.20\textwidth]{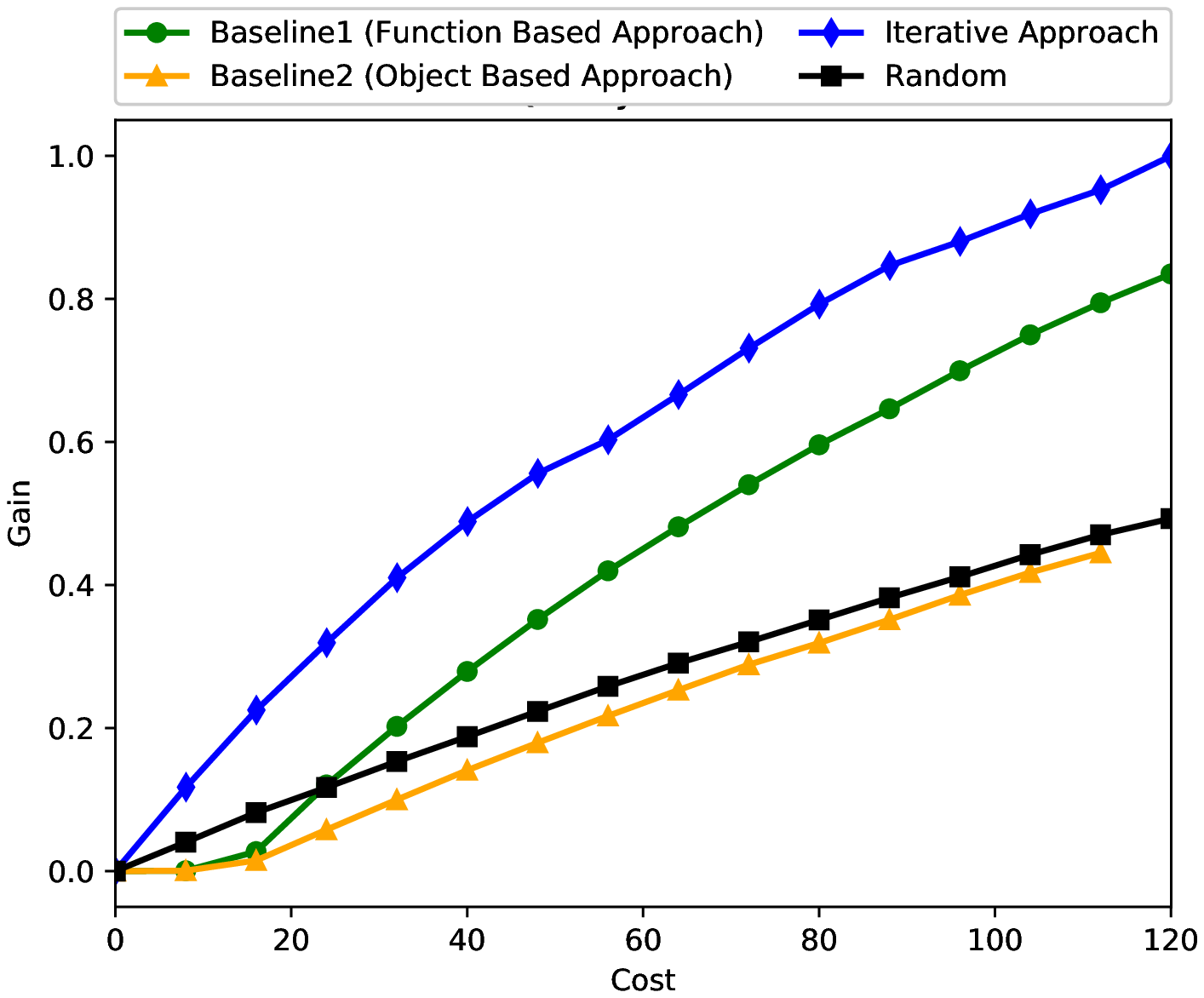}};
  
  \node[below=of img1, node distance=0cm, yshift=1.2cm,font=\color{black}] {\small Time (Seconds)};
  \node[left=of img1, node distance=0cm, rotate=90, anchor=center,yshift=-0.9cm,font=\color{black}] {\small Quality (Gain)};
  \hfill;
  \node[right=of img1,xshift=-1cm] (img2)  {\includegraphics[trim={1.1cm 0.6cm 1cm 1.32cm},clip,width=0.20\textwidth]{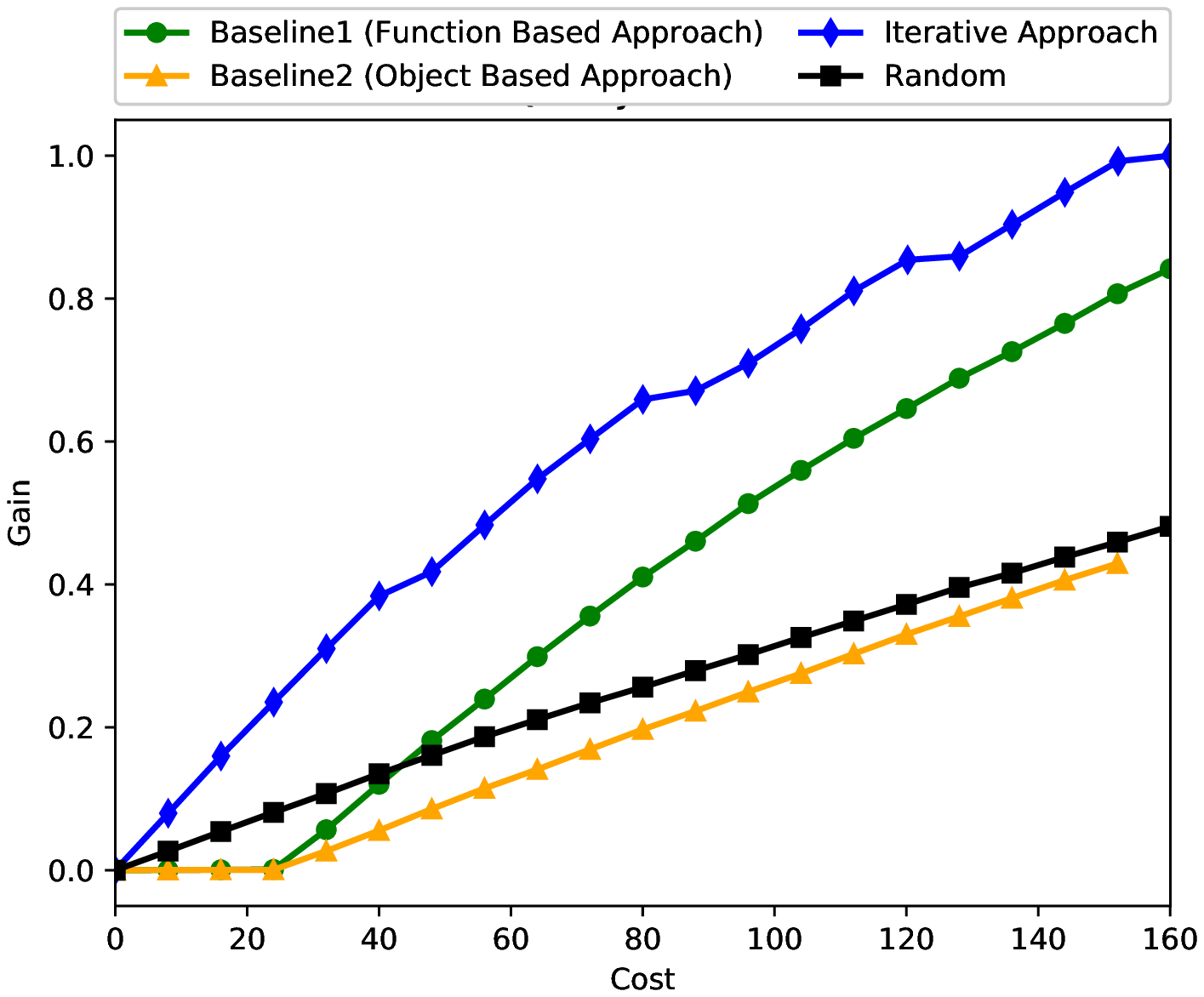}};
  

  \node[below=of img2, node distance=0cm, yshift=1.2cm] {\small Time (Seconds)};
  \node[left=of img2, node distance=0cm, rotate=90, anchor=center,yshift=-1cm] {\small Quality (Gain)};
  
  \node[anchor=center, node distance=0cm,yshift= -0.1cm, xshift = -1cm, above = of img1](legend1){\color{blue} \small {$\blacklozenge$}  PIQUE };
    \node[anchor=center, node distance=0cm, right = of legend1](legend2){ \color{ao(english)} \small  { $\bullet$} Baseline 1 };
    
   \node[node distance=0cm, anchor=center, right = of legend2] (legend3)  {\color{amber} \small  { $\blacktriangle$}  Baseline 2};
       
   \node[below=of x-label1,  node distance=0cm, yshift=1.1cm]{\small (a) 0.2\% selectivity};
  \node[below=of x-label2,  node distance=0cm, yshift=1.1cm]{\small (b) 0.4\% selectivity};
  \node[node distance=0cm, anchor=center, right = of legend3] (legend4)  {\small { $\blacksquare$} Baseline 3};
\end{tikzpicture}
\vspace{-0.4cm}
\caption{Comparison of gain (Twitter dataset) for Q2.}
\label{fig:VariationOfGainExperimentTwitter}
\end{figure}

\vspace{0.1cm}
\noindent
\textbf{Experiment 2 (Epoch Size).}\label{sect:epoch_size}  The objective of this experiment is to determine the appropriate epoch size of PIQUE for different selectivity values. (Note that Baseline 3 always uses the same epoch size set for PIQUE.) Note that, if epoch size is small, then PIQUE performs the plan generation phase more frequently as compared to when the epoch size is large. 

Given a selectivity value $x$, we first determine the average completion time of PIQUE. To this end, we choose a set of different epoch sizes $\{1, 2,  \dots, 10\}$ seconds, run PIQUE till completion using each of those sizes, and then set $T_q(x)$ (which is the average  completion time of PIQUE when the selectivity value is $x$) to $min$($T_q(x, 1)$, $T_q(x, 2)$, $\dots$, $T_q(x, 10)$), where $T_q(x,i)$ is the completion time of PIQUE when the selectivity value is $x$ and the epoch size is $i$ seconds. Then, we run PIQUE ten times using these epoch sizes \{1\% of $T_q(x)$, 2\% of $T_q(x)$, $\cdots$, 10\% of $T_q(x)$\}. For each run, we compute the progressiveness score using Equation~\ref{progressiveness-metric}.

\begin{table}[t]
\vspace{-0.2cm}
\centering
\small
\begin{center}
 \begin{tabular}{|c|c|c|c|} 
 \hline
 \textbf{Sel.} &  $\mathbf{T_q(x)}$ & \textbf{ Appropriate  Epoch} & \textbf{Completion}   \\
   \textbf{($\mathbf{x}$)}  &     \textbf{(Sec.)}  &   \textbf{ Size (Sec.)}  & \textbf{  Time (Sec.) } \\
 \hline
2.5\% & 26 & 1\% of $T_q(x)$=0.26 & 23 \\ 
 \hline
 5\% & 60 & 2\% of $T_q(x)$ = 1.2 & 56  \\
 \hline
7.5\% & 86 & 3\% of $T_q(x)$ = 2.58 & 80 \\
 \hline
10\% & 100 & 5\% of $T_q(x)$ = 5  & 90 \\
 \hline
\end{tabular}
\vspace{-0.2cm}
\caption{Appropriate epoch size (MUCT dataset). }\label{ProgressiveScore}
\label{table:ProgressiveScore}
\end{center}
\vspace{-0.2cm}
\end{table}

\begin{table}[t]
\vspace{-0.3cm}
  \centering
  \subfloat[][]{
 \begin{tabular}{|c|c|c|} 
 \hline
 \textbf{Sel.} & \textbf{PIQUE} & \textbf{Default} \\
 \hline
2.5\% & 0.35 & 0.33 \\ 
 \hline
 5\% & 0.44 & 0.37 \\
 \hline
7.5\% &0.46 & 0.40 \\
 \hline
10\% & 0.49 & 0.44 \\ 
 \hline
\end{tabular}
}%
\qquad
 \subfloat[][]{
 \begin{tabular}{|c|c|c|} 
 \hline
 \textbf{Sel.} & \textbf{PIQUE} & \textbf{Default} \\
 \hline
2.5\% & 0.35 & 0.27 \\ 
 \hline
 5\% & 0.44 & 0.32 \\
 \hline
7.5\% &0.46 &  0.36 \\
 \hline
10\% & 0.49 &  0.40 \\
 \hline
\end{tabular}

}
\vspace{-0.2cm}
  \caption{Progressive score comparison in MUCT: (a) Candidate Selection and (b) Benefit Estimation.}%
  \label{table:StrategyCompare}%
\end{table}

Table \ref{table:ProgressiveScore} shows, for each selectivity value $x$, the completion time $T_q(x)$, the epoch size (as \% of $T_q(x)$) that generates the highest progressiveness score, and the completion time of PIQUE when run using that epoch size. We can conclude from the results that the epoch size should be higher when the selectivity is high since we need to execute more triples in the plan execution phase as compared to the case when the selectivity is low. As expected, we can see that the completion time of PIQUE increases with increase in the selectivity values.


\vspace{0.1cm}
\noindent
\textbf{Experiment 3 (Plan Generation Time). }\label{sect:plan_gen_time} We measure the average time taken by the plan generation phase of PIQUE for different selectivity values. For each of those values, we set the epoch size to the value derived in Experiment 2. 

The results (see Table \ref{table:planGenDetailsSeconds}) show that the plan generation time increases with the increase in the selectivity values. This is  because the higher the selectivity value is, the higher the number of triples generated in the triple generation step. 
Among the four steps, the benefit estimation step takes the highest amount of the plan generation time as it requires deriving the estimated uncertainty value of each triple using the decision table and then estimating the benefit metric of the triples (see Figure \ref{fig:PlanGendDetailsExp}). The triple selection step takes the least amount of the plan generation time as it is only responsible for retrieving the triples from the priority queue.

\begin{figure}[t]
\begin{tikzpicture}[scale=10]
 \node(img2)  {\includegraphics[trim={0.1cm 0.1cm 0.4cm 0cm},clip,width=0.48\textwidth]{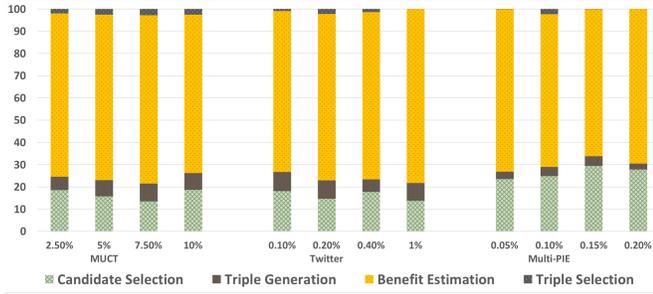}};
\end{tikzpicture}
\caption{Breakdown of the plan generation step time. 
} 
\label{fig:PlanGendDetailsExp}

\end{figure}

\begin{table}[t]
\vspace{-0.2cm}
\centering
\begin{center}
 \begin{tabular}{|c|c|c|c|c|c|} 
 \hline
 \multicolumn{2}{|c}{\textbf{MUCT}} & \multicolumn{2}{|c}{\textbf{Twitter}} & \multicolumn{2}{|c|}{\textbf{Multi-PIE}}\\
 \hline
 \textbf{Sel.} & \textbf{Time} & \textbf{Sel.} & \textbf{Time} & \textbf{Sel.} & \textbf{Time} \\
 \hline
2.5\% & 0.131837 & 0.1\% & 0.132745 & 0.05\% & 0.059073\\ 
 \hline
 5\% & 0.192911 & 0.2\% & 0.262242 & 0.1\%  & 0.067696\\
 \hline
7.5\% &0.399636 &  0.4\% & 0.435919 & 0.15\%  & 0.07026 \\
 \hline
10\% & 0.531469 & 1\% & 1.26086 & 0.2\%  & 1.271322\\
 \hline
\end{tabular}
\vspace{-0.2cm}
\caption{Average plan generation time (in seconds).}\label{PlanGenerationTableMultiPie}
\label{table:planGenDetailsSeconds}
\end{center}
\vspace{-0.2cm}
\end{table}

\begin{figure}[t]
\centering
\vspace{-0.4cm}
\begin{tikzpicture}
  
  \node (img1)  {\includegraphics[trim={1.1cm 0.6cm 1cm 1.32cm},clip,width=0.20\textwidth]{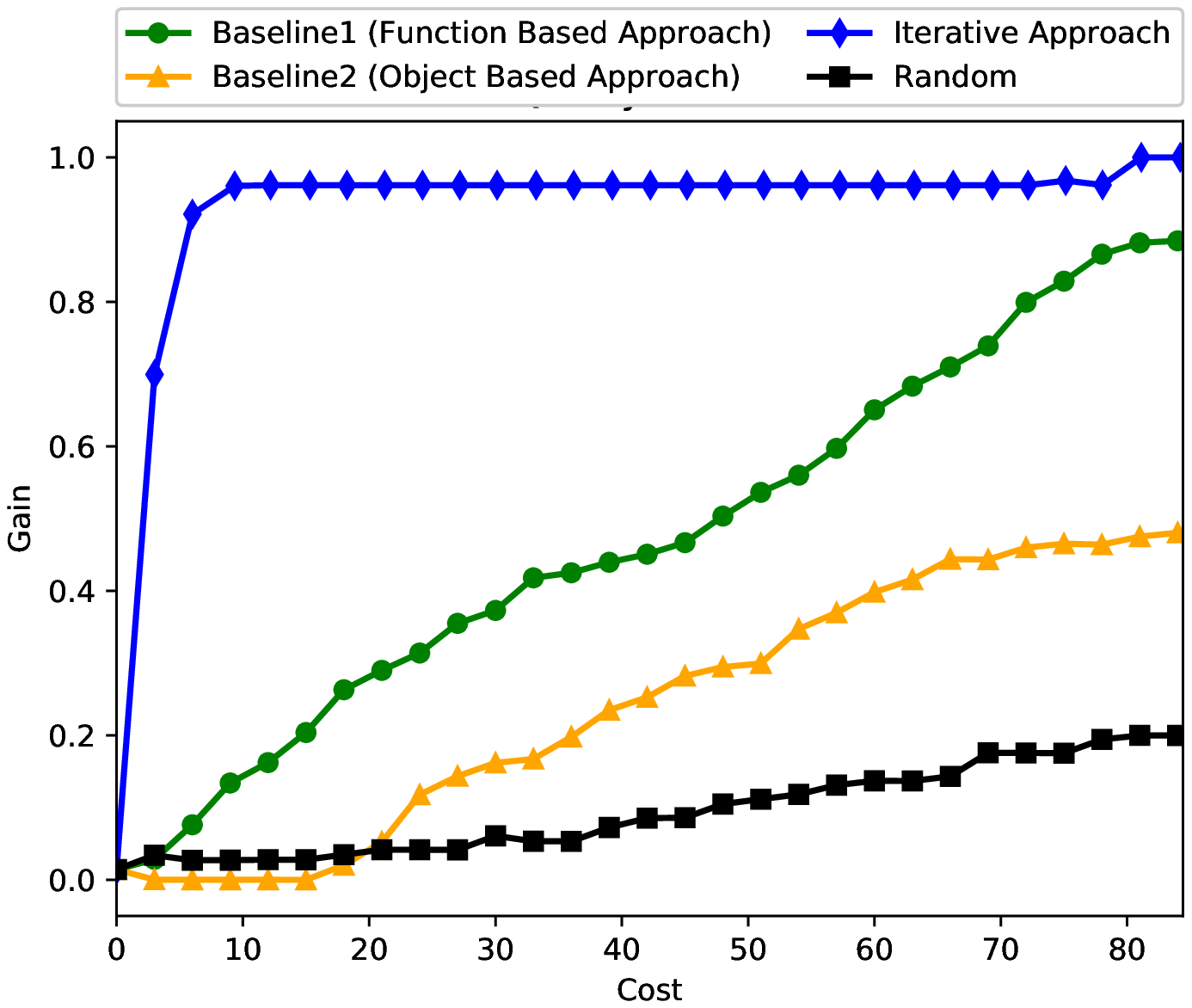}};

  \node[below=of img1, node distance=0cm, yshift=1.2cm,font=\color{black}] {\small Time (Seconds)};
  \node[left=of img1, node distance=0cm, rotate=90, anchor=center,yshift=-0.9cm,font=\color{black}] {\small Quality (Gain)};
  \hfill;
  
  \node[right=of img1,xshift=-1cm] (img2)  {\includegraphics[trim={1.1cm 0.6cm 1cm 1.32cm},clip,width= 0.20\textwidth]{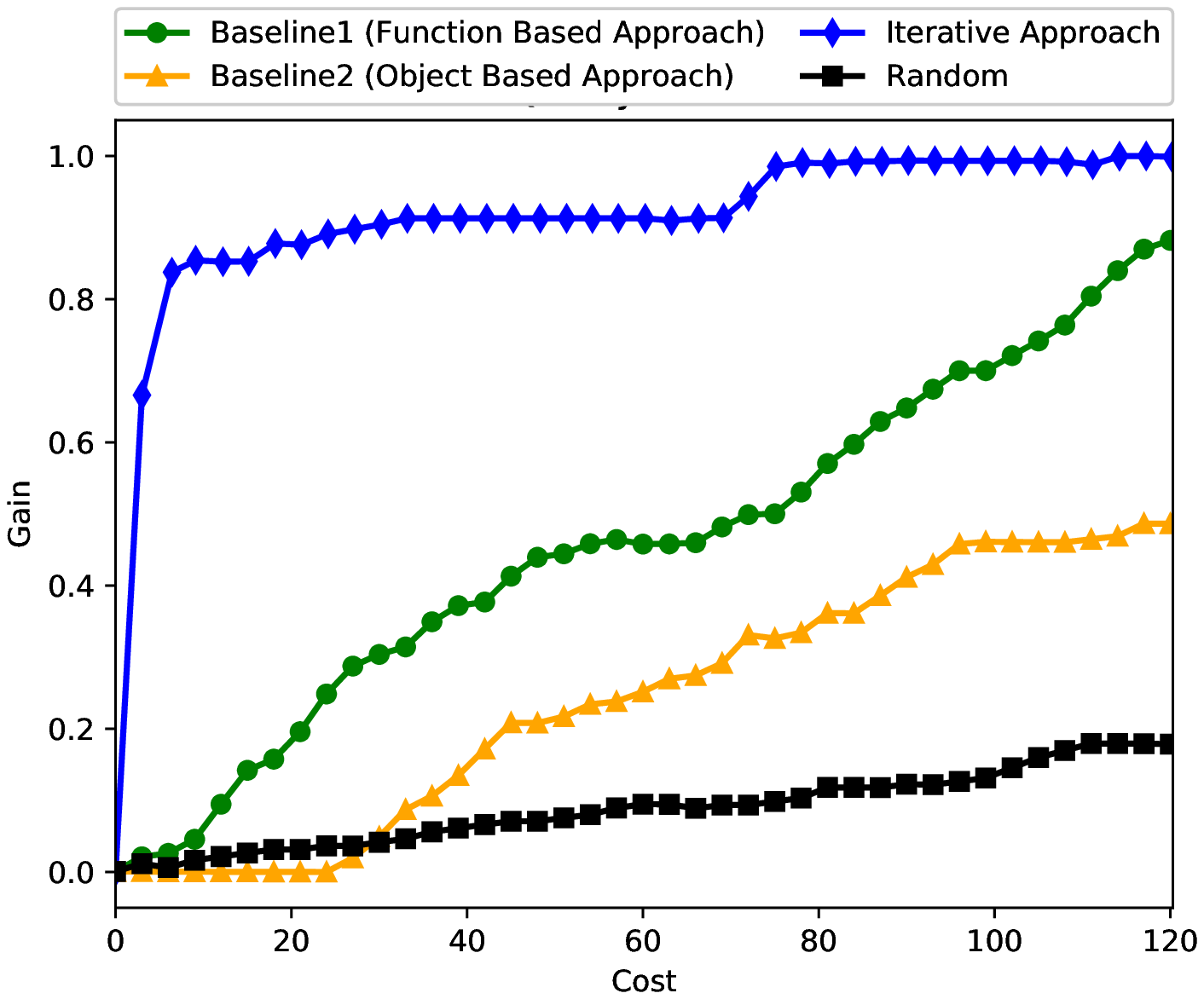}};
  \node[below=of img2, node distance=0cm, yshift=1.2cm] {\small Time (Seconds)};
  \node[left=of img2, node distance=0cm, rotate=90, anchor=center,yshift=-1cm] {\small Quality (Gain)};
  
   \node[anchor=center, node distance=0cm,yshift= -0.1cm, xshift = -1cm, above = of img1](legend1){\color{blue} \small {$\blacklozenge$}  PIQUE };
    \node[anchor=center, node distance=0cm, right = of legend1](legend2){\color{ao(english)} \small  { $\bullet$} Baseline 1 };
    
   \node[node distance=0cm, anchor=center, right = of legend2] (legend3)  {\color{amber} \small  { $\blacktriangle$}  Baseline 2};
    \node[below=of x-label1,  node distance=0cm, yshift=1.2cm]{\small (a) 0.1\% selectivity};
  \node[below=of x-label2,  node distance=0cm, yshift=1.2cm]{\small (b) 0.2\% selectivity};
  \node[node distance=0cm, anchor=center, right = of legend3] (legend4)  {\small { $\blacksquare$} Baseline 3};
  
\end{tikzpicture}
\caption{Comparison of gain (Multi-PIE dataset) for Q3.}
\label{fig:VariationOfQuality(MultiFeature)}
\end{figure}


\begin{figure}[t]
\centering
\vspace{-0.2cm}
\begin{tikzpicture}
  \node (img1)  {\includegraphics[trim={1.1cm 0.6cm 1cm 1.32cm},clip,width=0.20\textwidth]{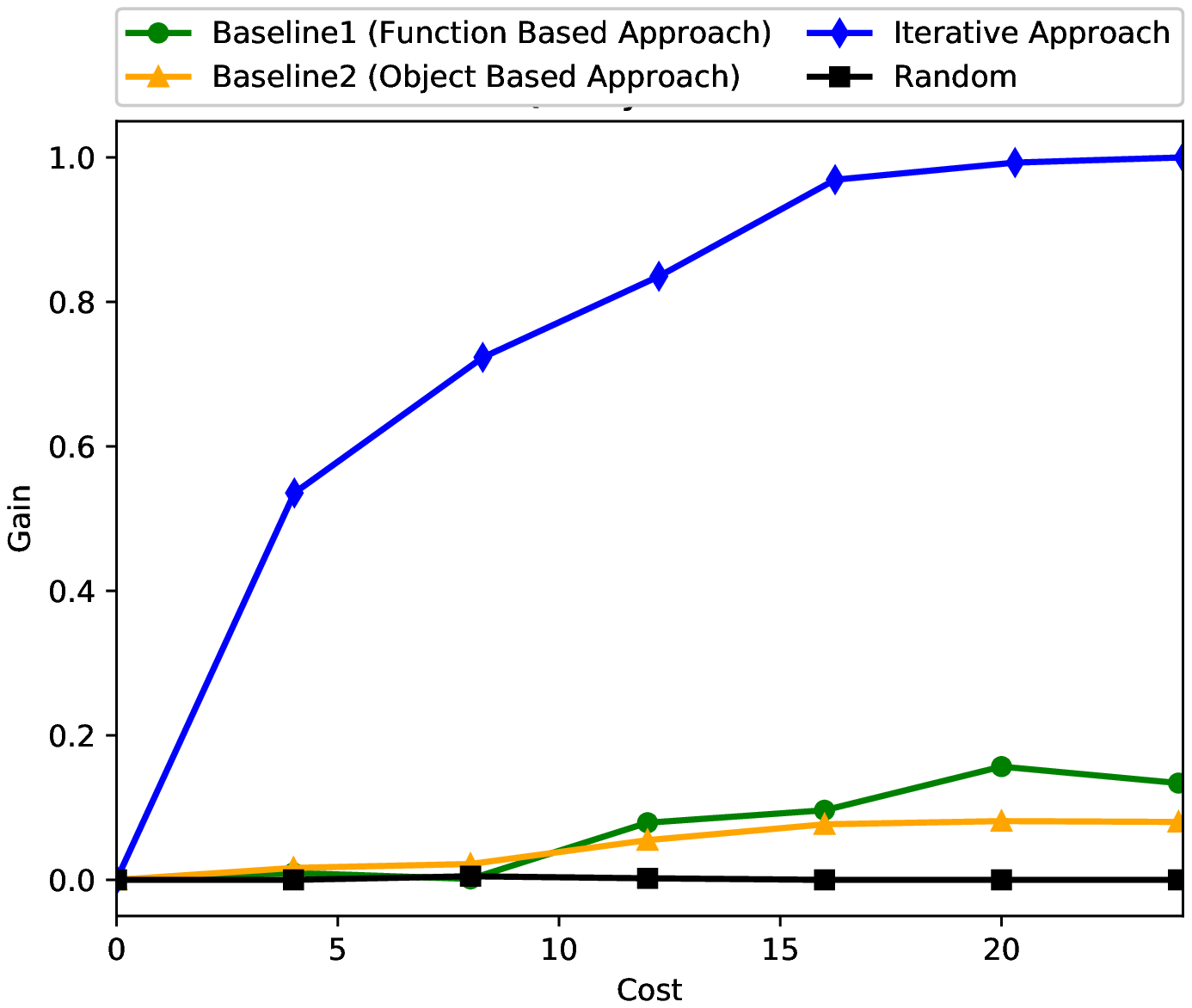}};
  \node[below=of img1, node distance=0cm, yshift=1.2cm,font=\color{black}] {\small Time (Seconds)};
  \node[left=of img1, node distance=0cm, rotate=90, anchor=center,yshift=-0.9cm,font=\color{black}] {\small Quality (Gain)};
  \hfill;
  
  \node[right=of img1,xshift=-1cm] (img2)  {\includegraphics[trim={1.1cm 0.6cm 1cm 1.32cm},clip,width=0.20\textwidth]{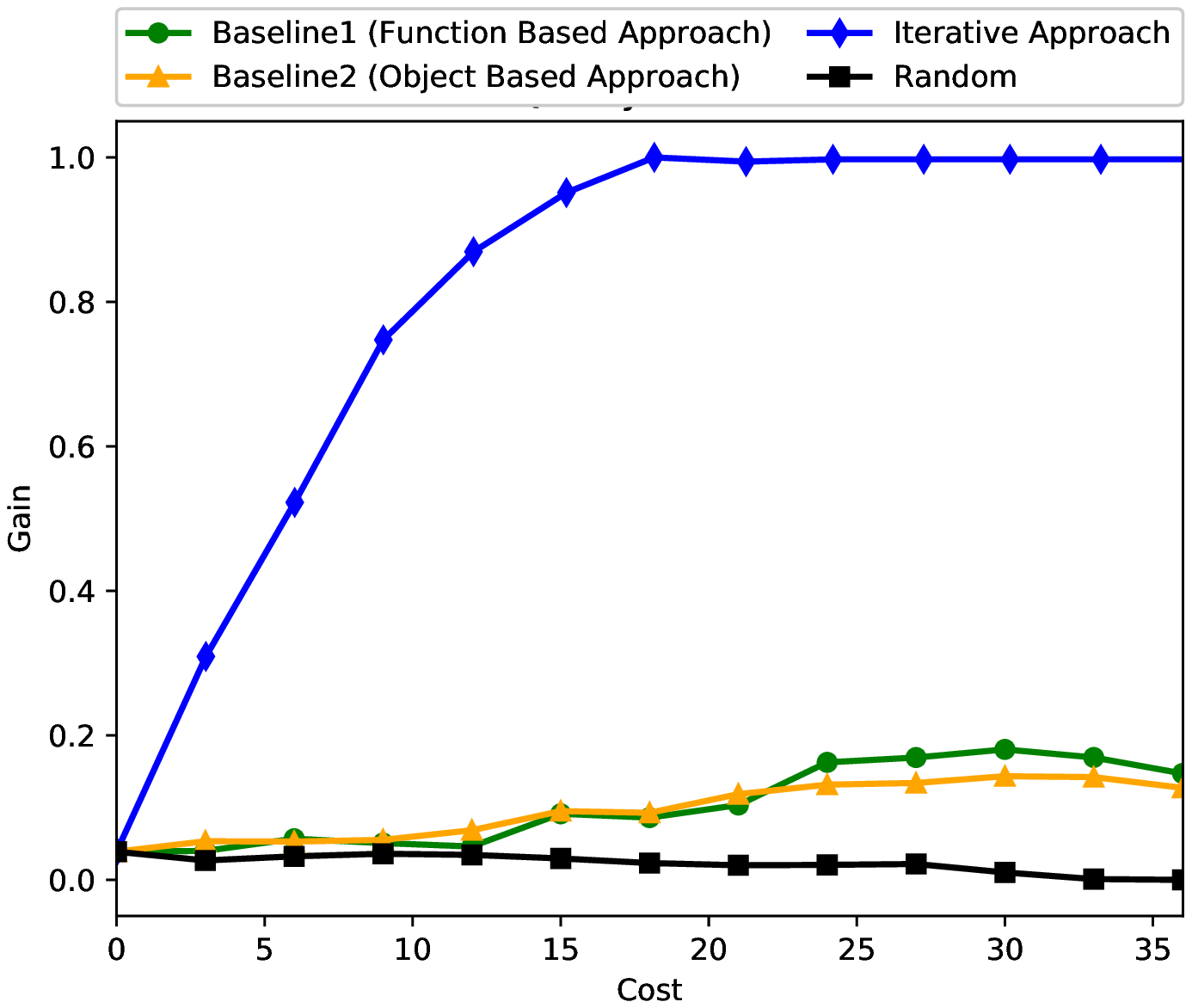}};
  
  \node[below=of img2, node distance=0cm, yshift=1.2cm] {\small Time (Seconds)};
  \node[left=of img2, node distance=0cm, rotate=90, anchor=center,yshift=-1cm] {\small Quality (Gain)};
  
   \node[anchor=center, node distance=0cm,yshift= -0.1cm, xshift = -1cm,  above = of img1](legend1){\color{blue} \small {$\blacklozenge$}  PIQUE };
    \node[anchor=center, node distance=0cm, right = of legend1](legend2){\color{ao(english)} \small  { $\bullet$} Baseline 1 };
    
   \node[node distance=0cm, anchor=center, right = of legend2] (legend3)  {\color{amber} \small  { $\blacktriangle$}  Baseline 2};
   \node[node distance=0cm, anchor=center, right = of legend3] (legend4)  {\small { $\blacksquare$} Baseline 3};
   \node[below=of x-label1, node distance=0cm, yshift=1.2cm, xshift = 1cm, text width=5cm] {\small (a) Query Q4 (0.1\% sel.)};
  \node[below=of x-label2, node distance=0cm, yshift=1.2cm, xshift = 1cm, text width=5cm] {\small (b) Query Q5 (0.1\% sel.)};
 \end{tikzpicture}
 
\caption{Comparison of gain (Multi-PIE) for Q4 and Q5.}
\label{fig:VariationOfQuality2(MultiFeature)}
\end{figure}


\begin{figure}[t]
\centering
\vspace{-0.4cm}
\begin{tikzpicture}
   \node(img1)  {\includegraphics[trim={1.1cm 0.6cm 1cm 1.32cm},clip,width=0.20\textwidth]{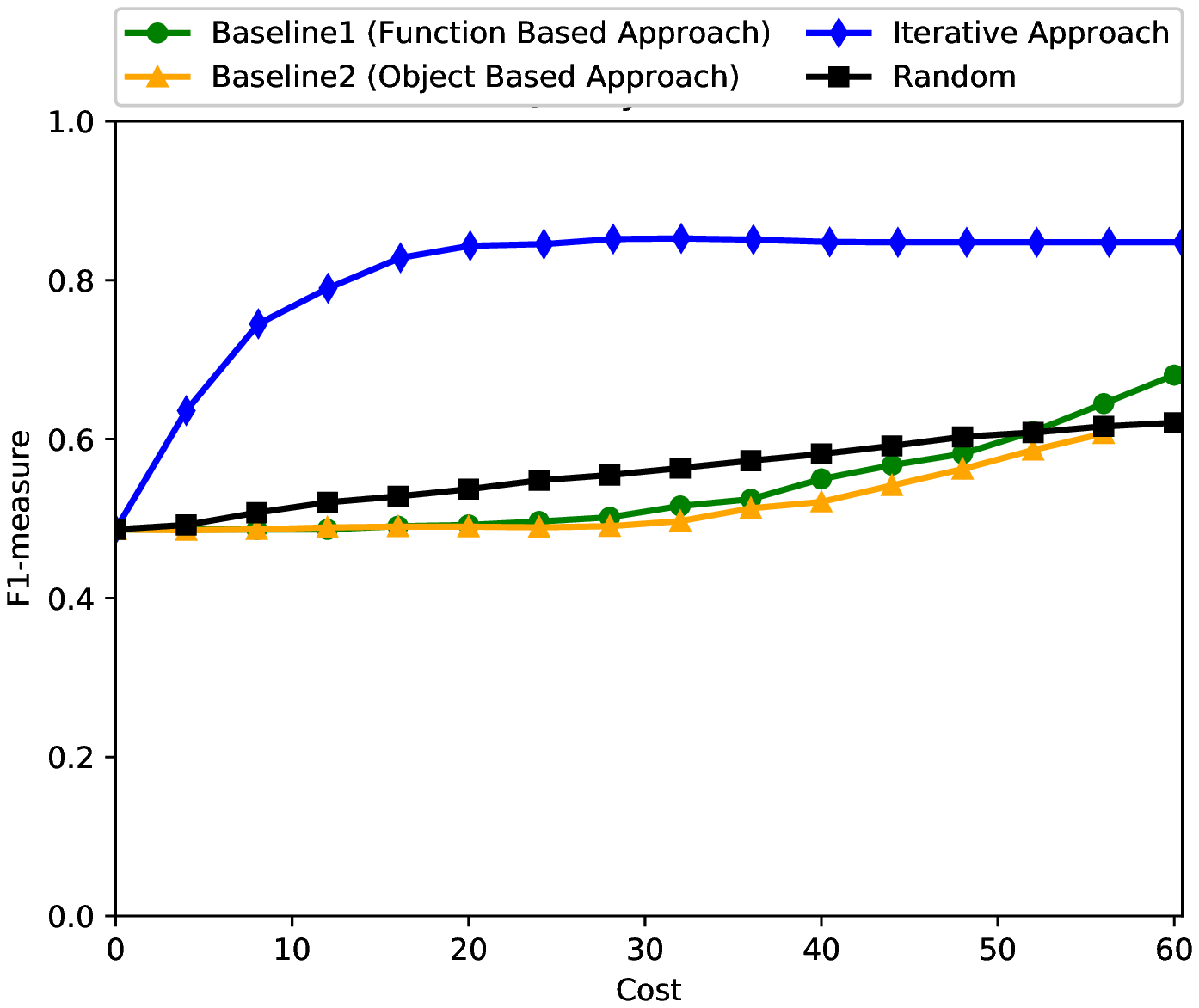}};
   
  \node[below=of img1, node distance=0cm, yshift=1.2cm] (x-label3){\small Time (Seconds)};
  \node[left=of img1, node distance=0cm, rotate=90, anchor=center,yshift=-1cm] {\small $F_1$ measure};
  
   \node[right=of img1,xshift=-1cm] (img2)  {\includegraphics[trim={1.1cm 0.6cm 1cm 1.32cm},clip,width=0.20\textwidth]{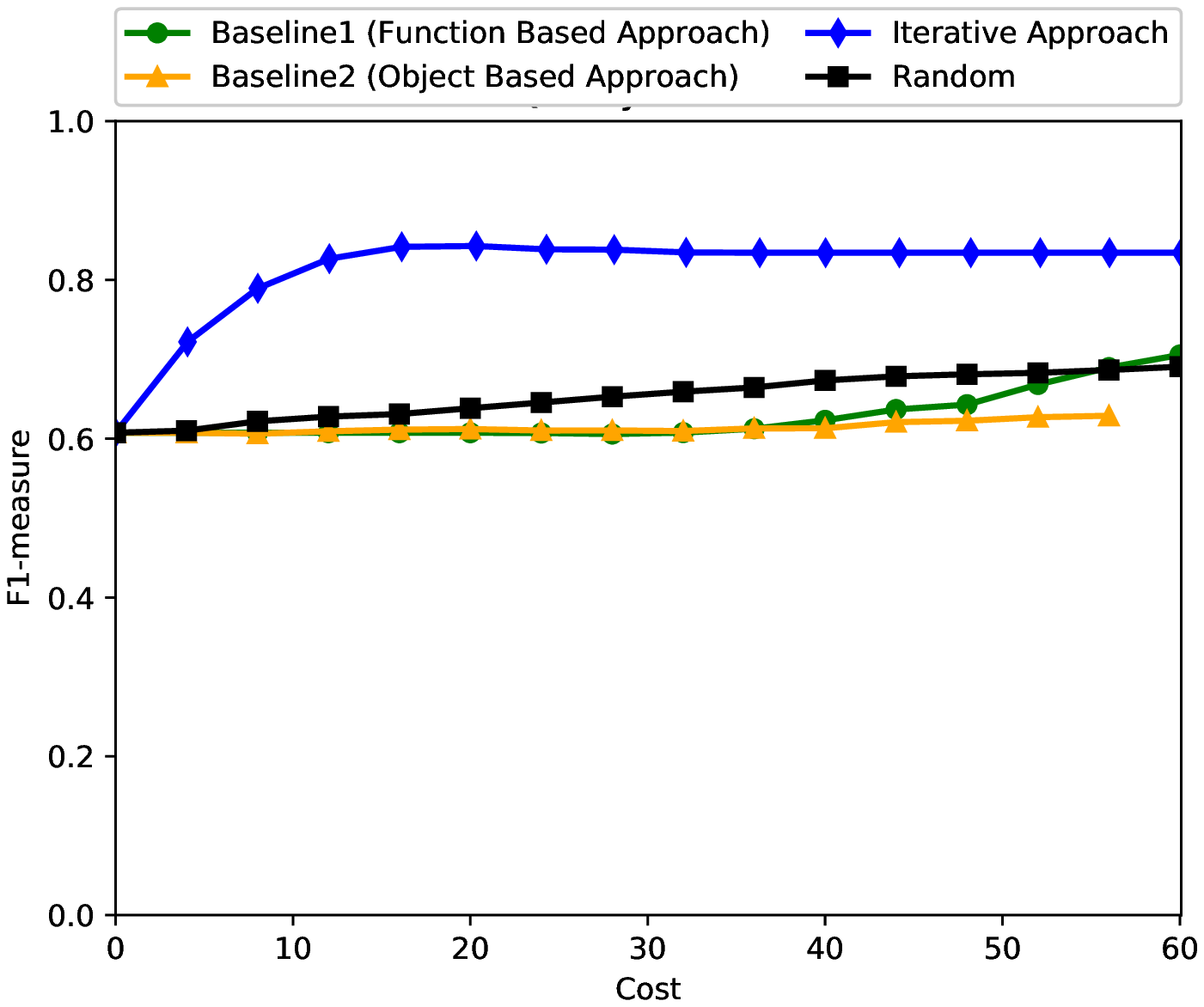}};
  \node[below=of img2, node distance=0cm, yshift=1.2cm](x-label4) {\small Time (Seconds)};
  \node[left=of img2, node distance=0cm, rotate=90, anchor=center,yshift=-1cm] {\small $F_1$ measure};
  
 \node[below=of x-label3,  node distance=0cm, yshift=1.1cm]{\small (a) 25\% caching};
  \node[below=of x-label4,  node distance=0cm, yshift=1.1cm]{\small (b) 50\% caching};
  
  \node[anchor=center, node distance=0cm,yshift= -0.1cm, xshift = -1cm, above = of img1](legend1){\color{blue} \small {$\blacklozenge$}  PIQUE};
    \node[anchor=center, node distance=0cm, right = of legend1](legend2){ \color{ao(english)} \small  { $\bullet$} Baseline 1 };
    
   \node[node distance=0cm, anchor=center, right = of legend2] (legend3)  {\color{amber} \small  { $\blacktriangle$} Baseline 2};
    \node[node distance=0cm, anchor=center, right = of legend3] (legend4)  {\small { $\blacksquare$} Baseline 3};   
  
\end{tikzpicture}
\vspace{-0.4cm}
\caption{\small Comparison of $F_1$ measure (MUCT, 10\% sel.) for Q1.}\label{fig:CachingVariationOfF1Experiment}
\vspace{-0.1cm}
\end{figure}



\vspace{0.1cm}
\noindent
\textbf{Experiment 4 (Candidate Selection Strategy). }\label{sect:candidate_selection} 
We evaluate our optimized candidate selection strategy as explained in Section~\ref{section:candidate_set_selection}. 
Table~\ref{table:StrategyCompare} (a) shows the results of comparing PIQUE with the default approach (that considers all objects in $O$ for the generation of plans) using their progressiveness scores (Equation~\ref{progressiveness-metric}). As expected, PIQUE outperforms the default strategy for the different selectivity values. As the size of the candidate set becomes smaller, the time taken by the benefit estimation step of PIQUE decreases. This empirically proves the validity of our candidate selection strategy.

\vspace{0.1cm}
\noindent
\textbf{Experiment 5 (Benefit Estimation Strategy).}\label{sect:benefit_estimation} We evaluate our efficient strategy of estimating the benefit of the triples present in $\mathsf{TS}_w$ (Equation~\ref{eqn:benefit}). 

In Table \ref{table:StrategyCompare} (b), we show the progressiveness scores of PIQUE (which uses the efficient strategy) and the default approach that require generating a new answer set in order to estimate the benefit of each triple (see Section~\ref{BenefitEstimation}). As expected, PIQUE outperforms the default strategy significantly because it can estimate the benefit of each triple using only the expression satisfiability probability of the object and the cost of the enrichment function present in the triple which can be performed in $\mathcal O(1)$. (Recall that the default strategy exhibits a time complexity of $\mathcal O(n)$ for estimating the benefit of each triple, where $n$ is the number of objects in the dataset.)

The default approach spends significantly more time in the benefit estimation step of the plan generation phase. In the benefit estimation step, for each triples present in $\mathsf{TS}_w$, we first perform the threshold selection algorithm and then estimate the improvement of the quality of the answer set due to the execution of the triple. As the threshold selection algorithm is performed for each triples, the aggregated time spent in the benefit estimation step becomes much higher compared to our approach.

\vspace{0.1cm}
\noindent
\textbf{Experiment 6 (Scalability of our Approach).} We evaluate the scalability of PIQUE with the help of the queries containing multiple predicates on tags. The results of query $Q3$ on the Multi-PIE dataset  are shown in Figure \ref{fig:VariationOfQuality(MultiFeature)}. The results of queries $Q4$ and $Q5$ on the same dataset are shown in Figure \ref{fig:VariationOfQuality2(MultiFeature)}.

PIQUE outperforms the baseline approaches significantly for all queries since it executes the best triples in each epoch. Thus, it can scale well on the Multi-PIE dataset with the increasing number of predicates containing different tags. We  observed that the cost of the plan generation phase increases with the increase in the number of those predicates. However, the benefit of executing the right set of triples overcomes the extra overhead paid in the plan generation phase.

\vspace{0.1cm}
\noindent
\textbf{Experiment 7 (Caching Previous Query Results).}
\label{exp:caching}
We examine the scenario where the results of some previous queries are cached.
PIQUE deals with caching by setting the starting state of each object $o_k$ to the state that accounts for the enrichment functions run on $o_k$ in the previous queries. For example, assuming there exists a predicate $R^i_j$ in $Q$, instead of setting $s(o_{k}, R^{i}_{j})$ to $[0, 0, 0, 0]$, we set it to $[1, 0, 0, 1]$ if functions $f^i_1$ and $f^i_4$ were run on $o_k$ before. The predicate probability of $o_k$ needs to account for the outputs of those functions.


In this experiment, we study the impact of different levels of caching on the performance of PIQUE and the baseline approaches. We  performed this experiment using query Q1 on the MUCT dataset. We assumed that one enrichment function was executed on a fraction of the dataset (i.e., 10\%, 25\%, 50\%, and 75\%) in a previous query and the properties of the objects (i.e., state, predicate probability, etc.) were cached. 

The results (see Figure~\ref{fig:CachingVariationOfF1Experiment}) show that PIQUE  outperforms the baseline approaches significantly for all caching levels. Since the outputs  of the enrichment function on the objects have been cached, the initial quality of the answer set increases as the level of caching  increases; the initial $F_1$ measure in the sub-figures (from left to right) are $0.49$ and $0.61$ respectively. 

\vspace{0.1cm}
\noindent
\textbf{Experiment 8 (Disk-Based Approach).}\label{exp:disk_based_approach} 
We compare the disk-based version of PIQUE with the baseline approaches. We change Baseline 1 as follows: 1)~Order disk blocks based on PIQUE's block benefit metric; 2)~Load the highest beneficial block from disk to memory and then execute the enrichment function with highest $q^i_j/c^i_j$ value on all of its objects; and~3)~Load the second highest beneficial block and execute the same enrichment function on all of its the objects. This process continues until that enrichment function is executed on the objects of all blocks. Then, we repeat the same process using each of the remaining enrichment functions.

 \begin{figure}[t]
 \centering
 \vspace{-0.4cm}
 \begin{tikzpicture}
  \node (img1)  {\includegraphics[trim={1.1cm 0.6cm 1cm 1.32cm},clip,width=0.20\textwidth]{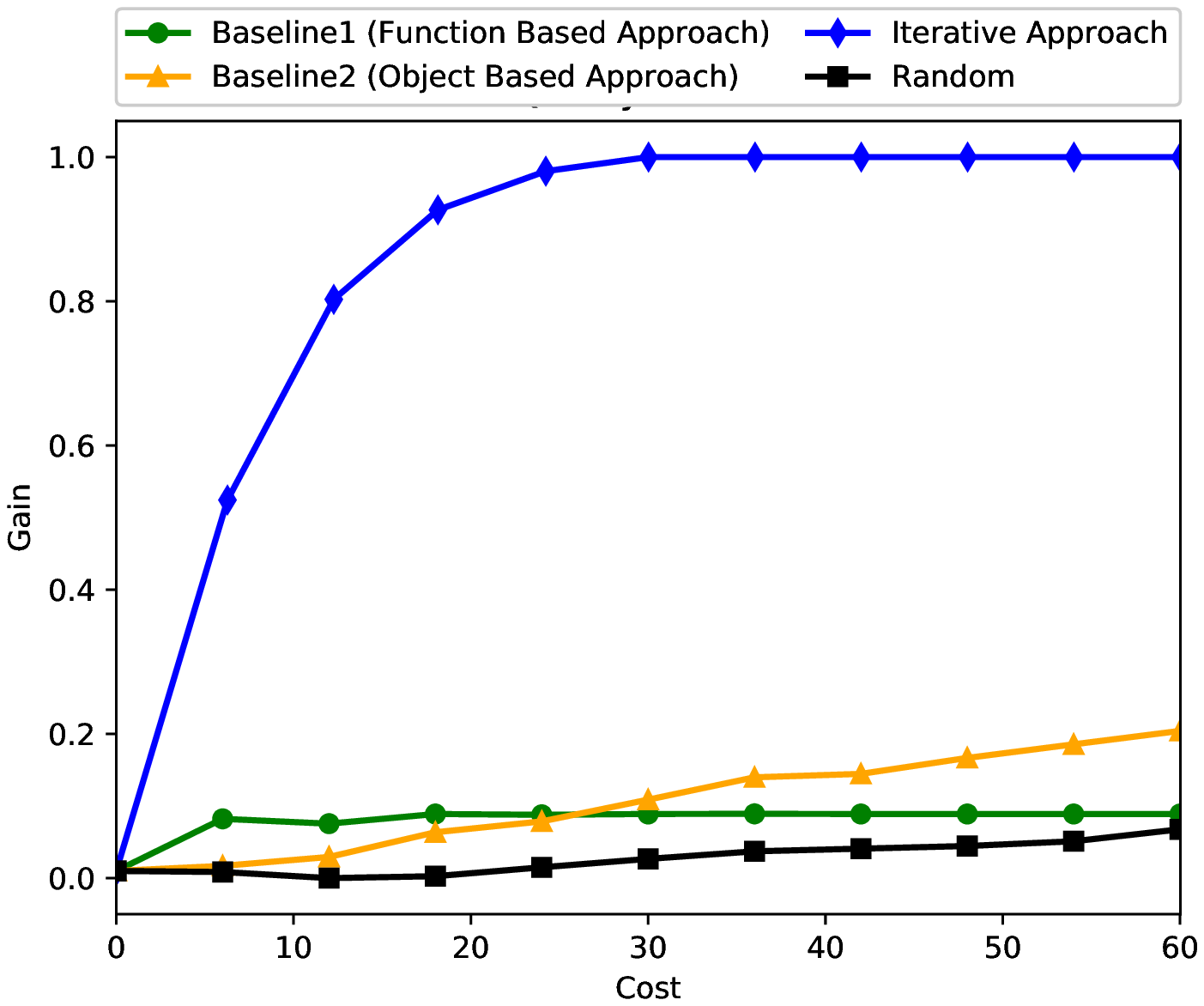}};
  \node[below=of img1, node distance=0cm, yshift=1.2cm,font=\color{black}] {\small Time (Seconds)};
  \node[left=of img1, node distance=0cm, rotate=90, anchor=center,yshift=-0.9cm,font=\color{black}] {\small Quality (Gain)};
  \hfill;
  \node[right=of img1,xshift=-1cm] (img2)  {\includegraphics[trim={1.1cm 0.6cm 1cm 1.32cm},clip,width=0.20\textwidth]{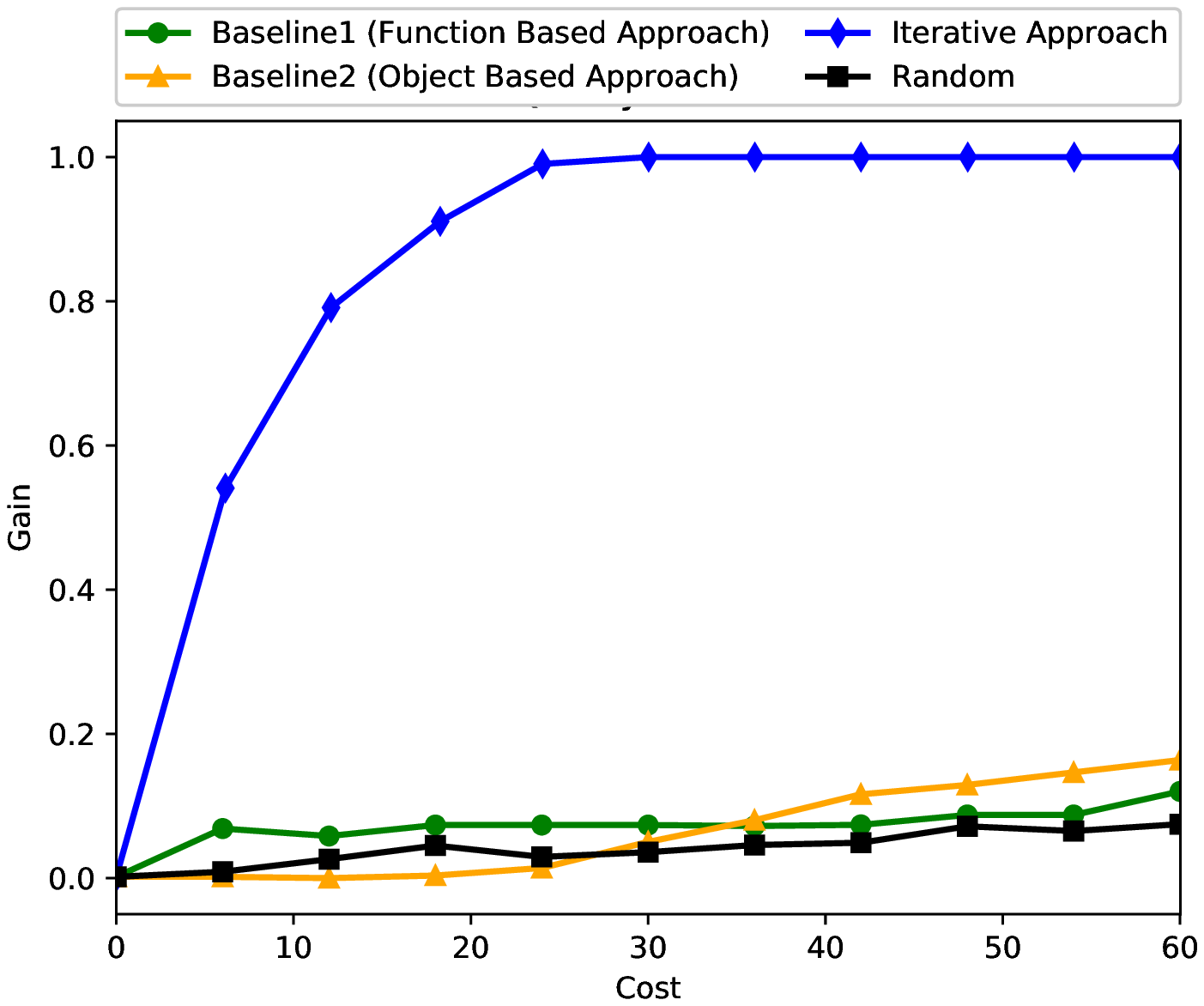}};
  \node[below=of img2, node distance=0cm, yshift=1.2cm] {\small Time (Seconds)};
  \node[left=of img2, node distance=0cm, rotate=90, anchor=center,yshift=-1cm] {\small Quality (Gain)};
  
   \node[anchor=center, node distance=0cm,yshift= -0.1cm, xshift = -1cm, above = of img1](legend1){\color{blue} \small {$\blacklozenge$}  Our Approach };
    \node[anchor=center, node distance=0cm, right = of legend1](legend2){\color{ao(english)} \small  { $\bullet$} Baseline 1 };
    
   \node[node distance=0cm, anchor=center, right = of legend2] (legend3)  {\color{amber} \small  { $\blacktriangle$}  Baseline 2};
   \node[node distance=0cm, anchor=center, right = of legend3] (legend4)  {\small { $\blacksquare$} Baseline 3};
   
    \node[below=of x-label1,  node distance=0cm, yshift=1.2cm, xshift = 1cm, text width=5cm]{\small (a) Query Q1 (50\% sel.)};
  \node[below=of x-label2,  node distance=0cm, yshift=1.2cm, xshift = 1cm, text width=5cm]{\small (b) Query Q1 (75\% sel.)};
 \end{tikzpicture}
 \vspace{-0.8cm}
 \caption{Gain for Disk-Based Approaches (MUCT).}
\label{fig:VariationOfQualityDiskBased}
\end{figure}

In Baseline 2, we make the following changes:~1)~Order the blocks in a decreasing order based on their block benefit values;~2)~Load the top blocks from disk to memory;~3)~Iterate over the objects from the loaded blocks (starting with the object with the highest expression satisfiability probability value) and for each object, execute all the enrichment functions on it;and ~4)~Once all objects are evaluated, load another set of blocks and repeat the same process. This process continues until all objects of all blocks are evaluated.

In Baseline 3, we make the following change. In each epoch, choose a number of blocks randomly. If a chosen block is not in memory, we load it in memory and write to disk the in-memory block that has least block benefit value. Once all chosen blocks are in-memory, we follow the same steps of Baseline 3 as explained in the previous section.

The results of this experiment (see Figure \ref{fig:VariationOfQualityDiskBased}) show that PIQUE performs significantly better than the baselines. This is mainly because PIQUE optimizes the process of selecting which blocks need to be loaded in memory at the beginning of each epoch and the process of choosing which triples should be executed during that epoch whereas the baseline approaches optimize none or only one of these two processes.


\section{Related Works}
\label{sect:relatedWork}
 


This paper develops a progressive approach to data enrichment and query evaluation. We are not aware of any work in the literature that directly addresses this problem. We review below works related to different aspects of the problem.

\vspace{0.2 cm}
\noindent
\textbf{Expensive Predicate Optimization.}  In expensive predicate optimization problem the goal is to optimize a query containing expensive predicates by appropriately ordering the evaluation of such predicates~\cite{OptimizationUserDefinedPredicates,PredicateMigration}. A subset of this problem, more relevant to our context, deals with the \textit{optimization of multi-version predicates}~\cite{Lazaridis:2007:OME:1247480.1247568,AcceleratingMachineLearning,ExploitingCorrelation}. E.g., in \cite{AcceleratingMachineLearning}, authors maintain a set of probabilistic predicates for different query predicates that can be present in an ML query on blobs. Given a query, appropriate probabilistic predicates were chosen to filter out large number of blobs which will not satisfy the query predicate in the future. In these works, the goal is to minimize the overall predicate evaluation cost with the help of a less computationally costly version of the predicate. In contrast with our context, in these works the predicates are considered to be deterministic. Furthermore enrichment of the objects are not considered by the problems of this domain.




\vspace{0.2 cm}
\noindent
\textbf{Adaptive Query Processing.}
In adaptive query processing, the goal is to optimize a query at the execution time with the help of intermixing the process of query execution with the process of exploration of the plan space~\cite{Eddies, Karanasos:2014:DOQ:2588555.2610531,Markl:2004:RQP:1007568.1007642}. Broadly, adaptive query processing addresses the problems of how to perform query optimization with missing statistics, unexpected correlations in data, unpredictable cost of operators, etc. Re-optimizing a query during the run time is similar in spirit to our approach, where we identify and execute the best set of triples in each epoch. However, these approaches do not address the issue of evaluating predicates using different non-deterministic functions which vary in cost and quality values. 



\vspace{0.05 cm}
\noindent
\textbf{Approximate Query Processing.} In the approximate query processing domain, the objective of the research has been to provide an approximate answers to aggregation queries along with an error bound \cite{BlinkDB,VerdictDB}. Authors used several sampling based methods, compression based methods, multi-resolution data structures to provide the approximate answers. Providing an inaccurate answer quickly and refining it with time is similar in spirit with our objective in PIQUE, but these systems do not consider enrichment of the underlying data before answering a query. Furthermore, the query plans generated in AQP systems are static, as the underlying assumption is once an analyst receives an approximate answer, the analyst will pose more meaningful queries according to the quality requirement. In contrary, PIQUE is an adaptive approach, which evaluates the query in an iterative manner but the interleaving enrichment process ensures that quality of the end result improves automatically $w.r.t.$ time. 

\section{Conclusions and Future Work}
\label{sect:conclusions}

In this paper, we have developed a progressive approach, entitled PIQUE, that enriches the right amount of data to the right degree so as to maximize the quality of the answer set. The goal is to use different tagging functions, which vary in cost and quality, in such a way that improves the quality of the answer progressively. We have proposed an efficient algorithm that generates a plan in every epoch with the objective of maximizing the quality of the answer at the end of the epoch. We have shown empirically that PIQUE executes the right set of triples in each epoch so that the quality of the answer set improves progressively over time. 

In its current form, PIQUE is a unary operator that enriches one or more tags (using appropriate enrichment functions based on the expression) from a single object collection, although our approach allows multiple PIQUE operators in a single query. We could extend PIQUE to binary operators wherein the expression considers enrichment of tags corresponding to multiple object sets with an appropriate join condition. Thus, an example to identify all the pairs of objects in two datasets with the same value for a tag might be:
\vspace{-0.2cm}
\begin{equation*}
\vspace{-0.2cm}
PIQUE(O1, O2, 4, O1.tag = O2.tag)
\end{equation*}

Extending PIQUE to such operators requires a non-trivial extension. For selection condition, the expression satisfiability probability value of the objects can be calculated in $\mathcal O(n)$ time whereas for join conditions, determining the ESP can be very expensive (i.e., $\mathcal O(n^2)$). Furthermore, the benefit estimation step of our approach needs to be changed as a single object can be part of multiple tuples of the answer set, and an enrichment of a single object will benefit the quality of the answer set in a much different way than the existing setup. Extending PIQUE to such a setup will be an interesting direction of future work.

\begin{appendix}
\section{Proofs}\label{appendix:proofs}


\begin{lemma}\label{lemma:thresholdFunction}
 Let  $\mathcal P_1$, $\mathcal P_2$, $\mathcal P_3$, $\dots$, $\mathcal P_{|O|}$ be the ESP values of the objects in $O$ in epoch $i$ such that $\mathcal P_1$ $\geq$ $\mathcal P_2$ $\geq$ $\mathcal P_3$ $\geq$ $\dots$ $\geq$ $\mathcal P_{|O|}$. The  threshold probability $\mathcal P^i_\tau$ is defined as the lowest value $\mathcal P_j$ for which the following conditions hold:
\begin{itemize}
\item $\mathcal P_{j} > \frac{(\mathcal P_1+\mathcal P_2+\dots \mathcal P_{j-1})}{(j-1+k)}$, and
\item $ \mathcal P_{j+1} < \frac{(\mathcal P_1+\mathcal P_2+\dots \mathcal P_j)}{(j+k)}$, where,  $k = \sum\limits_{i=1}^{|O|} \mathcal P_i$
\end{itemize}
\end{lemma}
\begin{proof}
As described in Theorem \ref{theorem:threshold}, $E(F_\alpha)$ measure of the answer set will monotonically increase up to a certain object $o_\tau$ with ESP value of $\mathcal P_\tau$ and then it will keep decreasing monotonically. Let us denote, $F_\tau$ as the $E(F_\alpha)$ measure of the answer set if we include $\tau-th$ object and $F_{\tau+1}$ be the $E(F_\alpha)$ measure, if we include $(\tau+1)$-th object in the answer set.  

\begin{equation}
\begin{split}
&\Delta(E(F_\alpha))^{\tau+1}_{\tau} \\&= \frac{(1+\alpha)(\mathcal P_1+\mathcal P_2+...\mathcal P_{\tau+1})}{\tau+1+k} - \frac{(1+\alpha)(\mathcal P_1+\mathcal P_2+...\mathcal P_\tau)}{\tau+k} \ \ \\& where \ \ k = \sum\limits_{i=1}^{|O|} \mathcal P_i \\
&= (1+\alpha)\frac{\splitfrac{(\tau+k)(\mathcal P_1+\mathcal P_2+...+\mathcal P_{\tau+1})}{-(\mathcal P_1+\mathcal P_2+...+\mathcal P_\tau)(\tau+1+k)}}{(\tau+1+k)(\tau+k)} \\
&=  (1+\alpha)\frac{\splitfrac{(\tau+k)(\mathcal P_1+\mathcal P_2+...+\mathcal P_{\tau}) + (\tau+k)\mathcal P_{\tau+1}}{\splitfrac{-(\mathcal P_1+\mathcal P_2+...+\mathcal P_\tau)(\tau+k)}{ - (\mathcal P_1+\mathcal P_2+...+\mathcal P_\tau)}}}{(\tau+1+k)(\tau+k)} \\
&= (1+\alpha)\frac{\mathcal P_{\tau+1}(\tau+k) - (\mathcal P_1+\mathcal P_2+...\mathcal P_\tau)}{(\tau+1+k)(\tau+k)}
\end{split}
\end{equation}

$E(F_\alpha)$ measure will keep increasing as long as the value of $\Delta(E(F_\alpha))^{\tau+1}_{\tau}$ remains positive. Thus the threshold value will be the lowest value of $\tau$ for which $\Delta(E(F_\alpha))^{\tau+1}_{\tau}$ value becomes negative. 

\begin{equation}\label{eqn:DeltaThreshold}
\begin{split}
\Delta(E(F_\alpha))^{\tau+1}_{\tau} < 0 & \Rightarrow \mathcal P_{\tau+1}(\tau+k) - (\mathcal P_1+\mathcal P_2+...\mathcal P_\tau) < 0 \\
& \Rightarrow \mathcal P_{\tau+1}(\tau+k) < (\mathcal P_1+\mathcal P_2+...\mathcal P_\tau) \\
& \Rightarrow \mathcal P_{\tau+1} < \frac{(\mathcal P_1+\mathcal P_2+...\mathcal P_\tau)}{(\tau+k)}
\end{split}
\end{equation}
The object $o_\tau$ with the ESP value of $\mathcal P_\tau$ will be the threshold probability if the following two inequalities hold: 
\begin{equation}\label{eqn:ThresholdInequalityGreater}
\begin{split}
 \mathcal P_{\tau} > \frac{(\mathcal P_1+\mathcal P_2+...\mathcal P_{\tau-1})}{(\tau-1+k)} 
 \end{split}
\end{equation}
\begin{equation}\label{eqn:ThresholdInequalityLess}
\mathcal P_{\tau+1} < \frac{(\mathcal P_1+\mathcal P_2+...\mathcal P_\tau)}{(\tau+k)}
\end{equation}
\end{proof}

\subsection{Proof of Theorem  \ref{theorem:thresholdFunctionInsideObjectAll}}

\begin{proof}
We prove this theorem with the help of the following lemmas. 

\begin{lemma}\label{theorem:lemmaInsideObject1}
If the ESP value of $o_k$ increases in epoch $w$, then the threshold $\mathcal P^{w}_\tau$ can remain the same as $\mathcal P^{w-1}_\tau$ or increase (some objects which were part of ${\mathsf{A}_{w-1}}$ can move out of ${\mathsf{A}_{w}}$). In both cases, the $E(F_\alpha)$ measure of the answer set will be higher than that of  ${\mathsf{A}_{w-1}}$. 
\end{lemma}

\begin{lemma}\label{theorem:lemmaInsideObject2}

If the ESP value of $o_k$ decreases in epoch $w$ but still remains higher than $\mathcal P^{w-1}_\tau$, then the threshold $\mathcal P^{w}_\tau$ can remain the same as $\mathcal P^{w-1}_\tau$ or decrease. This implies that the objects which were already part of ${\mathsf{A}_{w-1}}$ will still remain in ${\mathsf{A}_{w}}$ and some new objects might be added to ${\mathsf{A}_{w}}$. In both the cases, the $E(F_\alpha)$ measure of ${\mathsf{A}_{w}}$ will be lower than ${\mathsf{A}_{w-1}}$.
\end{lemma}

\begin{lemma}\label{theorem:lemmaInsideObject3}
If the ESP value of $o_k$ decreases in epoch $w$ and becomes less than $\mathcal P^{w-1}_\tau$, then the threshold $\mathcal P^{w}_\tau$ can increase or decrease. In both the cases, the $E(F_\alpha)$ measure of ${\mathsf{A}_{w}}$ will be lower than that of ${\mathsf{A}_{w-1}}$.
\end{lemma}

Based on the above lemmas we can conclude that, given an object $o_k$ $\in$ $\mathsf{A}_{w-1}$, the $E(F_\alpha)$ measure of $\mathsf{A}_{w}$ will increase {\em w.r.t.} $\mathsf{A}_{w-1}$ only if the ESP value of $o_k$ increases in epoch $w$. 
\vspace{-0.1cm}

\begin{lemma}\label{theorem:lemmaOutsideObject1}
If the ESP value of $o_k$ increases in epoch $w$ and becomes higher than \emph{$\mathcal P^{w-1}_\tau$}, then the threshold $\mathcal P^{w}_\tau$ can remain the same as \emph{$\mathcal P^{w-1}_\tau$} or increase (i.e., some objects which were part of $\mathsf{A}_{w-1}$, might move out of $\mathsf{A}_{w}$). In both the cases, the $E(F_\alpha)$ measure of $\mathsf{A}_{w}$ will be higher than  $\mathsf{A}_{w-1}$.
\end{lemma}

\begin{lemma}\label{theorem:lemmaOutsideObject2}
If the ESP value of $o_k$ increases in epoch $w$ but does not become higher than $\mathcal P^{w-1}_\tau$, then the threshold $\mathcal P^{w}_\tau$ will decrease. The $E(F_\alpha)$ measure of $\mathsf{A}_{w}$ will increase as compared to $\mathsf{A}_{w-1}$. 
\end{lemma}

\begin{lemma}\label{theorem:lemmaOutsideObject3}
If the ESP value of $o_k$ decreases in epoch $w$, then the threshold $\mathcal P^{w}_\tau$ will remain the same as $\mathcal P^{w-1}_\tau$. The $E(F_\alpha)$ measure of $\mathsf{A}_{w}$ will be equal to the $E(F_\alpha)$ measure of $\mathsf{A}_{w-1}$. 
\end{lemma}

According to Lemmas \ref{theorem:lemmaOutsideObject1}, \ref{theorem:lemmaOutsideObject2}, and \ref{theorem:lemmaOutsideObject3}, given an object $o_k$ $\not\in$ $\mathsf{A}_{w-1}$, we can conclude that the $E(F_\alpha)$ measure of $\mathsf{A}_{w}$ increases or remains the same {\em w.r.t.} $\mathsf{A}_{w-1}$ but it never decreases. 
\end{proof}


\subsection{Proof of Lemma \ref{theorem:lemmaInsideObject1} }
\begin{proof}
Let us assume that ESP value of object $o_k$ is increased from $\mathcal P_k$ to $\mathcal P_k + \Delta$. Since $\mathcal P^{w-1}_\tau$ is the threshold probability of epoch $w-1$, the following inequality holds:
\begin{equation}\label{ref:p_(m+1)}
\begin{split}
&\mathcal P^{w-1}_\tau > \frac{(\mathcal P_1+\mathcal P_2+...\mathcal P_{\tau-1})}{(\tau-1+K)} \\
&\Rightarrow \mathcal P_{\tau}(\tau-1+K) > (\mathcal P_1+\mathcal P_2+...\mathcal P_{\tau-1})
\end{split}
\end{equation}

In the above equation, $K = \sum\limits_{i=1}^{|O|} \mathcal P_i$. In epoch $w$, if $o_{\tau}$ is still the threshold even after increasing the ESP value object $o_k$ then the following inequality must hold:
\begin{equation}\label{eqn:deltaP}
\begin{split}
&\mathcal P_{\tau} > \frac{(\mathcal P_1+\mathcal P_2+..+\mathcal P_i+\Delta+....\mathcal P_{\tau-1})}{(\tau-1+K+\Delta)} ,  K = \sum\limits_{j=1}^{|O|} \mathcal P_j \\
&\Rightarrow \mathcal P_{\tau}(\tau-1+K+\Delta) > (\mathcal P_1+\mathcal P_2+..+\mathcal P_i+\Delta+....\mathcal P_{\tau-1}) \\
&\Rightarrow \mathcal P_{\tau}(\tau-1+K) + \Delta \mathcal P_{\tau} > \mathcal P_1+\mathcal P_2+...+\mathcal P_{\tau-1} + \Delta \\
&\Rightarrow \mathcal P_{\tau}(\tau+k) - (\mathcal P_1+\mathcal P_2+...+\mathcal P_{\tau-1}) > \Delta (1-\mathcal P_{\tau}) \\
& \Rightarrow \Delta (1-\mathcal P_{\tau}) < \mathcal P_{\tau}(\tau+K) - (\mathcal P_1+\mathcal P_2+...+\mathcal P_{\tau-1}) \\
&\Rightarrow\Delta < \frac{\mathcal P_{\tau}(\tau+K) - (\mathcal P_1+\mathcal P_2+...+\mathcal P_{\tau-1})}{1-\mathcal P_{\tau}}
\end{split}
\end{equation}


The threshold value of epoch $w$ will remain the same as $w-1$, as long as $\Delta$ value is less than $\frac{\mathcal P_{\tau}(\tau+k) - \sum\limits_{i=1}^{\tau-1}\mathcal P_i}{1-\mathcal P_{\tau}}$, otherwise it will be increased. 

From Equation \ref{eqn:expected_f1_measure}, we can see that as the ESP value of $o_k$
increases from $\mathcal P_k$ to $\mathcal P_k + \Delta$, the numerator increases by an amount of $(1+ \alpha) \cdot \Delta$. The denominator also increases but it increases by a smaller amount of $(\alpha \cdot \Delta)$ as compared to the numerator. This implies that the value of $E(F_{\alpha}(\mathsf{A}_w))$ will increase from epoch $w-1$. 
\end{proof}

\subsection{Proof of Lemma \ref{theorem:lemmaInsideObject2} }

\begin{proof}
Suppose the ESP value object $o_k$ decreases from $\mathcal P_k$ to $\mathcal P_k- \Delta$ in epoch $w$. According to the inequality \ref{eqn:ThresholdInequalityGreater}, if object $o_\tau$ is still in the answer set, then the following condition must hold:
\begin{equation}
\begin{split}
\mathcal P_{\tau} > \frac{(\mathcal P_1+\mathcal P_2+...\mathcal P_{\tau-1} - \Delta)}{(\tau-1+K - \Delta)}
\end{split}
\end{equation}
Now since the right side of the above inequality $\frac{(\mathcal P_1+\mathcal P_2+...\mathcal P_{\tau-1} - \Delta)}{(\tau-1+K - \Delta)}$ is reduced from the previous value, i.e. $\frac{(\mathcal P_1+\mathcal P_2+...\mathcal P_{\tau-1})}{(\tau-1+K)}$, this inequality will always hold. This implies that the threshold will remain the same or it will decrease.
From Equation \ref{eqn:expected_f1_measure}, we can see that the value in the numerator and denominator both decreases. However, the value of the numerator decreases more as compared to the value of the denominator which implies an overall reduction of $E(F_\alpha)$ measure of the answer set. 
\end{proof}

\subsection{Proof of Lemma \ref{theorem:lemmaInsideObject3} }
\begin{proof}
Suppose the ESP value of object $o_k$ is reduced from $\mathcal P_k$ to $\mathcal P_k-\Delta$ such that the new ESP value $\mathcal P_k -\Delta$ is lower than $\mathcal P^{w-1}_\tau$. From Equation \ref{eqn:expected_f1_measure}, it can be observed that in the new epoch the numerator is reduced by an amount of $(1+\alpha) \cdot (\mathcal P_k)$ whereas the denominator is reduced by a lower amount of $ \alpha \cdot \Delta$. This implies that the $E(F_\alpha)$ measure of the answer set will decrease from previous epoch $w-1$.

\end{proof}

\subsection{Proof of Lemma \ref{theorem:lemmaOutsideObject1} }
\begin{proof}
Suppose the ESP value of $o_k$ is increased from $\mathcal P_k$ to $\mathcal P_k+\Delta$, where the value of $\mathcal P_k + \Delta$ is higher than the previous threshold $\mathcal P^{w-1}_\tau$. In Equation \ref{eqn:expected_f1_measure}, it can be observed that in the new epoch the numerator is increased by an amount of $(1+\alpha) \cdot (\mathcal P_k + \Delta)$ whereas the denominator is increased by a lower amount of $ \alpha \cdot \Delta$. This implies that the $E(F_\alpha)$ measure of the answer set will increase from epoch $w-1$.

\end{proof}

\subsection{Proof of Lemma \ref{theorem:lemmaOutsideObject2} }
\begin{proof}

Let us consider that the ESP value of $o_k$ is increased from $\mathcal P_k$ to $\mathcal P_k+\Delta$, where the value of $\mathcal P_k + \Delta$ is lower than the previous threshold $\mathcal P^{w-1}_\tau$. In order to prove this Lemma, we first prove the condition on threshold probability of the new epoch $w$.

Let us consider $\mathcal P^{w-1}_\tau$ be the threshold probability of the answer set in epoch $w-1$. This implies $\mathcal P^{w-1}_\tau$ is higher than $\frac{(\mathcal P_1+\mathcal P_2+...\mathcal P_{\tau-1})}{(\tau-1+k)}$, according to inequality \ref{eqn:ThresholdInequalityGreater}. For simplicity of notation, we denote this fraction by $\frac{X}{Y}$. In the new epoch, in the R.H.S of Equation \ref{eqn:ThresholdInequalityGreater}, the value of numerator stays same, as no extra object was added in the answer set, whereas denominator increases from $Y$ to $Y+\Delta$. So the new right hand side $\frac{X}{Y+\Delta}$ gets smaller than previous right hand side $\frac{X}{Y}$, which implies previous threshold object $o_\tau$ will remain in the answer set of epoch $w$.

Let us consider the second condition of threshold related to object $o_{\tau+1}$. In epoch $w-1$, the following condition was true for object $o_{\tau+1}$:  $ \mathcal P_{\tau+1} < \frac{(\mathcal P_1+\mathcal P_2+...\mathcal P_\tau)}{(\tau+K)}$, where, $K = \sum\limits_{i=1}^N \mathcal P_i$. Let us denote the numerator and denominator by X and Y respectively. In this case, in new epoch $w$, the denominator Y will be increased as the ESP value of object $\mathcal P_k$ is increased from $\mathcal P_k$ to $\mathcal P_k+\Delta$. This implies that the value in the R.H.S (i.e., $\frac{X}{Y+\Delta}$) became smaller as compared to the previous value in the previous epoch (i.e., $\frac{X}{Y}$). This implies that object $o_{\tau+1}$ will become part of the answer set in epoch $w$. This implies that threshold will decrease. Once the threshold decreases, more number of objects are added to the answer set. This implies, in Equation \ref{eqn:expected_f1_measure}, the value of the numerator increases by at-least an amount of $(1+ \alpha) \cdot \mathcal P_{\tau+1}$. The value of the denominator also increases but only by a smaller amount of $(\alpha \cdot \Delta)$ as compared to the numerator. This implies that the $E(F_\alpha)$ measure of the answer set will increase.
\end{proof}

\subsection{Proof of Lemma \ref{theorem:lemmaOutsideObject3} }
\begin{proof}

Let us consider that the ESP value of $o_k$ is decreased from $\mathcal P_k$ to $\mathcal P_k - \Delta$, where the values of $\mathcal P_k $ and $\mathcal P_k - \Delta$, both are lower than the previous threshold $\mathcal P^{w-1}_\tau$. In this scenario, since the object $o_k$ was not part of the answer set in the previous epoch $w-1$, it did not contribute to the quality of the answer set $\mathsf A_{w-1}$. Since in the new epoch, the ESP value of object $o_k$ further decreases, it will not be part of the answer set of $\mathsf A_{w}$. This implies that the $E(F_\alpha)$ measure of the answer set $\mathsf A_{w}$ will remain the same as epoch $w-1$.
\end{proof}



\subsection{Disk-based Solution}\label{sect:disk_based_approach}

In this section, we have explained how our approach can be extended to the scenario of disk resident objects where the number of objects that require enrichment can not fit in memory. We have leveraged the algorithms proposed by the authors in \cite{Altowim:20F14:PAR:2732967.2732975} to develop the solution. 
In the following, we have explained our approach.

In the plan generation phase of our approach we add a \emph{block selection} step, where we determine which blocks of objects should be brought in memory from disk. We perform this step after the benefit estimation step of the triples. We associate a benefit metric with each block. This is calculated for a block $b_r$ in epoch $w$ as the summation of the benefit values of all the triples present in $\mathsf{TS_w}$ which contains objects from the block $b_r$. We denote this metric as \emph{Block Benefit} value.

We maintain two priority queues, one for the blocks which are present in memory and one for the blocks on the disk. The priority queue for the blocks present in memory is referred to as $PQ_{mem}$ and the priority queue for the disk resident blocks are referred to as $PQ_{disk}$. The blocks in both the priority queues are ordered based on their block benefit values. 

We create equal sized blocks based on the ESP values of the objects at the beginning of the query execution. We assume that the load time of each blocks are approximately same. We denote the load time of a block as $c_{l}$. Based on the available memory size, we determine how many blocks that can be loaded maximally in an epoch (denoted by $d_0$). The number of blocks loaded from disk to memory in epoch $w$ is denoted by $d_w$.

We generate a number of alternate plans where in each plan we flush a certain number of least beneficial blocks (derived from $PQ_{mem}$) from memory to disk and load equal number of highest beneficial blocks from disk to memory (derived from $ PQ_{disk}$). After the flushing and loading of the blocks, we generate a plan which consists of the triples from $\mathsf{TS_w}$ where the objects are present in the blocks residing in the memory. We vary the number of blocks from 0 to $d_0$ that can be flushed from memory to disk and generate $(d_0+1)$ number of alternate plans. For example, if $d_w=0$, then we do not flush any blocks from memory to disk and generate a plan with the triples containing objects present in memory from the previous epoch. If $d_w=1$, then we flush one block (block with least block benefit value) from memory to disk and load one block from disk to memory (block with highest block benefit value). This plan will contain the triples that can be executed in $(v_w - c_{l})$ amount of time. In this way, we generate $(d_0+1)$ alternate plans where each plan corresponds to flushing of $\{0, 1, \cdots, d_0\}$ number of least beneficial blocks from memory to disk and loading of the same number of highest beneficial blocks from disk to memory. 

We associate a benefit value with each of the generated plans and denote it as \emph{plan benefit} value. This value is calculated as the summation of the benefit values of all the triples present in the plan. The plan with the maximum plan benefit value is chosen in this step. Once the plan is chosen, the list of blocks that need to be flushed to disk are flushed first and the required blocks are loaded from the disk to memory.

We have experimentally compared the performance of our approach $w.r.t.$ the baselines with disk-resident objects.  
In Baseline 1 algorithm, we make the following changes from the approach presented in the main body of the paper. We first choose a tagging function and then execute it on all the objects after loading the required blocks to memory. We order all the blocks in disk based on their block benefit value. We load the top blocks based on their block benefit value to memory and then execute the tagging function on all the objects present in the blocks. We continue this process until all the required blocks are loaded in memory and the tagging function is executed on all the objects. 

In Baseline 2 algorithm, we make the following changes. We order all the blocks based on their block benefit value. We load the top blocks based on their block benefit value from disk to memory. We choose an object from the loaded blocks first and then execute all the tagging functions on the object. The objects were chosen based on their ESP values at the beginning of the query execution time, starting with the object with the highest ESP value.

The results are shown in Figure \ref{fig:VariationOfQualityDiskBased}. We can see that our approach performs significantly better than the baseline approaches. Our approach chooses the best blocks that should be loaded in memory at a particular epoch based on their block benefit value which results in execution of more beneficial blocks as compared to the baseline approaches. Furthermore, within each block we order the execution of triples according to their benefit value which results in the execution of more beneficial triples as compared to the baseline approaches.

\section{Data Structure Update} 
\label{appendix:data_structure}

In the plan execution phase of epoch $w$, we only update the object level data structures of the objects corresponding to the triples executed in that epoch. We efficiently update the state hash map, predicate probability hash map and the uncertainty list of these objects. Given a triple $(o_k,R^i_j,f^i_m)$ which was executed in epoch $w$, we update $PQ$ as follows: if the new ESP value of $o_k$ becomes higher than $\mathcal P^{\tau}_{w}$, then we update all the triples in $PQ$ containing $o_k$ using their updated benefit value. If the new ESP value becomes lower than $\mathcal P^{\tau}_{w}$, then we remove all the triples containing $o_k$ from $PQ$.

Given a triple $(o_k,R^i_j,f^i_m)$ which was not executed in epoch $w$, we update $PQ$ as follows: if the ESP value of $o_k$ is not changed $w.r.t.$ previous epoch $w-1$, then we do not update the benefit of the triples containing $o_k$ as the benefit values (i.e., $\nu_k(\mathcal P_k+\Delta \mathcal P_k)$ ) of those triples will remain the same $w.r.t.$ previous epoch. If the ESP value of $o_k$ is changed $w.r.t.$ previous epoch $w-1$, then we recompute the benefit of the triples containing $o_k$ and then add it to $PQ$.

\subsection{Online Learning of the parameters}\label{sect:online_param_learn}

In our experimental setup, we have learned the parameters cost, the quality of the tagging functions and the decision table with the help of a tagged validation dataset. Although this step is used as an offline step in our experiments, this step can be learned online as follows.

We assign a small portion of the first epoch for the learning of these parameters. We execute the tagging functions on a small set of objects. We store the execution time and the probability outputs of the tagging functions on each objects. We use these outputs to learn the cost of the tagging functions and the decision table. 

\vspace{0.1cm}
\noindent
\textbf{Learning the Cost parameter.} Given the execution costs of a tagging function on the objects, we measure the average execution time. We set the cost of the tagging function with this value of average execution time.

\vspace{0.1cm}
\noindent
\textbf{Learning the Decision Table.}
In Section \ref{triple generation}, we have explained the structure of the decision table. Given a range of uncertainty values $(a, b]$, we use the state, imprecise attribute, tag, and uncertainty values of the objects and learn the function that is expected to provide the highest uncertainty reduction among the other remaining functions. Note that, for learning this table we do not require the ground truth information of the data.


We use the probability outputs of the tagging functions for generating this table. In a given state, for each objects, we measure the uncertainty value. For an object, we determine which function among the remaining tagging functions reduces the uncertainty value of the object most at that state. We create a number of bins of possible uncertainty values (ranging between 0 and 1) and learn the next function which reduces the uncertainty of the object most within each bins. This value is stored in the fifth column (labeled as \textit{Next Function}). Once the next function is learned, given a bin (e.g., in the range of [0.5, 0.6)) we calculate the average amount of uncertainty reduced by the function. This is calculated by summing up the reduction of the uncertainty values of those objects and dividing it by the count of the objects.

\end{appendix}




%
%


\bibliographystyle{abbrv}

\bibliography{references}


\end{document}